\documentclass{article}

\usepackage{arxiv}

\usepackage[utf8]{inputenc} 
\usepackage[T1]{fontenc}    
\usepackage{hyperref}       
\usepackage{url}            
\usepackage{booktabs}       
\usepackage{amsfonts}       
\usepackage{nicefrac}       
\usepackage{microtype}      
\usepackage{lipsum}         
\usepackage{graphicx}
\usepackage{natbib}
\usepackage{doi}
\usepackage{hyperref}
\usepackage{float}
\usepackage{url}            
\usepackage{booktabs}       
\usepackage{amsfonts}       
\usepackage{amssymb}        
\usepackage{nicefrac}       
\usepackage{microtype}      
\usepackage{bm}				
\usepackage{mathrsfs}		
\usepackage{cite}			
\usepackage{graphicx}		
\usepackage{mathtools}		
\usepackage{algpseudocode}  	
\usepackage{algorithm}
\usepackage{caption}  

\usepackage{amsthm}			
\usepackage{enumitem}		
\usepackage{multirow}       
\usepackage{tabularx}	    
\usepackage{hhline}
\usepackage{arydshln}		
\usepackage{wrapfig}
\usepackage{lipsum}
\usepackage{diagbox}
\usepackage{soul}
\usepackage{amsthm}

\usepackage{algpseudocode}

\usepackage{amssymb}

\theoremstyle{plain} 
\newtheorem{proposition}{Proposition}

\newtheorem{theorem}{Theorem}
\newtheorem{lemma}{Lemma}

\newtheorem{assumption}{Assumption}
\newtheorem{remark}{Remark}

\title{Arbitrage-Free Multi-Maturity Risk-Neutral Marginals}

\date{}

\author{
{\large Hao Qin$^{\dagger}$, Ruozhong Yang$^{\dagger}$,  \href{https://orcid.org/0009-0009-4556-8664}{\includegraphics[scale=0.06]{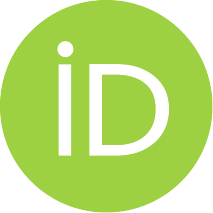}\hspace{1mm}Charlie Che$^{\dagger\ddagger}$}, and Liming Feng$^{\dagger}$\thanks{Corresponding author. Email: fenglm@illinois.edu.}}\\[1ex]
$^{\dagger}$Department of Industrial and Enterprise Systems Engineering, University of Illinois Urbana-Champaign, Illinois, United States\\
$^{\ddagger}$Quantitative Trading \& Research, JPMorgan Chase \& Co., New York, United States
}
\begin{document}
\maketitle

\begin{abstract}
Many quantitative finance methods and applications are formulated in terms of option-implied risk-neutral marginals rather than directly in terms of option prices. Representative examples include martingale optimal transport, Bass local-volatility calibration, scenario analysis, and option-implied tail-risk measurement. The desired risk-neutral marginals should define a genuine probability law on the entire support, reproduce the input arbitrage-free option prices exactly, be free of butterfly and calendar arbitrage, and admit efficient evaluation of the density, distribution function, and quantiles, as well as Monte Carlo sampling. Existing methods typically optimize only a subset of these properties, depending on their intended purpose. This leaves a gap between upstream arbitrage-free option prices and the readily usable risk-neutral marginals required by downstream applications. We propose an explicit construction of risk-neutral marginals from discrete arbitrage-free option prices. On the observed strike range, probability mass is assigned interval by interval to exactly reproduce the input option prices. Outside the observed range, closed-form power-law tails complete the distribution by satisfying price and slope boundary conditions and allocating the remaining probability mass. Butterfly- and calendar-arbitrage-freeness are guaranteed by construction. The construction is feasible by design and computationally efficient. The resulting marginal laws admit closed-form densities, distribution functions, quantiles, and efficient Monte Carlo sampling. Numerical experiments on synthetic SSVI data and S\&P~500 market data demonstrate that the proposed construction efficiently and robustly produces marginals satisfying all of these properties in practice.
\end{abstract}
\vspace{1em} 

\section{Introduction}

A growing class of methods in quantitative finance is framed directly in terms of the option-implied risk-neutral marginal laws of the underlying asset. In these approaches, the primitive input is a prescribed family of marginal laws, rather than the option quotes from which those laws are recovered.

The main examples fall into three groups. The first group consists of methods induced by martingale optimal transport (MOT). In the static MOT problem, the prescribed marginals define the feasible set of martingale measures over which model-independent price bounds are computed \citep{beiglboeck2013model, henry2017model}. Dynamic MOT-induced methods use the same marginal inputs constructively.  Martingale Benamou--Brenier (MBB) problems build martingales connecting prescribed marginals \citep{backhoff2020mbb}, while Bass local-volatility calibration uses them in a fixed-point procedure to construct calibrated dynamics \citep{conzehenry2021bass, acciaio2024calibration}. The second group is scenario analysis, where probabilities of stress regions are read directly from an option-implied marginal \citep{gemmill2000useful}. The third group is option-implied tail-risk measurement, where tail payoffs or loss functionals are integrated against the marginal law \citep{bollerslev2011tails, andersen2017short}. In each case, the working object is a marginal law, and the downstream output is only as good as that input. 

This observation also clarifies where marginal construction enters the usual market pipeline. A common workflow starts from raw option quotes, builds a de-arbitraged implied volatility surface, converts that surface back into call prices, and then differentiates the price curve to obtain a density. For the applications above, the downstream object is the marginal law, not the volatility parametrization. It is an intermediate representation used to organize and smooth option information. Once the upstream step has produced an arbitrage-free call-price curve, or across maturities an arbitrage-free call-price grid, the information needed for marginal construction is already present in the prices.

Existing methods differ in the upstream representation through which option observables are organized before marginal laws are obtained. In one line of work, a parametric or structured surface is built first, such as an SVI-type parametrization \citep{gatheral2014arbitrage}. Once the surface is fixed, the associated marginal is obtained from the strike curvature of the call-price curve. This gives a continuous object, but feasibility and exact price matching are limited by the chosen surface family. A second line of work constructs arbitrage-free option-price, implied-volatility, or state-price-density (SPD) representations from discrete quotes.  Recent examples include \citep{sanos2025, deschatres2026CVI, lu2021sieve}. These representations are useful upstream objects, but they are not yet the marginal laws required downstream. Turning them into well-defined probability laws still requires a numerical curvature step, tail completion, and a normalization or adjustment step.

This leaves a gap between upstream surface construction and downstream marginal-based methods. A fitted surface can provide arbitrage-free prices, but a downstream solver requires a marginal family with its own set of properties: each slice must be a probability law, each observed quote must be repriced exactly, the tails must carry the correct residual mass, and the family must be ordered across maturities with no arbitrage. These are distribution-level requirements on marginals, and are not automatically delivered by upstream representations.

We supply this missing layer through an explicit construction that is anchored at the level of call-price slopes and realized through non-negative curvature. For a single maturity, we first estimate admissible slopes at the observed strikes. The endpoint slope conditions encode total probability mass, while the slope increment across each strike interval determines the marginal mass assigned to that interval. Conditional on these slopes, we then construct the call-price curvature $C''$ as a non-negative piecewise-constant function. Integrating this curvature recovers the prescribed slopes and then the call-price curve, so each observed quote is matched exactly.

The construction also completes the unquoted regions of the marginal law. Beyond the observed strike range, the boundary slopes determine the residual probability mass that must be placed in the tails. We attach closed-form power-law tails whose parameters match the boundary price and slope, allocate the residual mass, preserve monotone decay, and keep the required moments finite. This tail layer is flexible and can be adjusted to the practitioner's preference, because it couples to the interior only through boundary prices and slopes. Its power-law form also gives wing behavior consistent with Lee's moment formula.

Across maturities, the main additional constraint is calendar no-arbitrage. We enforce it sequentially. After one maturity is constructed, its call-price curve provides the lower bound for the next maturity under forward moneyness alignment. The next maturity is then constructed interval by interval. If the single-maturity construction already stays above this lower bound, we keep it unchanged. If not, we modify only the problematic interval by using a three-piece curvature construction. This adds the flexibility needed to restore calendar no-arbitrage while keeping the change local. Tail consistency across maturities is handled by a small global optimization over boundary slopes, with only maturity-level variables rather than a full strike--maturity grid. Overall, the multi-maturity construction mostly remains local and sequential, rather than becoming a large-scale global surface optimization.

Following this construction, our output marginal family delivers five main properties.

\begin{enumerate}
\item \textbf{Valid probability laws.} Each maturity defines a non-negative unit-mass density on the full price domain.

\item \textbf{Arbitrage freedom across strikes and maturities.} Non-negative curvature rules out butterfly arbitrage. The sequential multi-maturity construction enforces calendar no-arbitrage pointwise across maturities.

\item \textbf{Exact quote consistency.} The constructed call-price curve passes through every observed quote. The method does not alter, refit, or smooth the input prices.

\item \textbf{Robust feasibility.} For arbitrage-free inputs satisfying a mild input calendar condition, we prove that the required construction exists. The method therefore avoids numerical feasibility failure and post-hoc repair.

\item \textbf{Practical efficiency.} The body construction is explicit and decomposes across strike intervals, so each maturity can be built by solving small independent subproblems. Across maturities, the procedure is sequential. The only global step is a small optimization over tail boundary slopes, with variables indexed by maturity rather than by the full strike--maturity grid.

\end{enumerate}
\label{sec:introduction-desiderata-anchor}

Because the construction is closed-form on every piece, the resulting marginals are directly usable by the applications above: the density, the cumulative distribution, and the quantile function are available by direct evaluation, and Monte Carlo samples are straightforward to draw.

The remainder of the paper is organized as follows.
Section~\ref{sec:litreview} reviews the related literature.
Section~\ref{sec:prelim} sets up the problem and our construction strategy.  Sections~\ref{sec:single} and~\ref{sec:multi} develop the construction under a single maturity and then across maturities.  Section~\ref{sec:numerical} reports the numerical
experiments on a synthetic SSVI surface and on S\&P~500 quotes.
Section~\ref{sec:conclusion} concludes.  Proofs are collected in the Appendix.	

\section{Literature Review}\label{sec:litreview}

We organize the literature by the role each body of work plays relative to our construction.  Upstream methods turn raw option quotes into arbitrage-free prices or surfaces, which serve as the input to our procedure (Section~\ref{sec:lit-vol-surface}).  Downstream methods take risk-neutral marginals as primitive inputs, and their input requirements determine the properties our construction must deliver (Section~\ref{sec:lit-mot}).  A third line of work is closest to ours in object but differs in target: it estimates or parametrizes risk-neutral densities from option prices (Section~\ref{sec:lit-density}).  Our construction sits between the first two and is different from the third.  It is a marginal-construction layer that turns arbitrage-free prices into marginal laws suitable as direct inputs to downstream methods.

\subsection{Upstream Arbitrage-Free Price and Surface Construction}\label{sec:lit-vol-surface}

A large literature builds arbitrage-free option-price or implied-volatility representations from market quotes, and these methods are organized by the object they fit.  Parametric surface models specify a volatility surface family and impose sufficient conditions for static arbitrage-freeness within it, with SSVI as the standard example \citep{gatheral2014arbitrage}.  Interpolation-based methods build an arbitrage-free call-price or local-volatility surface directly from observed prices \citep{kahale2004, andreasen2011volatility, lefloch2023qlvg}.  Constrained formulations enforce no-arbitrage through optimization. Convex volatility interpolation \citep{deschatres2026CVI} fits an implied-volatility surface using a variance-space parametrization and quadratic programming, while SANOS \citep{sanos2025} represents call prices as convex combinations of Black--Scholes payoffs calibrated by linear programming.  Neural-operator smoothers with no-arbitrage regularization have also been proposed for implied-volatility surfaces \citep{wiedemann2025ods}. These methods provide upstream arbitrage-free price or volatility inputs. In our framework, they can serve as tools for producing the price inputs on which the marginal construction starts. They do not, by themselves, complete the distribution-level object required by downstream marginal-based methods. 

\subsection{Downstream Marginal-Consuming Methods}\label{sec:lit-mot}

Several methods take the marginal family itself as the primitive input.  Martingale optimal transport is the clearest case.  It computes model-free price bounds over all martingale measures consistent with prescribed marginals \citep{hobson2011, beiglboeck2013model, henry2017model}.  Dynamic constructions in the same family, including martingale Benamou--Brenier formulations and numerical transport schemes, also start from fixed marginal laws \citep{backhoff2020mbb, guo2020}.  Bass local-volatility calibration is a closely related case, in which the calibrated dynamics solve a fixed-point system built around the prescribed marginals \citep{conzehenry2021bass, acciaio2024calibration,hao-basslv}.  These applications impose strong requirements: the marginals must reprice the quoted instruments exactly, and the stability of the iteration depends on their arbitrage-freeness and tail behavior.

Other applications use the same object less stringently.  Implied-distribution studies use option-implied marginals to infer market expectations around stress or event scenarios \citep{gemmill2000useful}.  Option-implied tail-risk studies use short-maturity option information to measure jump tails, rare-event compensation, and short-horizon downside risk \citep{bollerslev2011tails, andersen2017short}. These examples define the downstream role of our output.  The required strength of exact repricing, tail control, and cross-maturity consistency varies by application, but these distribution-level properties determine whether the marginals are usable as downstream inputs.

\subsection{Risk-Neutral Density Estimation from Option Prices}\label{sec:lit-density}

The literature closest to ours estimates or parametrizes a risk-neutral density from option prices.  Its target differs from ours: it recovers a density or a smooth state-price representation, whereas we construct a family of admissible marginals for downstream solvers.  We group it by how the density is obtained.

One route follows the Breeden--Litzenberger identity \citep{breeden1978prices}: a smooth price or implied-volatility curve is fitted first, and the density is then obtained by differentiation.  Examples include quadratic strike or delta-space fits \citep{shimko1993bounds, malz1997}, SVI-type specifications \citep{gatheral2014arbitrage}, and nonparametric or shape-constrained smoothing methods \citep{ait1998nonparametric, bliss2002testing, fengler2009arbitrage}.  Another route specifies or estimates the density directly and calibrates option prices by integration, using mixtures, expansion methods, or sieve estimators \citep{melick1997, gemmill2000useful, jarrow1982approximate, xiu2014hermite, lu2021sieve}.  Tail extrapolation is a recurring component of these density-estimation procedures, with arbitrage-free extrapolation and moment-based wing restrictions studied in \citet{brunner2003arbitrage}, \citet{lee2004moment}, \citet{benaim2009regvar}, and \citet{jacquierraval2023logmoment}.  These methods provide important tools for density or state-price estimation.  Our focus is different: we construct marginal laws directly from arbitrage-free prices, with exact repricing and local feasibility built into the construction.

Table~\ref{tab:comparison} summarizes how representative method families align with the marginal-construction target considered here. The comparison is target-specific: it does not evaluate these methods for the purposes for which they were designed, but indicates whether they directly provide the admissible marginal inputs required by downstream marginal-based applications.

\begin{table}[ht]
\centering
\small
\begin{tabular}{lcccc}
    \toprule
    Method family 
    & Valid marginal 
    & Exact fit 
    & Calendar consistency 
    & Feasible construction \\
    \midrule
    Parametric arbitrage-free surfaces
    & $\checkmark$ & $\triangle$ & $\checkmark$ & $\triangle$ \\
    Arbitrage-free price/surface interpolation
    & $\triangle$  & $\checkmark$ & $\triangle$  & $\triangle$ \\
    Differentiation-based density extraction
    & $\triangle$  & $\triangle$  & $\times$     & $\triangle$ \\
    Direct density parametrizations
    & $\checkmark$ & $\triangle$  & $\triangle$     & $\triangle$ \\
    \textbf{Our paper}
    & $\checkmark$ & $\checkmark$ & $\checkmark$ & $\checkmark$ \\
    \bottomrule
\end{tabular}
\caption{Representative method families assessed for the marginal-construction target studied in this paper.  The comparison is not a judgment of these methods in their intended settings; it only indicates whether they directly provide the marginal inputs required by downstream marginal-based applications.  $\checkmark$: delivered by construction; $\triangle$: method-dependent or requiring additional post-processing; $\times$: not generally guaranteed.}

\label{tab:comparison}
\end{table}

\section{Preliminaries and Construction Strategy}\label{sec:prelim}

The construction starts from the derivative structure of the call curve. The option-price slope $C'$ records discounted tail probabilities, and the curvature $C''$ records discounted probability density. Thus, once the endpoint slopes are fixed, assigning non-negative curvature across strike intervals is the same as assigning non-negative marginal mass. This is the mathematical reason why an arbitrage-free call-price grid is the natural input for marginal construction.

At each maturity, the object we construct is a risk-neutral marginal law, represented throughout by its call-price curve and the implied density $f=e^{rT}C''$. The construction proceeds in two stages: a single maturity (Section~\ref{sec:single}), then the coupling across maturities (Section~\ref{sec:multi}). Throughout, the market is arbitrage-free under the risk-neutral measure $\mathbb{Q}$, the risk-free rate $r$ and the continuous dividend yield $q$ are constant, and $S_T>0$ almost surely. The input is a set of discrete arbitrage-free European call prices; the put-based case is analogous. Any upstream method of Section~\ref{sec:lit-vol-surface}, or raw quotes validated directly against no-arbitrage inequalities, can supply this input.

Before constructing the intervals and tails, we put the market inputs into a common normalized form and record the boundary values implied by no-arbitrage. This initialization makes the construction independent of the scale of the underlying, while retaining the effects of interest rates and dividends through the boundary conditions.

We write $k_1<k_2<\cdots<k_n$ for the observed strikes at a given maturity, $C_{k_i}$ for the corresponding call prices, and $T$ for the time to maturity; in the multi-maturity case the maturities are $T_j$, $1\leq j\leq m$. The point $k_0=0$ is introduced as the left boundary node rather than an additional market quote; it fixes the left endpoint of the call-price curve and the initial slope from which probability mass is accumulated.

\paragraph{Dividends, normalization, and boundary values.}
Under $\mathbb{Q}$, the expected stock price at maturity is
\[
\mathbb{E}^{\mathbb{Q}}\left[S_T\right] = S_0\, e^{(r-q) T},
\]
and the call price at strike $K$ is
\[
C(K) = e^{-r T}\, \mathbb{E}^{\mathbb{Q}}\left[\left(S_T-K\right)^{+}\right] = e^{-r T} \int_K^{\infty}\left(x-K\right) f\left(x\right)\, dx,
\]
where $f$ is the risk-neutral density of $S_T$.  Differentiating twice in $K$ gives the Breeden--Litzenberger identity
$
\frac{d^2C}{dK^2}(K) = e^{-rT} f(K).$
The dividend yield does not alter this identity; it enters through the boundary values of $C$ and $C'$.  Setting $K = 0$ in the two displays above and using $S_T > 0$ gives
\[
C(0) = e^{-rT}\,\mathbb{E}^{\mathbb{Q}}[S_T] = S_0\, e^{-qT}, \qquad C'(0) = -e^{-rT}\,\mathbb{Q}(S_T > 0) = -e^{-rT}.
\]

We normalize all prices and strikes by the spot.  Writing $\tilde{S}_T := S_T/S_0$, $\tilde{k} := K/S_0$, and $\tilde{C}(\tilde{k}) := C(S_0 \tilde{k})/S_0$, the Breeden--Litzenberger identity keeps its form, with $\tilde{f}(\tilde{k}) := S_0\, f(S_0 \tilde{k})$ the density of $\tilde{S}_T$, and the boundary values become
\[
\tilde{C}(0) = e^{-qT}, \qquad \tilde{C}'(0) = -e^{-rT}.
\]
From here on we work in normalized units and drop the tildes.  All outputs convert back to the original units through
\[
C(K) = S_0\, \tilde{C}(K/S_0), \qquad f(K) = \tilde{f}(K/S_0)/S_0,
\]
so the normalization is invertible and loses no information.

\paragraph{Reduction to boundary conditions on the slope.}
Our construction should guarantee a well-defined density $f$ on $[0,\infty)$: a non-negative function with $\int_0^{\infty} f(k)\,dk = 1$ that is consistent with the observed call prices via $C''(k) = e^{-rT} f(k)$.  Fitting $C(k)$ first and differentiating twice is numerically fragile, and the result need not be non-negative or have unit mass. We instead anchor the construction at the level of the call-price slope.  Given valid slopes at the strike nodes, the curvature on each interval is chosen as a non-negative piecewise function that integrates to the required slope increment, and the call-price curve is then recovered by integration. Continuity of $C'$ and convexity of $C$ are then automatic.

Under this strategy the entire requirement that $f$ is a valid probability density reduces to a pair of boundary conditions on $C'(\cdot)$.  By the fundamental theorem of calculus,
\[
\int_{0}^{\infty} f(k)\, dk
= e^{rT}\!\int_{0}^{\infty} C''(k)\, dk
= e^{rT}\!\left(\lim_{k\to\infty} C'(k) - \lim_{k\to 0^+} C'(k)\right),
\]
so imposing
\[
\lim_{k\to 0^+} C'(k) = -e^{-rT}, \qquad \lim_{k\to\infty} C'(k) = 0
\]
is exactly equivalent to $\int_0^{\infty} f(k)\,dk = 1$.  The construction problem is thus reduced to choosing non-negative interval curvatures whose integrals match the prescribed slope increments and whose recovered call-price curve matches the observed prices.

\section{Construction under a Single Maturity}\label{sec:single}

A single maturity carries $n-1$ interior intervals and two tails.  In principle one could fit them all at once, with a global parameter vector of size $O(n)$.  We instead proceed interval by interval and use a two-piece curvature construction for two reasons.  First, locality: each interval matches only its own endpoint prices and slopes, so the subproblems decouple and solve independently, in parallel.  This suits single-maturity data, where the quotes at one expiry already pin down the local mass.  Second, solvability: the endpoint slopes already fix the total increase in slope across the interval, while the endpoint price is a further constraint.  A single constant curvature construction lacks the freedom to meet both; a two-piece curvature, with a free split point, supplies exactly what is needed to match the price while keeping curvature non-negative.  This idea extends naturally to the multi-maturity case (Section~\ref{sec:three-piece}): one further constraint, the calendar floor, calls for one further piece.

Fix a maturity $T$.  The data are the observed strikes $k_1 < \cdots < k_n$ with call prices $C_{k_1}, \ldots, C_{k_n}$, supplemented at the boundary node $k_0 = 0$ by the exact values $C(0) = e^{-qT}$ and $C'(0) = -e^{-rT}$ derived above.  The construction proceeds in three steps: estimating the slope $C'$ at every observed strike (Section~\ref{sec:Ai-Lagrange}), building $C''$ on each interval between adjacent observed strikes (Sections~\ref{sec:second-order-construction} and~\ref{sec:two-piece-opt}), and completing the two tail regions $[0, k_1]$ and $[k_n, \infty)$ (Sections~\ref{sec:powerlaw_tail} and~\ref{sec:left_tail_pow}).  Given the reduction above, we impose $C'(0)=-e^{-rT}$ and $C'(\infty)=0$ as boundary conditions; together with $C''(k)\ge 0$, they make $f = e^{rT} C''$ a valid probability density.  For adjacent strikes $k_i, k_{i+1}$ ($1\leq i\leq n-1$), we set $A_i := C'(k_{i+1}) - C'(k_i)$, which is non-negative by convexity.

\subsection{Estimating $C'(k_i)$ and $A_i$ at Observed Strikes}\label{sec:Ai-Lagrange}

The slope $C'(k_i)$ is not observed directly, and finitely many call prices do not uniquely determine it. We therefore select one estimate $\widehat{C'}(k_i)$ at each observed strike and use these estimated slopes as boundary data for the interval curvature construction. This subsection describes how these estimates are obtained.

We use $C'(k)$ for the true, unobserved slope and $\widehat{C'}(k_i)$ for its estimate at an observed strike $k_i$.  The estimated slopes define the interval slope increments $A_i := \widehat{C'}(k_{i+1}) - \widehat{C'}(k_i)$ for $i = 1, \ldots, n-1$.  In later interval-construction sections, these estimates are used as fixed boundary data and written $C'_{k_i} := \widehat{C'}(k_i)$.

We estimate $C'(k_i)$ by applying Lagrange finite-difference formulas to the observed call prices. Recall that the Lagrange basis polynomial through nodes $x_0, \ldots, x_k$ is $\ell_j(\xi) = \prod_{i=0,\, i \neq j}^k \frac{\xi - x_i}{x_j - x_i}$, so the polynomial through the prices is $p(\xi) = \sum_j C(x_j)\,\ell_j(\xi)$, and differentiating it at a node $x_m$ gives the slope estimate $p'(x_m) = \sum_j C(x_j)\,\ell_j'(x_m)$.

Interior strikes pose no difficulty: for $k_i$ with $i = 3, \ldots, n-2$, we use the five-point estimation on the window $\{k_{i-2}, k_{i-1}, k_i, k_{i+1}, k_{i+2}\}$.  The two boundaries require one-sided treatment, and the two sides are handled differently.  We treat them in turn.

At the right boundary, no observed strikes lie beyond $k_n$, so the estimates are one-sided: a four-point estimation on $\{k_{n-3}, k_{n-2}, k_{n-1}, k_n\}$ for $\widehat{C'}(k_{n-1})$, and a backward three-point estimation on $\{k_{n-2}, k_{n-1}, k_n\}$ for $\widehat{C'}(k_n)$.  Every node in these windows is an observed strike.  No information about $C$ is available beyond $k_n$, and none is needed.

The left boundary differs in one essential respect: exact information is available to the left of $k_1$, namely the boundary values $C(0) = e^{-qT}$ and $C'(0) = -e^{-rT}$.  How to use it depends on how far the boundary node $0$ lies from $k_1$.  A purely forward three-point estimate on the observed strikes $\{k_1, k_2, k_3\}$ ignores the boundary entirely.  In the typical case the gap $k_1 - 0 = k_1$ is one to two orders of magnitude larger than the first body spacing $k_2 - k_1$, so $0$ is distant, and a quadratic fit through the asymmetric window $\{0, k_1, k_2\}$ that includes it imposes artificial curvature on the nearly linear region $[0, k_1]$.

We therefore route the boundary information through the left-tail extrapolation of Section~\ref{sec:left_tail_pow} instead of through a single distant node.  The bootstrap carries the boundary data to a node close to $k_1$.  A one-step left tail, built from the boundary data and the observed prices, already encodes the boundary values $C(0)$ and $C'(0)$ that the left-tail construction matches; it is used only to estimate the slope, not as the final left tail.  We evaluate it at an auxiliary node $k_{\rm aux} \in (0, k_1)$ placed one body spacing below $k_1$, and re-estimate the slope on the resulting equally spaced window $\{k_{\rm aux}, k_1, k_2\}$.  Placing $k_{\rm aux} = k_1 - h_1$ inside $(0, k_1)$ requires the body spacing $h_1 = k_2 - k_1$ to be smaller than the boundary gap $k_1$; when instead $h_1 \ge k_1$, the boundary node $0$ is close enough to include directly, and Algorithm~\ref{alg:c-k-boundary} estimates $C'(k_1)$ from the four-point window $\{0, k_1, k_2, k_3\}$.  This keeps the boundary information while avoiding the asymmetry bias of the previous paragraph.  Algorithm~\ref{alg:c-k-boundary} makes this precise.  The slope at the next strike $k_2$ then uses the four-point estimation on $\{k_{\rm aux}, k_1, k_2, k_3\}$, as in the four-point treatment of $k_{n-1}$ at the right boundary.

\begin{algorithm}[ht]
    \caption{Boundary slope estimate $\widehat{C'}(k_1)$ at the smallest observed strike}\label{alg:c-k-boundary}
    \begin{algorithmic}[1]
        \Require Observed strikes $k_1 < k_2 < k_3$ with spacings $h_1 := k_2 - k_1$ and $h_2 := k_3 - k_2$, call prices $C_{k_1}, C_{k_2}, C_{k_3}$, market data $r, q, T$, and the boundary values $C(0) = e^{-qT}$ and $C'(0) = -e^{-rT}$.
        \If{$h_1 \geq k_1$}  \Comment{the boundary node $0$ is within one body spacing of $k_1$}
            \State \Return the four-point Lagrange-derivative estimate of $C'(k_1)$ on $\{0, k_1, k_2, k_3\}$.
        \EndIf
        \State Bootstrap.  Compute the forward three-point Lagrange-derivative estimate on $\{k_1, k_2, k_3\}$,
        \begin{equation}\label{eq:c-k-bootstrap}
            \widehat{C'}(k_1)^{(0)} := -\frac{2h_1 + h_2}{h_1\,(h_1 + h_2)}\,C_{k_1} + \frac{h_1 + h_2}{h_1\,h_2}\,C_{k_2} - \frac{h_1}{h_2\,(h_1 + h_2)}\,C_{k_3}.
        \end{equation}
        \State Build a temporary left segment.  To create a price point just left of $k_1$, build a temporary segment $C_{\rm tail}$ on $[0, k_1]$ from the boundary data $C(0), C'(0)$ together with the values $C(k_1) = C_{k_1}$ and $\widehat{C'}(k_1)^{(0)}$ at $k_1$.  This is the power-law tail of Proposition~\ref{prop:left_pow}, given in closed form by~\eqref{eq:left_pow_compact}.
        \State Synthesize the auxiliary node.  Set
        \begin{equation}\label{eq:k-aux-def}
            k_{\rm aux} := k_1 - h_1 \in (0, k_1),
        \end{equation}
        and write $C_{k_{\rm aux}} := C_{\rm tail}(k_{\rm aux})$.
        \State Recompute.  Compute the centered three-point Lagrange-derivative estimate on the equally spaced window $\{k_{\rm aux}, k_1, k_2\}$,
        \begin{equation}\label{eq:c-k-refined}
            \widehat{C'}(k_1) := \frac{C_{k_2} - C_{k_{\rm aux}}}{2\,h_1}.
        \end{equation}
        \State \Return $\widehat{C'}(k_1)$.
    \end{algorithmic}
\end{algorithm}

\paragraph{Enforcing the no-arbitrage bracket.}
At each observed strike, butterfly-freedom of the input data demands
\begin{equation}\label{eq:no-arb-bracket}
    \frac{C_{k_i}-C_{k_{i-1}}}{k_i-k_{i-1}} \leq \widehat{C'}(k_i) \leq \frac{C_{k_{i+1}}-C_{k_i}}{k_{i+1}-k_{i}},
\end{equation}
with the natural one-sided versions at the boundary nodes.  A centered three-point Lagrange-derivative estimate respects this bracket automatically.  Higher-order formulas, and any estimate evaluated at an endpoint of its window such as the bootstrap value~\eqref{eq:c-k-bootstrap} at $k_1$, can fall outside it when the local convexity is sharp or the grid is strongly non-uniform.

We enforce the no-arbitrage slope constraint by projecting each preliminary estimate onto the interval implied by the neighboring chord slopes.  This correction is needed only when a higher-order or one-sided estimate falls outside the interval.  At $k_1$, the estimate from Algorithm~\ref{alg:c-k-boundary} is projected in the same way; in the tail-anchored case ($h_1 < k_1$), if the refined estimate~\eqref{eq:c-k-refined} violates the one-sided constraint, we use the bootstrap estimate $\widehat{C'}(k_1)^{(0)}$ instead and project it.  Therefore, all final slope estimates are consistent with the discrete butterfly-free input prices.

\subsection{Second-Order Derivative Construction for Main Observations}\label{sec:second-order-construction}
Given the slope estimation, we now build the curvature $C''$ on each interval between adjacent observed strikes.  Butterfly freedom requires $C''(k)\geq 0$ for all $k\geq 0$.  Without further information about the market, the simplest non-negative shape on $[k_i, k_{i+1}]$ with enough freedom to match the interval data is piecewise constant with one interior breakpoint:
\begin{equation}\label{eq:two-piece-Cpp}
C^{\prime \prime}(k)= \begin{cases}h_i, & \text { for } k \in\left[k_i,\, k_i+\beta_i \Delta k_i\right], \\[3pt] \dfrac{A_i-h_i \beta_i \Delta k_i}{\left(1-\beta_i\right) \Delta k_i}, & \text { for } k \in\left(k_i+\beta_i \Delta k_i,\, k_{i+1}\right],\end{cases}
\end{equation}
where $\beta_{i} \in[0,1]$, $\Delta k_i := k_{i+1}-k_i$, and $h_i\geq 0$.  The interval data are the endpoint slopes $C'(k_i)=C'_{k_i}$ and $C'(k_{i+1})=C'_{k_{i+1}}$, summarized by the slope increment $A_i=C'_{k_{i+1}}-C'_{k_i}$, together with the endpoint prices $C(k_i)=C_{k_i}$ and $C(k_{i+1})=C_{k_{i+1}}$; the right endpoint price is imposed as a constraint in the next subsection.  Integrating~\eqref{eq:two-piece-Cpp} once from $C'(k_i) = C'_{k_i}$ gives the first derivative,
$$
C^{\prime}(k)= \begin{cases}C_{k_i}^{\prime}+h_i\left(k-k_i\right), & \text { for } k \in\left[k_i, k_i+\beta_i \Delta k_i\right], \\ \left(C_{k_i}^{\prime}+h_i \beta_i \Delta k_i\right)+\frac{A_i-h_i \beta_i \Delta k_i}{\left(1-\beta_i\right) \Delta k_i}\left(k-k_i-\beta_i \Delta k_i\right), & \text { for } k \in\left(k_i+\beta_i \Delta k_i, k_{i+1}\right],\end{cases}
$$
and integrating once more from $C(k_i) = C_{k_i}$ gives the price curve,
$$
C(k)= \begin{cases}C_{k_i}+C_{k_i}^{\prime}\left(k-k_i\right)+\frac{1}{2} h_i\left(k-k_i\right)^2, & \text { for } k \in\left[k_i, k_i+\beta_i \Delta k_i\right], \\ C_{k_i}+C_{k_i}^{\prime}\left(k-k_i\right)+\frac{1}{2} h_i\left(\beta_i \Delta k_i\right)^2+h_i \beta_i \Delta k_i\left(k-k_i-\beta_i \Delta k_i\right) & \\ +\frac{1}{2}\left(\frac{A_i-h_i \beta_i \Delta k_i}{\left(1-\beta_i\right) \Delta k_i}\right)\left(k-k_i-\beta_i \Delta k_i\right)^2, & \text { for } k \in\left(k_i+\beta_i \Delta k_i, k_{i+1}\right].\end{cases}
$$

\begin{remark}[Curvature design]\label{rmk:step-density}
    As discussed above, the construction is anchored at the level of call-price slopes.  Conditional on the estimated endpoint slopes, we choose the most direct and computationally efficient curvature representation that can match the interval data.  Consequently, the implied body density $f = e^{rT} C''$ is a step function on $[k_1, k_n]$.  This is aligned with the marginal-construction target of the paper.  The downstream applications considered in Section~\ref{sec:lit-mot} take distribution-level inputs that do not require the density itself to be continuously differentiable.  However, density-level shape regularity can still be encouraged as a secondary criterion through the objective in the per-interval optimization below, and through the optional density-anchored tail upgrade in Section~\ref{sec:density_anchored_tail}.
\end{remark}

\subsection{Optimization Formulation}\label{sec:two-piece-opt}

By imposing the boundary condition $C(k_{i+1})=C_{k_{i+1}}$, we obtain the price-matching constraint:
\begin{equation}\label{eq:price-match}
    h\left(\beta_i, h_i\right)=C_{k_i}-C_{k_{i+1}}+C_{k_i}^{\prime} \Delta k_i+\frac{1}{2} h_i \beta_i (\Delta k_i)^2+\frac{A_i \Delta k_i}{2}\left(1-\beta_i\right)=0.
\end{equation}

Denote the second-piece curvature by $c_i = \frac{A_i - h_i \beta_i \Delta k_i}{(1-\beta_i)\Delta k_i}$. As a secondary shape criterion, we minimize the curvature jump between the two pieces while enforcing non-negative curvature on each piece.  The choice of objective $F$ is left to the practitioner (any other smoothness or shape preference can be substituted) and does not affect the existence of a feasible solution: the feasible set is determined by the price-matching equality together with $h_i\ge 0$, $c_i\ge 0$, and is non-empty whenever the input data are arbitrage-free (Theorem~\ref{thm:two-piece} below).  The complete formulation for each interval $[k_i, k_{i+1}]$ is:
\begin{equation}\label{eq:two-piece-opt}
    \begin{array}{ll}
        \displaystyle\min_{\beta_i,\, h_i} & F(\beta_i, h_i) = (c_i - h_i)^2 \\[6pt]
        \text{s.t.} & h(\beta_i, h_i) = 0, \\[3pt]
        & h_i \geq 0, \\[3pt]
        & c_i \geq 0, \\[3pt]
        & \beta_i \in (0, 1].
    \end{array}
\end{equation}

The constraints $h_i \geq 0$ and $c_i \geq 0$ guarantee $C''(k) \geq 0$ on the entire interval, which is equivalent to the absence of butterfly arbitrage. The objective $F = (c_i - h_i)^2$ penalizes the curvature discontinuity at the internal knot $k_i + \beta_i \Delta k_i$.

\begin{theorem}[Two-piece interval construction]\label{thm:two-piece}
    Fix an interval $[k_i, k_{i+1}]$, $1 \leq i \leq n-1$, whose boundary data $C_{k_i}, C'_{k_i}, C_{k_{i+1}}, C'_{k_{i+1}}$ are free of butterfly arbitrage.  Then:
    \begin{enumerate}[label=\textup{(\alph*)}]
        \item the feasible set of~\eqref{eq:two-piece-opt} is non-empty; 
        \item problem~\eqref{eq:two-piece-opt} admits an explicit closed-form solution, computable in $O(1)$ arithmetic operations with no numerical solver.
    \end{enumerate}
\end{theorem}

Proof can be found in Appendix~\ref{sec:app_lemma3}.

\begin{remark}[Computational complexity]\label{rmk:O1}
    Each interval is solved in $O(1)$ operations by the closed-form solution of Theorem~\ref{thm:two-piece}(b), so the body construction over all $n-1$ intervals costs $O(n)$, with no iterative solver, line search, or convergence criterion.  This is the principal computational advantage of the single-maturity construction.
\end{remark}

\subsection{Right-Tail Estimation}\label{sec:powerlaw_tail}

The body construction ends at the last quote $k_n$, but the marginal extends beyond it.  On $[k_n,\infty)$ we complete the law with a tail density $f$ satisfying two conditions: matching the boundary data $C_{k_n}$ and $C'(k_n)$, and preserving non-negativity with monotone decay.  The boundary slope fixes the probability mass beyond $k_n$, while the boundary call price fixes the discounted average excess payoff carried by that mass.  Thus the tail is not unconstrained: the last observed price and slope determine both how much mass remains in the right tail and how far, on average, that mass must extend beyond the last strike.

We use a power-law family as a direct way to encode this information.  Its scale matches the residual mass, while its index controls the rate at which the remaining mass decays away from $k_n$.  This gives a closed-form tail whose mass, boundary price, and boundary slope can be matched exactly.  The same index also determines the tail moment structure, making the construction compatible with Lee's moment formula, which links implied-volatility wing behavior to the moments of the underlying distribution \citep{lee2004moment}.  The resulting wing behavior is therefore inferred from the boundary data rather than imposed as a fixed tail-heaviness assumption.

Define $\lambda = -C'(k_n)\, e^{rT} \in (0,1]$, the tail probability mass beyond $k_n$.  The power-law tail density on $[k_n, \infty)$ is
\begin{align}\label{eq:powerlaw_simple}
    f(x) = \lambda \cdot \frac{\alpha - 1}{k_n} \left(\frac{k_n}{x}\right)^{\alpha}, \quad x \geq k_n,
\end{align}
where the tail index $\alpha > 2$ is uniquely determined by the price-matching condition:
\begin{align}\label{eq:alpha_match}
    \alpha = 2 + \frac{\lambda\, k_n \, e^{-rT}}{C_{k_n}}.
\end{align}

\begin{proposition}[Power-law tail properties]\label{prop:powerlaw}
    The density~\eqref{eq:powerlaw_simple} with $\alpha$ given by~\eqref{eq:alpha_match} satisfies:
    \begin{enumerate}
        \item \textbf{Price matching:} $e^{-rT}\int_{k_n}^{\infty}(x - k_n)f(x)\,dx = C_{k_n}$.
        \item \textbf{Derivative matching:} $-e^{-rT}\int_{k_n}^{\infty} f(x)\,dx = C'(k_n)$.
        \item \textbf{Tail mass:} $\int_{k_n}^{\infty} f(x)\,dx = \lambda < 1$, the inequality strict since $k_n > 0$.
        \item \textbf{Non-negativity and monotonicity:} $f(x) > 0$ and $f'(x) < 0$ for all $x > k_n$.
        \item \textbf{Finite first moment:} $\int_{k_n}^{\infty} x\, f(x)\,dx < \infty$, since $\alpha > 2$.
        \item \textbf{Asymptotic boundary slope:} $\lim_{k\to\infty} C'(k) = 0$, the right-side
        boundary condition required for the global density to integrate to one.
    \end{enumerate}
\end{proposition}
Proof can be found in Appendix~\ref{sec:app_powerlaw}.

\subsection{Left-Tail Estimation}\label{sec:left_tail_pow}

The region $[0,k_1]$ to the left of the first observed strike requires a different construction from the right tail.  The right tail is anchored at a single endpoint and is completed on an unbounded domain.  The left region is instead a bounded boundary segment with information at both ends: the exact boundary values $C(0)=e^{-qT}$ and $C'(0)=-e^{-rT}$ at the origin, and the continuity conditions $C(k_1)=C_{k_1}$ and $C'(k_1)=C'_{k_1}$ at the first observed strike.  Thus the left construction must connect two endpoint anchors while preserving non-negative curvature.

The boundary values at the origin determine the boundary tangent line $\ell(k) := e^{-qT} - e^{-rT}k$, and the data at $k_1$ enter only through their departure from it, measured by two residuals,
\begin{align}\label{eq:left_tail_params}
\gamma_L := C_{k_1} - e^{-qT} + k_1\,e^{-rT}, \qquad
\delta_L := C'_{k_1} + e^{-rT},
\end{align}
the price excess $\gamma_L$ of $C_{k_1}$ above $\ell(k_1)$ and the slope increment $\delta_L$ from $C'(0) = -e^{-rT}$ to $C'_{k_1}$.  Both residuals vanish exactly when the left segment is the line $\ell$ itself; we call the left data non-degenerate in the opposite case, when $C''$ does not vanish identically on $(0, k_1)$ and the region carries genuine mass.  The following lemma gives the range of these residuals: in the non-degenerate case, convexity makes them positive and gives the ordering needed for an integrable left curvature.

\begin{lemma}[Strict positivity of the boundary residuals]\label{lem:left_data_pos}
If the input is butterfly-free and the left data are non-degenerate, then $\gamma_L > 0$, $\delta_L > 0$, and $\delta_L\,k_1 > \gamma_L$.
\end{lemma}
Proof can be found in Appendix~\ref{sec:app_left_pow}.

This suggests a minimal two-parameter curvature form for connecting the origin to the first observed strike: the two origin conditions are absorbed as integration constants, and the two remaining conditions at $k_1$ determine the two free parameters.  We take a power-law curvature on the segment,
\[
C''(k) = a_L\,k^{\beta_L}, \qquad k \in [0, k_1], \quad C''(k) \geq 0.
\]

Solving the two matching conditions at $k_1$ gives the exponent and amplitude:
\begin{align}\label{eq:left_pow_beta}
\beta_L = \frac{k_1\,\delta_L}{\gamma_L} - 2, \qquad
a_L = \frac{(\beta_L + 1)\,\delta_L}{k_1^{\beta_L+1}}.
\end{align}
The inequality $\delta_L k_1>\gamma_L$ is exactly the condition $\beta_L>-1$, so the resulting curvature is integrable near the origin.

The resulting density and call price on $[0,k_1]$ are
\begin{align}
f(k) &= e^{rT}\,a_L\,k^{\beta_L}, \quad k \in [0, k_1], \label{eq:left_pow}\\
C(k) &= e^{-qT} - e^{-rT}k + \gamma_L\left(\frac{k}{k_1}\right)^{\!\beta_L+2}, \quad k \in [0, k_1]. \label{eq:left_pow_compact}
\end{align}

\begin{proposition}[Power-law left tail properties]\label{prop:left_pow}
For non-degenerate, butterfly-free left data, the density~\eqref{eq:left_pow} with $\beta_L$ and $a_L$ given by~\eqref{eq:left_pow_beta} satisfies:
\begin{enumerate}
    \item \textbf{Price and derivative matching:} $C(0) = e^{-qT}$, $C'(0) = -e^{-rT}$, $C(k_1) = C_{k_1}$, $C'(k_1) = C'_{k_1}$.
    \item \textbf{Non-negative density:} $f(k) \geq 0$ for all $k \in [0, k_1]$, with $f(0^+) = 0$ when $\beta_L > 0$.
    \item \textbf{Left tail mass:} $\int_0^{k_1} f(x)\,dx = \delta_L\,e^{rT}$.
    \item \textbf{Uniqueness:} The parameters $(a_L, \beta_L)$ are uniquely determined by $(\gamma_L, \delta_L, k_1)$.
    \item \textbf{Finite first moment:} $\int_0^{k_1} x\,f(x)\,dx < \infty$ since $\beta_L > -1$.
    \item \textbf{Boundary slope coherence:} the matched value $C'(0) = -e^{-rT}$ in Property~1 is the left-side boundary condition required for the global density to integrate to one; together with Property~6 of Proposition~\ref{prop:powerlaw}, the construction closes the unit-mass bookkeeping.
\end{enumerate}
\end{proposition}
Proof can be found in Appendix~\ref{sec:app_left_pow}.

When the left segment falls into the degenerate case $\delta_L = 0$, the boundary slope already equals the first observed slope and $C''$ vanishes on $(0, k_1)$, and there is no curvature to fit.  The region $[0, k_1]$ is then the straight line through the boundary data, which is butterfly-free and matches the endpoints, so the construction uses it directly in place of the power-law tail.

\begin{theorem}[Single-maturity construction]\label{thm:single-maturity}
Fix a maturity $T$ with butterfly-free input prices.  Combining the
interior construction of Theorem~\ref{thm:two-piece}, the right-tail
completion of Proposition~\ref{prop:powerlaw}, and the left-tail completion of
Proposition~\ref{prop:left_pow} yields a call-price curve that
induces a risk-neutral marginal law on $[0,\infty)$.  This marginal law has
unit mass, reprices every observed quote exactly, and is free of butterfly
arbitrage at every strike.
\end{theorem}
\begin{proof}
Theorem~\ref{thm:two-piece} gives a convex interpolation on every interior
strike interval and matches the endpoint prices and slopes.  Propositions
\ref{prop:powerlaw} and~\ref{prop:left_pow} complete the two boundary
regions while preserving convexity and matching the required boundary data.  Hence the assembled call-price curve is convex on
$[0,\infty)$ and therefore induces a non-negative marginal measure.

The total mass of this measure is fixed by the endpoint slopes.  The left
boundary condition gives $C'(0)=-e^{-rT}$, and the right-tail construction
gives $C'(\infty)=0$.  By the slope-mass identity in Section~\ref{sec:prelim},
these two conditions imply unit total mass.  Exact repricing follows from
the endpoint matching on every observed strike, and butterfly freedom follows
from convexity of the assembled call-price curve.
\end{proof}

\section{Construction under Multiple Maturities}\label{sec:multi}
The single-maturity construction of Section~\ref{sec:single} applies independently to each maturity $T_j$.  The only additional issue is calendar arbitrage: the construction at $T_j$ must be consistent with the previously constructed curves at $T_1, \ldots, T_{j-1}$.  We process maturities sequentially; under forward moneyness alignment, the calendar-arbitrage-free condition takes the following form.

\paragraph{Calendar-arbitrage-free condition.}  
We use $T_{j-1} < T_j$ ($j \ge 2$) for a consecutive maturity pair.  Write $\Delta T_j := T_j - T_{j-1}$ and $\phi_j := e^{-(r-q)\Delta T_j}$, and for a strike $k$ of $T_j$ call $\phi_j k$ its forward-aligned strike at $T_{j-1}$.  The absence of calendar arbitrage between the two maturities requires
\[
C(k, T_j) \geq e^{-q\,\Delta T_j}\, C(\phi_j k,\, T_{j-1}) \qquad \text{for all } k \geq 0,
\]
where $C(k, T)$ denotes the European call price at strike $k$ and maturity $T$.  This is the standard forward-moneyness calendar-arbitrage condition; see \citet[Prop.~2.1]{fengler2009arbitrage}.

At the level of the observed strike grid, this ordering is treated as an input property.  Writing $C^{T_j}_{k_l}$ for the input price at an observed strike $k_l$ of $T_j$, one has $
C^{T_j}_{k_l} \geq e^{-q\,\Delta T_j}\, C\!\bigl(\phi_j k_l,\, T_{j-1}\bigr),$
with the right-hand side read from the input prices at $T_{j-1}$. 

In practice, such a discrete calendar ordering is typically supplied by an upstream de-arbitrage or surface-fitting step: after forward alignment, the quoted prices at the later maturity lie above the corresponding prices at the earlier maturity on the observed grid. However, this is only a grid-level guarantee. For the marginal-based applications targeted in this paper, the output is not merely a finite set of arbitrage-free quotes, but a family of globally defined marginal laws, equivalently call-price curves defined on the full strike domain. Therefore, calendar consistency must hold pointwise for every strike, not only at the observed grid points. 
The discrete ordering also does not by itself specify how the earlier-maturity price should be evaluated when the forward-aligned strike falls between two quoted strikes.

We therefore take the grid-level calendar ordering as an input property and focus on upgrading it to the distribution-level object required downstream: starting from discretely aligned prices, our construction builds globally defined marginal laws whose associated call-price curves satisfy the calendar inequality pointwise over the full strike domain. 
To make this upgrade well-defined when the forward-aligned strike falls between two observed strikes at $T_{j-1}$, we impose a mildly strengthened discrete condition in which the earlier-maturity price is read through a specified chord.

\begin{assumption}[Input calendar condition]\label{assump:input-calendar}
For each consecutive maturity pair $T_{j-1}<T_j$, $j\ge 2$, let
$\Delta T_j:=T_j-T_{j-1}$ and $\phi_j:=e^{-(r-q)\Delta T_j}$.
For maturity $T_{j-1}$, write its observed strikes locally as
$k_1<\cdots<k_N$. On each interval
$k\in[k_a,k_{a+1}]$, define the chord
$I_C(k):=\frac{k_{a+1}-k}{k_{a+1}-k_a}\,C^{T_{j-1}}_{k_a}
+\frac{k-k_a}{k_{a+1}-k_a}\,C^{T_{j-1}}_{k_{a+1}}$. We assume the following:
\begin{enumerate}[label=\textup{(\roman*)}]
    \item At the origin, the calendar relation holds with equality:
    $C^{T_j}_0=e^{-q\Delta T_j}C^{T_{j-1}}_0$.

    \item For every observed strike $k_l>0$ of $T_j$ whose forward-aligned
    strike $\phi_j k_l$ lies within the observed strike range of
    $T_{j-1}$,
    \begin{equation}\label{eq:S-input}
        C^{T_j}_{k_l} > e^{-q\,\Delta T_j}\, I_C(\phi_j k_l).
    \end{equation}
\end{enumerate}
\end{assumption}

The first assumption is trivial and has been stated before. For the second assumption, using the chord $I_C^{T_{j-1}}$ makes the assumption well-defined even when $\phi_j k_l$ falls between two observed strikes of $T_{j-1}$, and the strict inequality is what the construction will consume.  We treat Eq~\eqref{eq:S-input} as an upstream responsibility of the surface-fitting layer that supplies our input data.

The role of Eq~\eqref{eq:S-input} is to make the input calendar condition independent of the upstream interpolation. Different de-arbitrage or surface-fitting methods may assign different prices between quoted strikes. They may therefore imply different calendar slacks when $\phi_j k_l$ does not coincide with a quoted strike at $T_{j-1}$. The chord provides a common conservative reference. Indeed, any convex call curve matching the neighboring input prices at $T_{j-1}$ lies below $I_C^{T_{j-1}}$. 
Thus, positive slack against the chord implies positive slack against any such upstream convex curve.

On a coarse upstream grid, the chord-based condition may be stricter than the underlying pointwise calendar condition. 
This happens because the chord-to-curve gap can absorb the available calendar slack. 
If the upstream representation is already pointwise calendar-free with strictly positive slack, this is only a discretization issue. 
Sampling additional forward-aligned strikes from the same representation reduces the chord-to-curve gap.  On a sufficiently refined grid, condition~\eqref{eq:S-input} is then recovered. 
This refinement preserves the upstream no-arbitrage structure, and it also gives our construction a common discrete interface to different upstream methods.

Eq~\eqref{eq:S-input} is imposed only for strikes whose
forward-aligned locations lie inside the observed strike range of
$T_{j-1}$. Aligned strikes outside this range are handled by the tail-specific
calendar conditions, and the case $j=1$ is trivial because there is no previous
maturity.

Given this assumption, we next establish a lemma that converts the input condition at observed strikes into strict positive calendar slack at those strikes. 
This lemma will be used in the construction below to certify calendar-arbitrage-freeness.

\begin{lemma}[Endpoint slack from the input calendar condition]
    \label{lemma:cross-maturity}
    Suppose Assumption \ref{assump:input-calendar} holds and that the constructed
    call price curves at every prior maturity $T_l$ ($l < j$) (i)~are
    butterfly-free, i.e., $C(\cdot, T_l)$ is convex in $k$ on each
    observed interval, and (ii)~match the input prices at every observed
    strike, i.e., $C(k_a^{T_l}, T_l) = C^{T_l}_{k_a^{T_l}}$ for every
    observed strike $k_a^{T_l}$.  Let
    $L(k) := e^{-q\,\Delta T_j}\,C\!\bigl(\phi_j k,\, T_{j-1}\bigr)$
    be the calendar lower bound at $T_j$ obtained from the constructed
    $T_{j-1}$ curve for strike $k$.  Then for every $j \geq 2$ and every
    observed strike $k$ of $T_j$ satisfying the forward-alignment
    hypothesis of Assumption \ref{assump:input-calendar},
    \[
    C^{T_j}_{k} - L(k) > 0.
    \]
    Consequently, for every interior interval $[k_l, k_r]$ of $T_j$ both
    of whose endpoint strikes satisfy Assumption \ref{assump:input-calendar}, the endpoint calendar slacks
    \[
    \delta_l := C^{T_j}_{k_l} - L(k_l), \qquad
    \delta_r := C^{T_j}_{k_r} - L(k_r)
    \]
    satisfy $\delta_{\min} := \min(\delta_l, \delta_r) > 0$.
\end{lemma}
Proof can be found in Appendix~\ref{sec:app_cross_maturity}.

Lemma~\ref{lemma:cross-maturity} supplies strict calendar slack at the observed strikes. The next step is to extend this endpoint slack across each interior interval. This is done by the three-piece construction below, which upgrades only those intervals where the two-piece construction would not satisfy the inherited calendar floor. Together with Proposition~\ref{prop:calendar-reduction} and Lemma~\ref{lemma:three-piece-feasibility}, this gives pointwise calendar freedom throughout the interior.

\subsection{Three-Piece Construction for Calendar Arbitrage Elimination}\label{sec:three-piece}

On every interior interval $[k_l,k_r]$ at maturity $T_j$ with $j\ge 2$, set $\Delta k_l := k_r-k_l$, $A_l := C'_{k_r}-C'_{k_l}$, $P_l := C_{k_r}-C_{k_l}-C'_{k_l}\Delta k_l$, and $\rho := P_l/(A_l\Delta k_l)$. 
Throughout this subsection we work on non-degenerate intervals, where $A_l>0$ and $\rho\in(0,1)$. 
Degenerate endpoint cases, including $A_l=0$, $\rho=0$, and $\rho=1$, are treated separately and are not part of the local three-piece construction.
We also use $C(0,T)=e^{-qT}$, so calendar equality holds at $k=0$.

The two-piece construction of Section~\ref{sec:two-piece-opt} matches endpoint prices and slopes, but it leaves no free control over the first-piece curvature. 
This becomes restrictive when the inherited calendar floor requires a larger initial curvature. 
We therefore use a three-piece piecewise-constant curvature on such intervals. 
The additional degrees of freedom make $h_1$ adjustable, while the remaining parameters preserve endpoint matching and non-negative curvature.

\paragraph{Per-interval invocation.}
The three-piece form is invoked locally, not globally. 
If the two-piece fit already satisfies the calendar condition on an interval, we keep the two-piece construction unchanged. 
Only intervals that violate the inherited calendar floor are replaced by the three-piece form. 
This local replacement keeps the construction computationally efficient. 
At a fixed maturity $T_j$, the interval problems remain decoupled and can be solved in parallel. 
Across maturities, the algorithm remains sequential because the calendar floor $h_1^{\rm cal}$ discussed later is inherited from the previously constructed curve. 

On an interval $[k_l, k_r]$ where the three-piece form is invoked, with $u := k - k_l \in [0, \Delta k_l]$, $\Delta k_l := k_r - k_l$, $A_l := C'_{k_r} - C'_{k_l}$, and the local breakpoints
\[
a_1 := \beta_1\,\Delta k_l, \qquad a_2 := (1 - \beta_2)\,\Delta k_l,
\]
we define
\begin{equation}\label{eq:three-Cpp}
    C^{\prime\prime}(k_l + u) =
    \begin{cases}
        h_1, & u \in \bigl[0,\; a_1\bigr), \\[3pt]
        h_{\rm mid}, & u \in \bigl[a_1,\; a_2\bigr), \\[3pt]
        h_3, & u \in \bigl[a_2,\; \Delta k_l\bigr],
    \end{cases}
\end{equation}
where the five free parameters are subject to
\[
h_1,\, h_{\rm mid},\, h_3 \geq 0,
\qquad \beta_1,\, \beta_2 \in [0, 1],
\qquad \beta_1 + \beta_2 \leq 1.
\]
Non-negativity of $h_1, h_{\rm mid}, h_3$ enforces butterfly freedom.

The parameters $\beta_1$ and $\beta_2$ specify the left and right piece widths as fractions of $\Delta k_l$. 
Thus the middle piece has width $(1-\beta_1-\beta_2)\Delta k_l$. 
The breakpoints $a_1=\beta_1\Delta k_l$ and $a_2=(1-\beta_2)\Delta k_l$ are measured from the left endpoint $k_l$. 
In particular, $a_2$ is a breakpoint location, not the length of the middle piece. 
When $\beta_2=0$, the third piece degenerates and the form reduces to the two-piece construction of Section~\ref{sec:two-piece-opt}.

Integrating~\eqref{eq:three-Cpp} once with the boundary value $C'(k_l) = C'_{k_l}$ yields the first derivative, piecewise linear on $[0, \Delta k_l]$ and continuous at both interior breakpoints $a_1$ and $a_2$:
\begin{equation}\label{eq:three-Cp}
    C'(k_l + u) =
    \begin{cases}
        C'_{k_l} + h_1\,u, & u \in [0,\; a_1], \\[3pt]
        C'_{k_l} + h_1\,a_1 + h_{\rm mid}\,(u - a_1), & u \in [a_1,\; a_2], \\[3pt]
        C'_{k_l} + h_1\,a_1 + h_{\rm mid}\,(a_2 - a_1) + h_3\,(u - a_2),
        & u \in [a_2,\; \Delta k_l].
    \end{cases}
\end{equation}
Although $C''$ may jump at $a_1$ and $a_2$, integrating the curvature makes $C'$ continuous across both breakpoints. 
Integrating once more with the boundary value $C(k_l)=C_{k_l}$, and writing $g(u):=C(k_l+u)-C_{k_l}-C'_{k_l}u$ with $g(0)=g'(0)=0$, this gives

\begin{equation}\label{eq:three-C}
    g(u) =
    \begin{cases}
        \tfrac{1}{2}\,h_1\,u^2,
        & u \in [0,\; a_1], \\[5pt]
        \tfrac{1}{2}\,h_1\,a_1^2 + h_1\,a_1\,(u - a_1)
        + \tfrac{1}{2}\,h_{\rm mid}\,(u - a_1)^2,
        & u \in [a_1,\; a_2], \\[5pt]
        \tfrac{1}{2}\,h_1\,a_1^2 + h_1\,a_1\,(a_2 - a_1)
        + \tfrac{1}{2}\,h_{\rm mid}\,(a_2 - a_1)^2 \\
        \quad + \bigl[h_1\,a_1 + h_{\rm mid}\,(a_2 - a_1)\bigr]\,(u - a_2)
        + \tfrac{1}{2}\,h_3\,(u - a_2)^2,
        & u \in [a_2,\; \Delta k_l].
    \end{cases}
\end{equation}
The constructed curve $C(\cdot,T_j)$ is $C^1$ on each interior interval. 
At this point, the three-piece form defines a local shape class, but its five parameters have not yet been tied to the interval data. 
We next impose endpoint slope and price matching. 
These two constraints preserve exact interpolation of the input prices and leave the remaining local freedom for the calendar floor and the optional smoothness criterion.

\paragraph{Endpoint matching constraints.}
The five free parameters are tied to the input interval data $(A_l, P_l)$ via two equality constraints, where $P_l := C_{k_r} - C_{k_l} - C'_{k_l}\,\Delta k_l$.  Imposing $C'(k_r) = C'_{k_l} + A_l$ in the third piece of~\eqref{eq:three-Cp} at $u = \Delta k_l$ gives the slope-matching constraint
\begin{equation}\label{eq:three-slope}
    h_1\,\beta_1\,\Delta k_l
    + h_{\rm mid}\,(1 - \beta_1 - \beta_2)\,\Delta k_l
    + h_3\,\beta_2\,\Delta k_l
    = A_l.
\end{equation}
Imposing $C(k_r) = C_{k_r}$, equivalently $g(\Delta k_l) = P_l$, in the third piece of~\eqref{eq:three-C} at $u = \Delta k_l$, and simplifying using~\eqref{eq:three-slope}, gives the price-matching constraint
\begin{equation}\label{eq:three-h}
    h_1\,\beta_1\,\Bigl(1 - \tfrac{\beta_1}{2}\Bigr)\,(\Delta k_l)^2
    + h_{\rm mid}\,
    \tfrac{(1 - \beta_1 - \beta_2)(1 - \beta_1 + \beta_2)}{2}\,(\Delta k_l)^2
    + h_3\,\tfrac{\beta_2^{\,2}}{2}\,(\Delta k_l)^2
    = P_l.
\end{equation}

The price-matching equation has a useful integral interpretation. 
Since $g(0)=g'(0)=0$, the endpoint residual satisfies $P_l=g(\Delta k_l)=\int_0^{\Delta k_l}(\Delta k_l-v)C''(k_l+v)\,dv$. Because $C''$ is constant on each piece, this integral decomposes into three weighted curvature contributions over the intervals $[0,a_1]$, $[a_1,a_2]$, and $[a_2,\Delta k_l]$. The weight $\Delta k_l-v$ is the remaining distance from location $v$ to the right endpoint, so curvature placed farther to the left has a larger effect on the endpoint price. 
Thus~\eqref{eq:three-h} is exactly the endpoint price-matching condition. 
Together with the slope-matching condition~\eqref{eq:three-slope}, it leaves three local degrees of freedom in the five-parameter family.

After endpoint matching, the remaining freedom is used to enforce the calendar condition $C(\cdot,T_j)\ge L$ on $[k_l,k_r]$. This is a pointwise inequality in $k$, hence an infinite-dimensional constraint. The purpose of this subsection is to reduce it to a single lower bound on the first-piece curvature $h_1$.

The reduction compares three curves on the interval: the inherited calendar floor $L$, the chord upper bound $I_C$, and the constructed interval curve $C$. Here $I_C$ is the chord through the endpoint prices, $I_C(k):=\tfrac{k_r-k}{k_r-k_l}\,C_{k_l}+\tfrac{k-k_l}{k_r-k_l}\,C_{k_r}$, and lies above the convex curve $C$ on the interval. With $\delta_l:=C_{k_l}-L(k_l)$, $\delta_r:=C_{k_r}-L(k_r)$, and $\delta_{\min}:=\min\{\delta_l,\delta_r\}$, convexity of $L$ gives $I_C(k)-L(k)\ge \delta_{\min}$ for every $k\in[k_l,k_r]$, and $I_C$ exceeds the calendar floor by at least $\delta_{\min}$.

 A calendar violation can occur only where the constructed interval curve lies below $I_C$: since $C(k)-L(k)=[I_C(k)-L(k)]-[I_C(k)-C(k)]$ and the first term is at least $\delta_{\min}$, the calendar condition holds whenever $\sup_{k\in[k_l,k_r]}\{I_C(k)-C(k)\}\le \delta_{\min}$.

It remains to control the maximum gap below \(I_C\). Rather than bounding this gap for every curve in the local construction class, we give an explicit feasible curve whose gap can be computed in closed form. Fix the first-piece curvature \(h_1\), write \(u=k-k_l\), and define \(D(u):=I_C(k_l+u)-C(k_l+u)\). Since the chord \(I_C\) has slope \(C'_{k_l}+\rho A_l\) and the constructed curve starts with slope \(C'_{k_l}\), the initial gap grows at rate \(D'(0)=\rho A_l\). On the first piece, \(D'(u)=\rho A_l-h_1u\). Hence \(D\) reaches its largest value at \(u^\star=\rho A_l/h_1\), where the slope gap vanishes. This gives \(D(u^\star)=\int_0^{u^\star}(\rho A_l-h_1s)\,ds=\rho^2A_l^2/(2h_1)\). Thus a larger first-piece curvature \(h_1\) decreases \(u^\star\) and reduces the maximum gap below \(I_C\).

The calculation above gives the key quantity needed for the reduction: for a suitable constructed interval curve, the maximum gap below the chord \(I_C\) is \(\rho^2A_l^2/(2h_1)\). The proposition below justifies this calculation by constructing a feasible curve \(C^\star\) with this gap bound. This construction proves existence. The actual per-interval optimization may select a different curve depending on the chosen objective.

The gap formula alone is not yet enough for the calendar reduction. It treats \(h_1\) as a candidate first-piece curvature, but a valid interval construction must also fit inside \([k_l,k_r]\) and stay above the calendar floor. The first requirement is geometric. In the explicit construction used for the existence proof, the left nonzero-curvature piece has length \(\rho A_l/h_1\), and the right nonzero-curvature piece required by endpoint matching has length \(\rho^2A_l/(h_1(1-\rho))\). These two pieces fit inside the interval exactly when \(\rho A_l/h_1+\rho^2A_l/(h_1(1-\rho))\le \Delta k_l\), equivalently \(h_1\ge \rho A_l/((1-\rho)\Delta k_l)\). This defines the width-feasibility bound \(h_1^{\rm geom}\).

The second requirement comes from the calendar slack. Since the chord \(I_C\) lies at least \(\delta_{\min}\) above the calendar floor \(L\), it is enough to keep the maximum gap below \(I_C\) no larger than \(\delta_{\min}\). Using the gap formula gives \(\rho^2A_l^2/(2h_1)\le\delta_{\min}\), equivalently \(h_1\ge \rho^2A_l^2/(2\delta_{\min})\). This defines the calendar-slack bound \(h_1^{\rm slack}\). The proposition below combines the chord bound, the explicit gap bound, and these two requirements into one scalar curvature floor.

\begin{proposition}[Calendar reduction]\label{prop:calendar-reduction}
    Let $[k_l,k_r]$ be an interval at maturity $T_j$ with $j\geq 2$. 
    Define $\Delta k_l:=k_r-k_l$, $A_l:=C'_{k_r}-C'_{k_l}$, $P_l:=C_{k_r}-C_{k_l}-C'_{k_l}\Delta k_l$, and $\rho:=P_l/(A_l\Delta k_l)\in(0,1)$. 
    Let $L(k):=e^{-q\Delta T_j}C(e^{-(r-q)\Delta T_j}k,T_{j-1})$ be the inherited calendar floor. 
    Suppose that the endpoint slacks $\delta_l:=C_{k_l}-L(k_l)$ and $\delta_r:=C_{k_r}-L(k_r)$ are strictly positive, and set $\delta_{\min}:=\min(\delta_l,\delta_r)$. 
    Let $I_C$ be the chord upper bound through the endpoint prices $(k_l,C_{k_l})$ and $(k_r,C_{k_r})$.

    Define
    \begin{equation}\label{eq:h1-geom}
        h_1^{\rm geom}:=\frac{\rho A_l}{(1-\rho)\Delta k_l},
    \end{equation}
    \begin{equation}\label{eq:h1-slack}
        h_1^{\rm slack}:=\frac{\rho^2 A_l^2}{2\delta_{\min}},
    \end{equation}
    and
    \begin{equation}\label{eq:h1-cal}
        h_1^{\rm cal}:=\max\{h_1^{\rm geom},h_1^{\rm slack}\}.
    \end{equation}

    Then the following statements hold.
    \begin{enumerate}[label=\textup{(\alph*)}]
        \item \textbf{Chord domination.}
        For all $k\in[k_l,k_r]$,
        \[
            I_C(k)-L(k)\geq \delta_{\min}.
        \]

        \item \textbf{Gap bound under the geometric lower bound.}
        Fix any first-piece curvature $h_1$ satisfying
        \[
            h_1\geq h_1^{\rm geom}.
        \]
        Then there exists an interval curve $C^\star$ in the local construction class, matching the endpoint prices and slopes, such that
        \[
            \sup_{k\in[k_l,k_r]}\bigl[I_C(k)-C^\star(k)\bigr]
            =
            \frac{\rho^2 A_l^2}{2h_1}.
        \]

        \item \textbf{Sufficient curvature floor for calendar freedom.}
        If $h_1\geq h_1^{\rm cal}$, then the interval curve $C^\star$ in part~\textup{(b)} satisfies
        \[
            C^\star(k)\geq L(k),\qquad k\in[k_l,k_r].
        \]
    \end{enumerate}
\end{proposition}

Proof can be found in Appendix~\ref{sec:app_calendar_reduction}.

Proposition~\ref{prop:calendar-reduction} turns the pointwise calendar constraint into a constructive curvature requirement. The argument does not need to characterize every interval curve satisfying \(C\ge L\). It works with curves whose maximum gap below the chord \(I_C\) is no larger than the endpoint calendar slack. This stronger condition implies \(C\ge L\), and the proposition shows that such a curve can be constructed explicitly.

The role of \(h_1^{\rm geom}\) is to make this explicit curve well-defined on the interval. It ensures that the pieces used in the construction fit inside \([k_l,k_r]\), so the gap formula is attached to an actual feasible curve. The role of \(h_1^{\rm slack}\) is then to keep the certified maximum gap below \(I_C\) within \(\delta_{\min}\). Their maximum \(h_1^{\rm cal}\) gives one scalar lower bound that secures both requirements.

In the per-interval optimization~\eqref{eq:three-piece-opt} below, this scalar bound is imposed on the free first-piece curvature \(h_1\). Lemma~\ref{lemma:three-piece-feasibility} then shows that the constrained interval problem remains feasible. Thus the calendar constraint does not create a feasibility obstruction; the construction retains a certified calendar-free choice on every interval.
The full per-interval optimization problem on $[k_l, k_r]$ is then:
\begin{equation}\label{eq:three-piece-opt}
    \boxed{\;
        \begin{aligned}
            \min_{h_1,\, h_{\rm mid},\, h_3,\, \beta_1,\, \beta_2}\quad &
            \mathcal{F}\bigl(h_1,\, h_{\rm mid},\, h_3,\, \beta_1,\, \beta_2\bigr)
            \\[3pt]
            \text{s.t.}\quad &
            \text{\eqref{eq:three-slope} and \eqref{eq:three-h}}
            \quad \text{(slope and price matching)} \\[2pt]
            & h_1,\, h_{\rm mid},\, h_3 \geq 0
            \quad \text{(butterfly freedom)} \\[2pt]
            & h_1 \geq h_1^{\rm cal}
            \quad \text{(local calendar floor; see~\eqref{eq:h1-cal})} \\[2pt]
            & \beta_1,\, \beta_2 \in [0, 1], \quad \beta_1 + \beta_2 \leq 1.
        \end{aligned}
        \;}
\end{equation}
The objective $\mathcal{F}$ is a per-interval smoothness criterion chosen by the practitioner; reasonable defaults include minimizing the curvature jumps $|h_3 - h_1|$ and $|h_{\rm mid} - h_1|$, or a quadratic functional that penalizes $h_{\rm mid}$ approaching zero (which would create a zero-density region on the middle piece).  The constraints alone guarantee exact price and slope matching, butterfly freedom, and pointwise calendar arbitrage freedom. The choice of $\mathcal{F}$ does not affect this correctness.

\begin{lemma}[Three-piece feasibility]\label{lemma:three-piece-feasibility}
    Fix an interior interval $[k_l, k_r]$ at maturity $T_j$ ($j \ge 2$) with $\Delta k_l > 0$, $A_l > 0$, $\rho := P_l/(A_l\,\Delta k_l) \in (0, 1)$, and the cross-maturity slack $\delta_{\min} > 0$ from Lemma~\ref{lemma:cross-maturity}.  The feasible set of the per-interval optimization~\eqref{eq:three-piece-opt} is non-empty.  Consequently, the three-piece construction always admits a butterfly-free, calendar-free completion that exactly matches the endpoint prices and slopes.
\end{lemma}
Proof can be found in Appendix~\ref{sec:app_three_piece_feasibility}.

Combining Lemma~\ref{lemma:three-piece-feasibility} with Proposition~\ref{prop:calendar-reduction} and the cross-maturity slack of Lemma~\ref{lemma:cross-maturity}, we obtain the main feasibility result:

\begin{theorem}[Unified interior feasibility]\label{thm:unified}
Let \(T_1<\cdots<T_m\) be the maturity grid. Suppose that the input prices are butterfly-free at each maturity, and that for every \(j\ge 2\) the input calendar assumption \ref{assump:input-calendar} holds at every observed strike of \(T_j\). Then there exists a feasible construction such that the interval curve matches the endpoint prices and slopes, has non-negative curvature, and satisfies the pointwise calendar inequality \(C(k,T_j)\ge e^{-q\Delta T_j}C(e^{-(r-q)\Delta T_j}k,T_{j-1})\) throughout the interval.
\end{theorem}

Proof can be found in Appendix~\ref{sec:app_unified}.

Assumption~\ref{assump:input-calendar} and Lemma~\ref{lemma:cross-maturity} give strictly positive endpoint slack on each interior interval whose endpoint strikes satisfy~\eqref{eq:S-input}. Proposition~\ref{prop:calendar-reduction} converts this slack into the scalar curvature floor $h_1^{\rm cal}$. Lemma~\ref{lemma:three-piece-feasibility} then shows that the interval constraints remain feasible under this floor. Hence each such interval admits a constructed curve that matches the endpoint data, has non-negative curvature, and satisfies the inherited calendar floor.

This completes the calendar-free construction on the observed strike range. It remains to handle the two unquoted regions. The right tail starts from the largest observed strike, while the left tail connects the origin to the smallest observed strike. The next two subsections give sufficient conditions under forward moneyness alignment that make these power-law tail pieces calendar-free across maturities. The boundary slopes used in these tails are then selected globally in Section~\ref{sec:global_slope}.

\subsection{Power-Law Right Tail under Forward Moneyness Alignment}\label{sec:right_tail_calendar}

Recall from Section~\ref{sec:powerlaw_tail} that the power-law tail density at maturity $T_j$ is $f(x) = \lambda_j \frac{\alpha_j-1}{k_n^{T_j}}\left(\frac{k_n^{T_j}}{x}\right)^{\alpha_j}$. Integrating to obtain the call price (see Appendix~\ref{sec:app_powerlaw}) yields:
$$
C(k, T_j) = C(k_n^{T_j}, T_j) \cdot \left(\frac{k_n^{T_j}}{k}\right)^{\alpha_j - 2}, \quad k \geq k_n^{T_j},
$$
where $\alpha_j = 2 + \frac{|C'(k_n^{T_j})| \cdot k_n^{T_j}}{C(k_n^{T_j}, T_j)} > 2$, consistent with~\eqref{eq:alpha_match}.

Under forward moneyness alignment $k_n^{T_j} = e^{-(r-q)(T_{j+1}-T_j)} k_n^{T_{j+1}}$, we obtain the following result:

\begin{lemma}[Calendar arbitrage-free power-law tail]\label{lemma:powerlaw_calendar}
    Let $C(k, T_j) = C(k_n^{T_j}, T_j) \left(\frac{k_n^{T_j}}{k}\right)^{\alpha_j - 2}$ for $k \geq k_n^{T_j}$, with forward moneyness alignment $k_n^{T_j} = e^{-(r-q)(T_{j+1}-T_j)} k_n^{T_{j+1}}$ for all $j \in \{1,\ldots,m-1\}$.
    
    If $\alpha_j$ is non-increasing in $j$ and the anchor-point calendar condition $C(k_n^{T_{j+1}}, T_{j+1}) \geq e^{-q\Delta T_j}\,C(k_n^{T_j}, T_j)$ holds for all $j$, then the power-law right-tail construction is calendar arbitrage-free under forward moneyness alignment.
\end{lemma}

Proof can be found in Appendix~\ref{sec:app_powerlaw_calendar}.

The anchor-point calendar condition is ensured by the interior construction at the observed strike~$k_n$.

\subsection{Power-Law Left Tail under Forward Moneyness Alignment}\label{sec:lhs_pow_calendar}

When the power-law left tail~\eqref{eq:left_pow} is used, the calendar condition takes a form that mirrors the right-tail result.  Recall that at maturity~$T_j$, the call price on $[0, k_1^{T_j}]$ is, via~\eqref{eq:left_pow_compact},
$$
C(k, T_j) = e^{-qT_j} - e^{-rT_j}k + \gamma_{L,j}\left(\frac{k}{k_1^{T_j}}\right)^{\!\beta_j + 2}.
$$

Under forward moneyness alignment $k_1^{T_j} = e^{-(r-q)\Delta T_j}\,k_1^{T_{j+1}}$, we obtain:

\begin{lemma}[Calendar arbitrage-free power-law left tail]\label{lemma:lhs_pow_calendar}
    Let $C(k, T_j)$ be given by~\eqref{eq:left_pow_compact} on $[0, k_1^{T_j}]$ with exponent~$\beta_j$, and let $k_1^{T_j} = e^{-(r-q)\Delta T_j}\,k_1^{T_{j+1}}$ for all $j \in \{1,\ldots,m-1\}$.
    
    If $\beta_j$ is non-increasing in $j$ and the anchor-point calendar condition $\gamma_{L,j+1} \geq e^{-q\Delta T_j}\,\gamma_{L,j}$ holds for all $j$, then the power-law left-tail construction is calendar arbitrage-free under forward moneyness alignment.
\end{lemma}
Proof can be found in Appendix~\ref{sec:app_lhs_pow_calendar}.

The anchor-point calendar condition is ensured by the interior construction at the observed strike~$k_1$.

The preceding lemma covers the non-degenerate left tail, where the power-law
curvature in~\eqref{eq:left_pow} is used. It remains to discuss the boundary
case in which the left segment carries no curvature. In that case the
power-law formula is not invoked, and the segment is fixed by the boundary
values at the origin. The following lemma shows that this replacement is also
calendar-free.

\begin{lemma}[Degenerate left segment under calendar alignment]\label{lemma:base-case}
Suppose the left segment at $T_j$, $j \ge 2$, is degenerate, so that
$C''$ vanishes on $(0,k_1)$. Equivalently, the left-segment residuals at $T_j$ vanish, $\gamma_L=\delta_L=0$, and the
non-degenerate hypothesis of Proposition~\ref{prop:left_pow} fails. Then the
left region is the linear segment
$C(k,T_j)=e^{-qT_j}-e^{-rT_j}k$ on $[0,k_1]$. It is butterfly-free, matches
the endpoint price and slope, and satisfies the calendar inequality on
$[0,k_1]$.
\end{lemma}

Proof can be found in Appendix~\ref{app:degenerate_left_tail}.

\subsection{Global Slope Selection for Calendar-Free Tails}\label{sec:global_slope}

To ensure the calendar conditions on both tails, the boundary slopes $C'(k_1^{T_j})$ and $C'(k_n^{T_j})$, $j = 1,\ldots,m$, must be selected so that $\alpha_j$ is non-increasing (right tail, Lemma~\ref{lemma:powerlaw_calendar}) and $\beta_j$ is non-increasing (left tail, Lemma~\ref{lemma:lhs_pow_calendar}).  We first compute per-maturity boundary-slope estimates $\widehat{C'}(k_1^{T_j})$ via Algorithm~\ref{alg:c-k-boundary} and $\widehat{C'}(k_{n_j}^{T_j})$ via backward three-point Lagrange differentiation on the window at the largest observed strikes.  We then solve a global optimization for the tail-feasible slopes closest to these estimates.

\paragraph{Right tail.}  The calendar condition requires $\alpha_j$ to be non-increasing, where $\alpha_j - 2 = |C'(k^{T_j}_n)| \cdot k^{T_j}_n / C(k^{T_j}_n, T_j)$.  This yields:
\begin{align}
    &\min_{C'(k^{T_j}_n)} \sum_{j=1}^m \left( C'(k^{T_j}_n) - \widehat{C'}(k^{T_j}_n) \right)^2 \label{eq:global_pow}\\
    &\text{s.t.} \quad \frac{|C'(k^{T_j}_n)| \cdot k^{T_j}_n}{C(k^{T_j}_n, T_j)} \geq \frac{|C'(k^{T_{j+1}}_n)| \cdot k^{T_{j+1}}_n}{C(k^{T_{j+1}}_n, T_{j+1})}, \quad j = 1,\ldots,m-1. \nonumber
\end{align}

\paragraph{Left tail.}  The free parameter is $s_j := C'(k_1^{T_j})$, which determines $\beta_j$ via~\eqref{eq:left_pow_beta}.  Since $\beta_j = k_1^{T_j}(s_j + e^{-rT_j})/\gamma_{L,j} - 2$ is increasing in $s_j$, the non-increasing condition on~$\beta_j$ is a linear inequality in~$s_j$:
\begin{align}
    &\min_{s_j} \sum_{j=1}^m \left(s_j - \widehat{C'}(k_1^{T_j}) \right)^2 \label{eq:global_lhs_pow}\\
    &\text{s.t.} \quad \frac{k_1^{T_j}(s_j + e^{-rT_j})}{\gamma_{L,j}} \geq \frac{k_1^{T_{j+1}}(s_{j+1} + e^{-rT_{j+1}})}{\gamma_{L,j+1}}, \quad j = 1,\ldots,m-1. \nonumber
\end{align}

Both tail-slope programs are finite-dimensional convex quadratic programs with
linear inequality constraints. They are solved independently because the
right-tail slopes \(C'(k_n^{T_j})\) and the left-tail slopes \(C'(k_1^{T_j})\)
enter separate tail formulas.

The feasible set of each program is non-empty by direct construction, as the
following examples show. For the right tail, since \(C'(k_n^{T_j})<0\) we have
\(\alpha_j-2=-C'(k_n^{T_j})\,k_n^{T_j}/C(k_n^{T_j},T_j)\); choosing any small
\(a>0\) and setting \(C'(k_n^{T_j})=-a\,C(k_n^{T_j},T_j)/k_n^{T_j}\) gives
\(\alpha_j=2+a\) for every \(j\), so the right-tail monotonicity constraints
hold with equality. For the left tail, since
\(\beta_j=k_1^{T_j}(C'(k_1^{T_j})+e^{-rT_j})/\gamma_{L,j}-2\), choosing any
\(b>-1\) and setting \(C'(k_1^{T_j})=-e^{-rT_j}+(b+2)\,\gamma_{L,j}/k_1^{T_j}\)
gives \(\beta_j=b\) for every \(j\), so the left-tail monotonicity constraints
also hold with equality.

The selected boundary slopes are then used as endpoint slopes in the adjacent
body intervals. The body construction enforces non-negative curvature on the
observed strike range, while Propositions~\ref{prop:powerlaw} and
\ref{prop:left_pow} enforce non-negative curvature on the two tails. Thus the
assembled curve remains convex across the body-tail junctions.

Algorithm~\ref{alg:global} summarizes the procedure.

\begin{algorithm}[ht]
    \caption{Global slope selection for calendar-free tails}\label{alg:global}
    \begin{algorithmic}[1]
        \Require Maturities $T_1 < T_2 < \dots < T_m$
        \For{$j = 1$ to $m$}
        \State Compute boundary-slope estimates $\widehat{C'}(k_1^{T_j})$ via Algorithm~\ref{alg:c-k-boundary} and $\widehat{C'}(k_{n_j}^{T_j})$ via backward three-point Lagrange on $\{k_{n_j-2}^{T_j}, k_{n_j-1}^{T_j}, k_{n_j}^{T_j}\}$
        \EndFor
        \State Solve~\eqref{eq:global_pow} for right-tail slopes $\{C'(k^{T_j}_n)\}_{j=1}^m$
        \State Solve~\eqref{eq:global_lhs_pow} for left-tail slopes $\{C'(k_1^{T_j})\}_{j=1}^m$
        \State \Return $\{C'(k_1^{T_j})\}_{j=1}^m$ and $\{C'(k^{T_j}_n)\}_{j=1}^m$
    \end{algorithmic}
\end{algorithm}

\subsection{Density-Anchored Right Tail Extension}\label{sec:density_anchored_tail}

The two-anchor right-tail construction introduced in Section~\ref{sec:powerlaw_tail} and carried to the multi-maturity setting in Section~\ref{sec:right_tail_calendar} matches the boundary data $(C(k_n^{T_j},T_j), C'(k_n^{T_j}))$ using two free parameters of the power-law construction. No degree of freedom remains to match the interior density value $f(k_n^{T_j-}) = e^{rT_j}\,h_3^{T_j}$ from the last interior interval.  The constructed density can therefore be discontinuous at the junction strike $k_n^{T_j}$, where the interior construction meets the tail.  This subsection introduces an optional extension that adds a third anchor, the interior density at the junction strike, together with a per-maturity selection rule: some maturities may use the augmented family while others retain the two-anchor form, all under a single calendar-arbitrage condition.  As in Remark~\ref{rmk:step-density}, this extension is an optional smoothness refinement; every guarantee of the construction holds without it.

\paragraph{Unified shifted power-law family.}
For each maturity $T_j$, consider the family
\begin{align}\label{eq:shifted_pl_family}
    f(x) = a_j\,(x - k_n^{T_j} + \mu_j)^{-\alpha_j}, \qquad x \in [k_n^{T_j}, \infty),
\end{align}
with $\alpha_j > 2$, $\mu_j > 0$, $a_j > 0$.  Two specializations are given:
\begin{itemize}
    \item \textbf{Two-anchor power-law tail (PL$_2$)}: $\mu_j = k_n^{T_j}$ is fixed; the remaining two parameters $(a_j, \alpha_j)$ are determined by mass and price matching.  This recovers exactly the construction of Section~\ref{sec:powerlaw_tail} and Proposition~\ref{prop:powerlaw}.
    \item \textbf{Three-anchor power-law tail (PL$_3$)}: $\mu_j$ is free; the three parameters $(a_j, \alpha_j, \mu_j)$ are determined by mass, price, and a junction density match.
\end{itemize}

A choice of PL$_2$ or PL$_3$ at each maturity defines a sequence we call a mixed-tail configuration.

The three-anchor case uses the density produced by the last interior interval
as an additional boundary condition. If the last interior curvature at
$k_n^{T_j}$ is $h_3^{T_j}$, then the left-limit density at the junction is
$e^{rT_j}h_3^{T_j}$. We therefore set
$D_0^{T_j}:=e^{rT_j}h_3^{T_j}$. The other two quantities are the tail mass
$\lambda_j:=-C'(k_n^{T_j})e^{rT_j}$ and the undiscounted boundary call value
$\Pi_j:=e^{rT_j}C(k_n^{T_j},T_j)$. The three-anchor tail is then determined
by matching tail mass, boundary call value, and junction density:

\begin{align}
    \int_{k_n^{T_j}}^{\infty} f(x)\,dx &= \lambda_j, \label{eq:pl3_mass}\\
    e^{-rT_j}\!\int_{k_n^{T_j}}^{\infty} (x - k_n^{T_j})\,f(x)\,dx &= C(k_n^{T_j}, T_j), \label{eq:pl3_price}\\
    f(k_n^{T_j+}) = a_j\,\mu_j^{-\alpha_j} &= D_0^{T_j}. \label{eq:pl3_density}
\end{align}

The system~\eqref{eq:pl3_mass}--\eqref{eq:pl3_density} has a unique solution with $\alpha_j>2$, $\mu_j>0$, and $a_j>0$ if and only if this inequality holds:
\begin{align}\label{eq:pl3_admissibility}
    D_0^{T_j}\,\Pi_j > \lambda_j^2.
\end{align}

This condition has a direct tail interpretation. Since $\Pi_j/\lambda_j$ is the average excess distance of the residual right-tail mass beyond $k_n^{T_j}$, condition~\eqref{eq:pl3_admissibility} is equivalent to $D_0^{T_j}>\lambda_j/(\Pi_j/\lambda_j)$. Thus the junction density must be large enough relative to the residual tail mass per unit average excess distance. It rules out the case where the interior density at the last strike is too small to support a positive, decreasing shifted power-law tail with the prescribed mass and boundary call value.

When~\eqref{eq:pl3_admissibility} holds, the unique solution
$(\alpha_j,\mu_j,a_j)$ of
\eqref{eq:pl3_mass}--\eqref{eq:pl3_density} is
\begin{align}\label{eq:pl3_closed_form}
    \alpha_j
    = \frac{2D_0^{T_j}\Pi_j-\lambda_j^2}{D_0^{T_j}\Pi_j-\lambda_j^2},
    \qquad
    \mu_j = \frac{(\alpha_j-1)\lambda_j}{D_0^{T_j}},
    \qquad
    a_j = D_0^{T_j}\mu_j^{\alpha_j}.
\end{align}

The corresponding closed-form call price for $K \geq k_n^{T_j}$ is
\begin{align}\label{eq:pl3_call_price}
    C(K, T_j) = e^{-rT_j}\,\frac{a_j\,(K - k_n^{T_j} + \mu_j)^{2-\alpha_j}}{(\alpha_j - 1)(\alpha_j - 2)}.
\end{align}

\begin{remark}[PL$_2$ as a special case of PL$_3$]\label{rmk:pl3_includes_pl2}
    Substituting $\mu_j = k_n^{T_j}$ into~\eqref{eq:pl3_closed_form} forces $D_0^{T_j} = \lambda_j(\alpha_j - 1)/k_n^{T_j}$, after which the system~\eqref{eq:pl3_mass}--\eqref{eq:pl3_density} reduces to the two-anchor matching of Section~\ref{sec:powerlaw_tail} and~\eqref{eq:alpha_match} is recovered.  Thus the two-anchor tail is the constrained special case $\mu_j = k_n^{T_j}$ of the three-anchor family.
\end{remark}

The recovered parameters define a candidate PL$_3$ tail whenever
\eqref{eq:pl3_admissibility} holds. The proposition below verifies that this
tail preserves the core properties of the two-anchor construction: price
matching, slope matching, tail mass, positivity, monotonicity, and finite first
moment. Its additional benefit is density matching at the junction strike.

\begin{proposition}[Density-anchored power-law tail properties]\label{prop:powerlaw_density}
  Let $k_n^{T_j} > 0$, $C(k_n^{T_j}, T_j) > 0$,
    $C'(k_n^{T_j}) \in (-e^{-rT_j}, 0)$, so $\lambda_j \in (0, 1)$.
    Suppose $D_0^{T_j} > 0$ and~\eqref{eq:pl3_admissibility} holds. Then
    the density~\eqref{eq:shifted_pl_family} with parameters~\eqref{eq:pl3_closed_form} satisfies:
    \begin{enumerate}
        \item \textbf{Price matching:} $e^{-rT_j}\!\int_{k_n^{T_j}}^{\infty}(x - k_n^{T_j})\,f(x)\,dx = C(k_n^{T_j}, T_j)$.
        \item \textbf{Derivative matching:} $-e^{-rT_j}\!\int_{k_n^{T_j}}^{\infty} f(x)\,dx = C'(k_n^{T_j})$.
        \item \textbf{Density matching:} $f(k_n^{T_j+}) = D_0^{T_j}$, i.e.\ $C''$ is continuous across $k_n^{T_j}$.
        \item \textbf{Tail mass:} $\int_{k_n^{T_j}}^{\infty} f(x)\,dx = \lambda_j \in (0, 1)$.
        \item \textbf{Non-negativity and monotonicity:} $f(x) > 0$ and $f'(x) < 0$ for all $x > k_n^{T_j}$.
        \item \textbf{Finite first moment:} $\int_{k_n^{T_j}}^{\infty} x\,f(x)\,dx < \infty$, since $\alpha_j > 2$.
    \end{enumerate}
\end{proposition}
Proof can be found in Appendix~\ref{sec:app_powerlaw_density}.

The density-anchored construction above is a per-maturity refinement. Calendar
freedom, however, is a pairwise condition across maturities. When all maturities
use PL$_2$, Lemma~\ref{lemma:powerlaw_calendar} reduces the right-tail calendar
check to the ordering of the tail indices. Once PL$_2$ and PL$_3$ coexist,
the shift parameter $\mu_j$ also enters the comparison. We therefore state the
next lemma in the unified shifted power-law form. For PL$_2$, $\mu_j=k_n^{T_j}$;
for PL$_3$, $\mu_j$ is given by~\eqref{eq:pl3_closed_form}.

\begin{lemma}[Mixed-tail serial calendar]\label{lemma:mixed_tail_calendar}
    Let $T_j < T_{j+1}$ with $\Delta T_j := T_{j+1} - T_j > 0$ and $\gamma_j := e^{(r-q)\Delta T_j}$, and let the right tails at the two maturities both be in the unified shifted-PL family~\eqref{eq:shifted_pl_family} with parameters $(a_j, \alpha_j, \mu_j)$ and $(a_{j+1}, \alpha_{j+1}, \mu_{j+1})$, $\alpha_j, \alpha_{j+1} > 2$.  Assume:
    \begin{center}
    \begin{tabular}{@{}p{0.50\textwidth}p{0.42\textwidth}@{}}
        \textnormal{(i)} forward moneyness alignment $k_n^{T_j} = \gamma_j^{-1}\,k_n^{T_{j+1}}$; &
        \textnormal{(ii)} $\alpha_{j+1} \leq \alpha_j$; \\[2pt]
        \textnormal{(iii)} $(\alpha_j - 2)\,\mu_{j+1} \geq (\alpha_{j+1} - 2)\,\gamma_j\,\mu_j$; &
        \textnormal{(iv)} $C(k_n^{T_{j+1}}, T_{j+1}) \geq e^{-q\Delta T_j}\,C(k_n^{T_j}, T_j)$. \\
    \end{tabular}
    \end{center}
    Then the pointwise calendar inequality
    \begin{align}
        C(K, T_{j+1}) \geq e^{-q\Delta T_j}\,C(K^{T_j}, T_j), \qquad K^{T_j} := e^{-(r-q)\Delta T_j}\,K,
    \end{align}
    holds for all $K \geq k_n^{T_{j+1}}$.
\end{lemma}
Proof can be found in Appendix~\ref{sec:app_mixed_tail_calendar}.

The lemma serves as a calendar filter for optional PL$_3$ upgrades.
The global slope program has already chosen the right-boundary slopes so that
the all-PL$_2$ tail sequence satisfies Lemma~\ref{lemma:powerlaw_calendar}.
If a maturity is replaced by PL$_3$, the shift parameter $\mu_j$ and the
recovered tail index change the pairwise tail comparison. Lemma~\ref{lemma:mixed_tail_calendar}
gives sufficient conditions under which this replacement remains calendar-free
with the adjacent maturity. The anchor condition~\textup{(iv)} is inherited
from the constructed interior curve, while conditions~\textup{(ii)}--\textup{(iii)}
are checked for the affected adjacent pair.

For a fixed adjacent pair, condition~\textup{(iii)} becomes a scalar check once
the tail choices at $T_j$ and $T_{j+1}$ are known. This gives the following checks:
\begin{center}
\begin{tabular}{@{}ll@{}}
    \hline
    Tail choice at $(T_j, T_{j+1})$ & Reduced condition~\textup{(iii)} \\
    \hline
    (PL$_2$, PL$_2$) & collapses to condition~\textup{(ii)} and recovers Lemma~\ref{lemma:powerlaw_calendar} \\
    (PL$_2$, PL$_3$) & $\mu_{j+1} \geq \dfrac{\alpha_{j+1} - 2}{\alpha_j - 2}\,k_n^{T_{j+1}}$ \\
    (PL$_3$, PL$_2$) & $\mu_j \leq \dfrac{\alpha_j - 2}{\alpha_{j+1} - 2}\,k_n^{T_j}$ \\
    (PL$_3$, PL$_3$) & $(\alpha_j - 2)\,\mu_{j+1} \geq (\alpha_{j+1} - 2)\,\gamma_j\,\mu_j$ \\
    \hline
\end{tabular}
\end{center}

The selection step is therefore conservative. We first keep PL$_2$ as the
baseline at every maturity. A PL$_3$ candidate is committed only if
\eqref{eq:pl3_admissibility} holds and the mixed-tail calendar checks pass for
the adjacent pair affected by the change. Otherwise the maturity remains PL$_2$.
This preserves right-tail calendar freedom because the baseline already
satisfies Lemma~\ref{lemma:powerlaw_calendar}.  Algorithm~\ref{alg:right_tail_select}
summarizes this upgrading procedure.

\begin{algorithm}[ht]
    \caption{Right-tail upgrading procedure}\label{alg:right_tail_select}
    \begin{algorithmic}[1]
        \Require Body outputs $\{h_3^{T_j}\}_{j=1}^m$; finalized boundary data
        $\{C(k_n^{T_j},T_j), C'(k_n^{T_j})\}_{j=1}^m$; optional cap
        $\alpha_{\max}>2$.
        \State Initialize every maturity with the PL$_2$ tail from~\eqref{eq:alpha_match}.
        \For{$j=1$ to $m$}
            \State Set $D_0^{T_j}:=e^{rT_j}h_3^{T_j}$,
            $\lambda_j:=-C'(k_n^{T_j})e^{rT_j}$, and
            $\Pi_j:=e^{rT_j}C(k_n^{T_j},T_j)$.
            \If{$D_0^{T_j}\Pi_j>\lambda_j^2$}
                \State Compute the PL$_3$ candidate
                $(\alpha_j^{(3)},\mu_j^{(3)},a_j^{(3)})$ from~\eqref{eq:pl3_closed_form}.
                \If{$\alpha_j^{(3)}\le \alpha_{\max}$}
                    \State Tentatively replace the PL$_2$ tail at $T_j$ by this PL$_3$ tail.
                    \If{conditions~\textup{(ii)}--\textup{(iii)} of
                    Lemma~\ref{lemma:mixed_tail_calendar} hold for every adjacent pair affected by this replacement}
                        \State Commit the PL$_3$ tail at $T_j$.
                    \Else
                        \State Keep the PL$_2$ tail at $T_j$.
                    \EndIf
                \EndIf
            \EndIf
        \EndFor
        \State \Return the selected right-tail family and parameters at all maturities.
    \end{algorithmic}
\end{algorithm}

\subsection{Assembling the Multi-Maturity Construction}\label{sec:assemble}

The full construction combines three ingredients developed above. First, a
pre-processing stage estimates the per-strike slopes of
Section~\ref{sec:Ai-Lagrange} and selects the tail boundary slopes through the
global program of Section~\ref{sec:global_slope}. Second, the interior of each
maturity is built interval by interval: the two-piece construction is retained
when its first-piece curvature already meets the calendar floor
$h_1^{\rm cal}$, and the three-piece optimization~\eqref{eq:three-piece-opt}
is solved otherwise. Third, the tails are attached. The left region uses the
power-law left tail. The right tail uses PL$_2$ as the baseline and may be
upgraded to the density-anchored PL$_3$ tail through
Algorithm~\ref{alg:right_tail_select}.

The algorithm for the multi-maturity case applies these steps maturity by maturity, and is summarized in Algorithm \ref{alg:multi-mat-construction}. The outer
loop is sequential, because the calendar floor at $T_j$ is inherited from the
constructed curve at $T_{j-1}$. Conditional on this floor, the interior
intervals at a given maturity are independent, so the inner loop over intervals
can run in parallel.

\begin{algorithm}
    \caption{Multi-maturity Marginal Construction}\label{alg:multi-mat-construction}
    \begin{algorithmic}[1]
        \Require Maturities $T_1 < T_2 < \cdots < T_m$ with input call prices $\{C^{T_j}_{k_l}\}$ at strike grids $0 = k_0^{T_j} < k_1^{T_j} < \cdots < k_{n_j}^{T_j}$.
        \State Pre-process.  Compute per-strike slope estimates $\{\widehat{C'}(k_l^{T_j})\}$ via Section~\ref{sec:Ai-Lagrange} and select tail-feasible boundary slopes via Section~\ref{sec:global_slope}.
        \State Construct $T_1$ via the two-piece solution~\eqref{eq:two-piece-opt} on each interior interval and build the tails.
        \For{$j = 2$ to $m$}
        \State Construct the left region $[0, k_1^{T_j}]$ as a power-law left tail (Proposition~\ref{prop:left_pow}); on a degenerate left segment, take the linear segment given by Lemma~\ref{lemma:base-case}.
        \For{each interior interval $[k_l^{T_j}, k_{l+1}^{T_j}]$, $l = 1, \ldots, n_j - 1$}
        \State Compute the per-interval data $(\Delta k_l, A_l, P_l, \rho)$ and the cross-maturity slack $\delta_{\min, l}$.
        \If{the natural two-piece curvature satisfies $h_1^{\textrm{two-piece}} \ge h_1^{\rm cal}$}
        \State Retain the two-piece solution.
        \Else
        \State Solve the per-interval three-piece optimization~\eqref{eq:three-piece-opt}.
        \EndIf
        \EndFor
        \State Construct the right tail at $T_j$ via Section~\ref{sec:powerlaw_tail}.
        \EndFor
        \State \Return $\{C(\cdot, T_j)\}_{j=1}^{m}$.
    \end{algorithmic}
\end{algorithm}

\section{Numerical Experiments}\label{sec:numerical}

This section reports three experiments. Section~\ref{sec:ssvi_single} uses a single-maturity SSVI Power-Law benchmark to test reconstruction from finitely many observed quotes. It checks exact quote matching at the quoted strikes, off-grid price accuracy, and the effect of the right-tail extension from Section~\ref{sec:density_anchored_tail}. Section~\ref{sec:ssvi_multi} applies the same construction to a multi-maturity SSVI benchmark and checks exact fit, butterfly arbitrage-freeness, calendar consistency, and runtime. Section~\ref{sec:sp500_numerical} uses end-of-day S\&P~500 quotes, with SANOS \citep{sanos2025} serving as an independent upstream smoother that supplies the arbitrage-free input grid.

Across these experiments, the construction itself is unchanged. The only difference is how the input grid is obtained: in the SSVI studies it is sampled from a closed-form synthetic market surface, while in the S\&P~500 study it is generated from market quotes by an upstream smoother. This setup tests whether the proposed procedure can turn finite arbitrage-free price grids into marginal laws that preserve exact quote consistency, maintain the required no-arbitrage properties, and remain close in shape to the corresponding benchmark or upstream reference.

\subsection{Single Maturity---SSVI Power-Law Convergence}\label{sec:ssvi_single}

We begin with a single maturity so that the per-interval body construction and the closed-form tails can be examined without cross-maturity coupling. The benchmark is the SSVI Power-Law parametrization of \citet{gatheral2014arbitrage} with parameters $\eta = 2$, $\gamma = 0.2$, $\rho = 0.5$, $\theta_t = 0.4\,t$, spot price $S_0 = 1$, risk-free rate $r = 0.03$, dividend yield $q = 0.01$, and maturity $T = 1$. The practitioner observes $n$ market quotes uniformly spaced in moneyness on $[0.5,1.5]$, all generated by the SSVI surface and therefore arbitrage-free by construction. We feed these $n$ quotes into the single-maturity construction and compare the reconstructed call-price curve against the SSVI benchmark only at unobserved strikes. Specifically, for each value of $n$, we evaluate the reconstruction at $1000$ evenly spaced off-grid moneyness points in $[0.5,1.5]$, excluding the observed quotes used as inputs.

At each evaluation point $k$, the relative price error is $\epsilon_{\rm rel}(k)=|C_{\rm rec}(k)-C_{\rm SSVI}(k)|/C_{\rm SSVI}(k)$. We summarize the errors by the average relative error $\epsilon_{\rm avg}$ and the maximum relative error $\epsilon_{\max}$ over these off-grid evaluation points.

Figure~\ref{fig:ssvi_relerr} reports $\epsilon_{\rm avg}$ and $\epsilon_{\max}$ as the number of observed quotes increases. Even with only $10$ observed quotes, $\epsilon_{\rm avg}$ is already below $5\times 10^{-5}$, and by $20$ quotes $\epsilon_{\max}$ is below $3\times 10^{-5}$. Since the observed quotes are excluded from the evaluation set, the plotted errors measure only off-grid reconstruction accuracy, while the input quotes are matched exactly by construction.

Figure~\ref{fig:ssvi_density_price_trio} gives a direct density-and-price comparison with the SSVI benchmark, and shows that the benchmark shape is already recovered accurately on a moderate quote grid and very closely by $n=50$. Table~\ref{tab:ssvi_runtime} reports the per-fit runtime as a function of $n$ at a single maturity. The growth is approximately linear in $n$, consistent with the $O(n)$ closed-form arithmetic of Section~\ref{sec:single}: a single-maturity fit takes about $2$~ms at $n=50$ and remains under $8$~ms even at $n=200$.

\begin{figure}[ht]
    \centering
    \begin{minipage}[c]{0.58\linewidth}
        \centering
        \includegraphics[width=\linewidth]{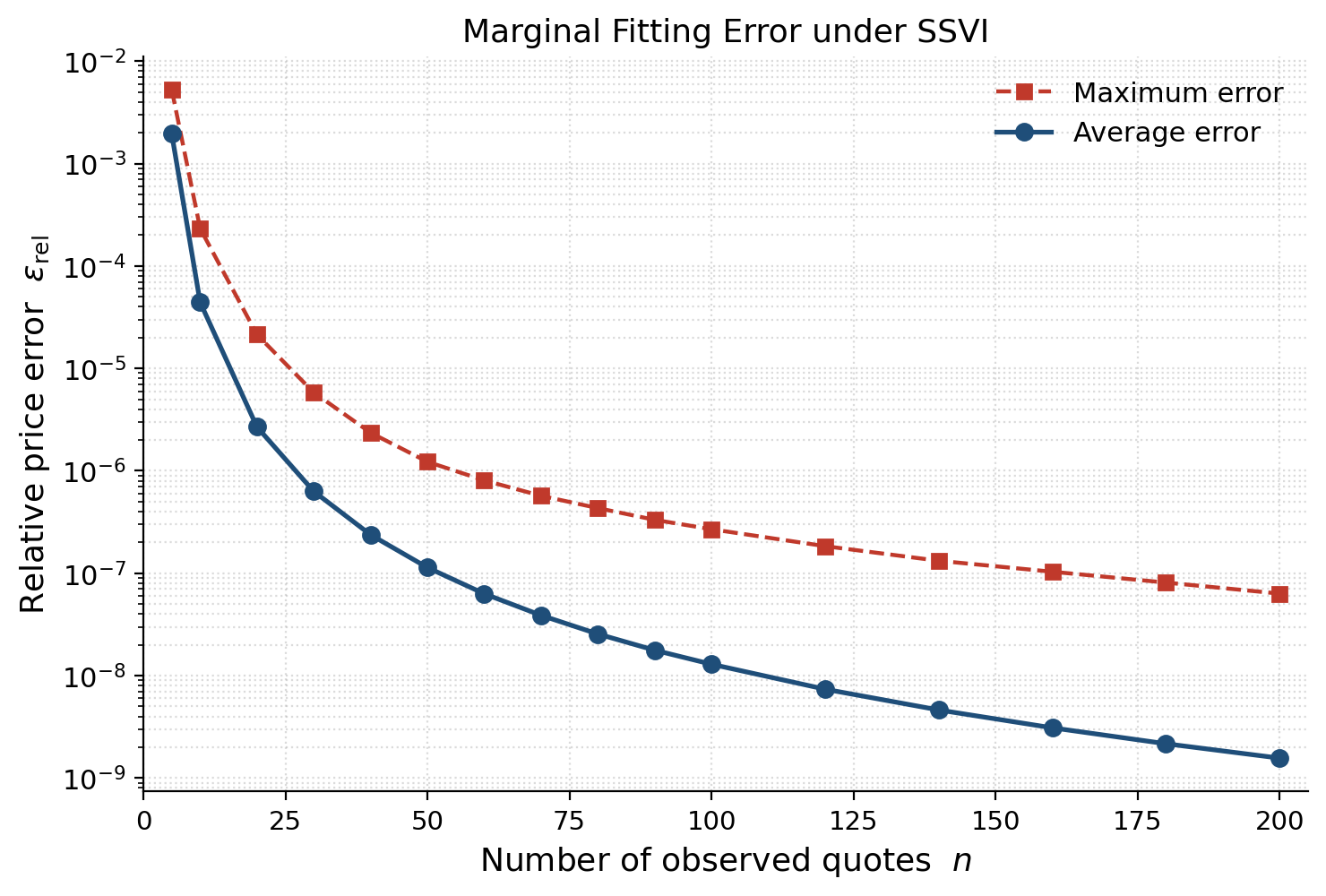}
\caption{Single-maturity SSVI benchmark accuracy. The figure reports the average relative price error $\epsilon_{\rm avg}$ and the maximum relative price error $\epsilon_{\max}$ of the reconstructed call price as the number of observed quotes $n$ increases. Errors are evaluated at $1000$ off-grid moneyness points in $[0.5,1.5]$, excluding the observed quotes used as inputs.}

        \label{fig:ssvi_relerr}
    \end{minipage}\hfill
    \begin{minipage}[c]{0.38\linewidth}
        \centering
        \begin{tabular}{cc}
            \toprule
            $n$ strikes & Per-fit runtime \\
            \midrule
            $5$    & $0.36$ ms \\
            $10$   & $0.51$ ms \\
            $20$   & $0.84$ ms \\
            $50$   & $1.88$ ms \\
            $200$  & $7.14$ ms \\
            \bottomrule
        \end{tabular}
\captionof{table}{Per-fit runtime for the single-maturity SSVI benchmark, averaged over $100$ timing runs for each value of $n$.}
        \label{tab:ssvi_runtime}
    \end{minipage}
\end{figure}

\begin{figure}[ht]
    \centering
    \includegraphics[width=.31\textwidth]{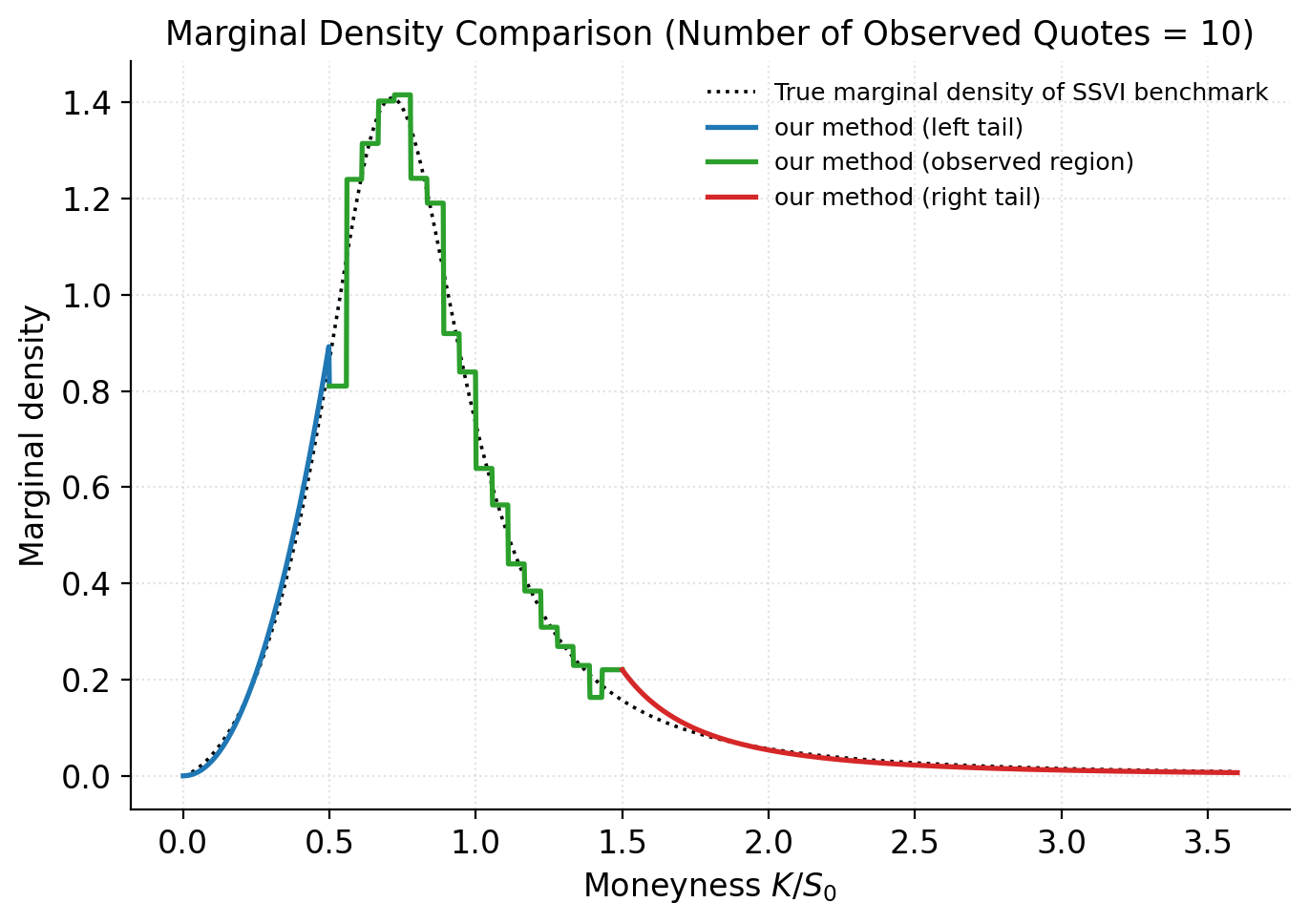}\hfill
    \includegraphics[width=.31\textwidth]{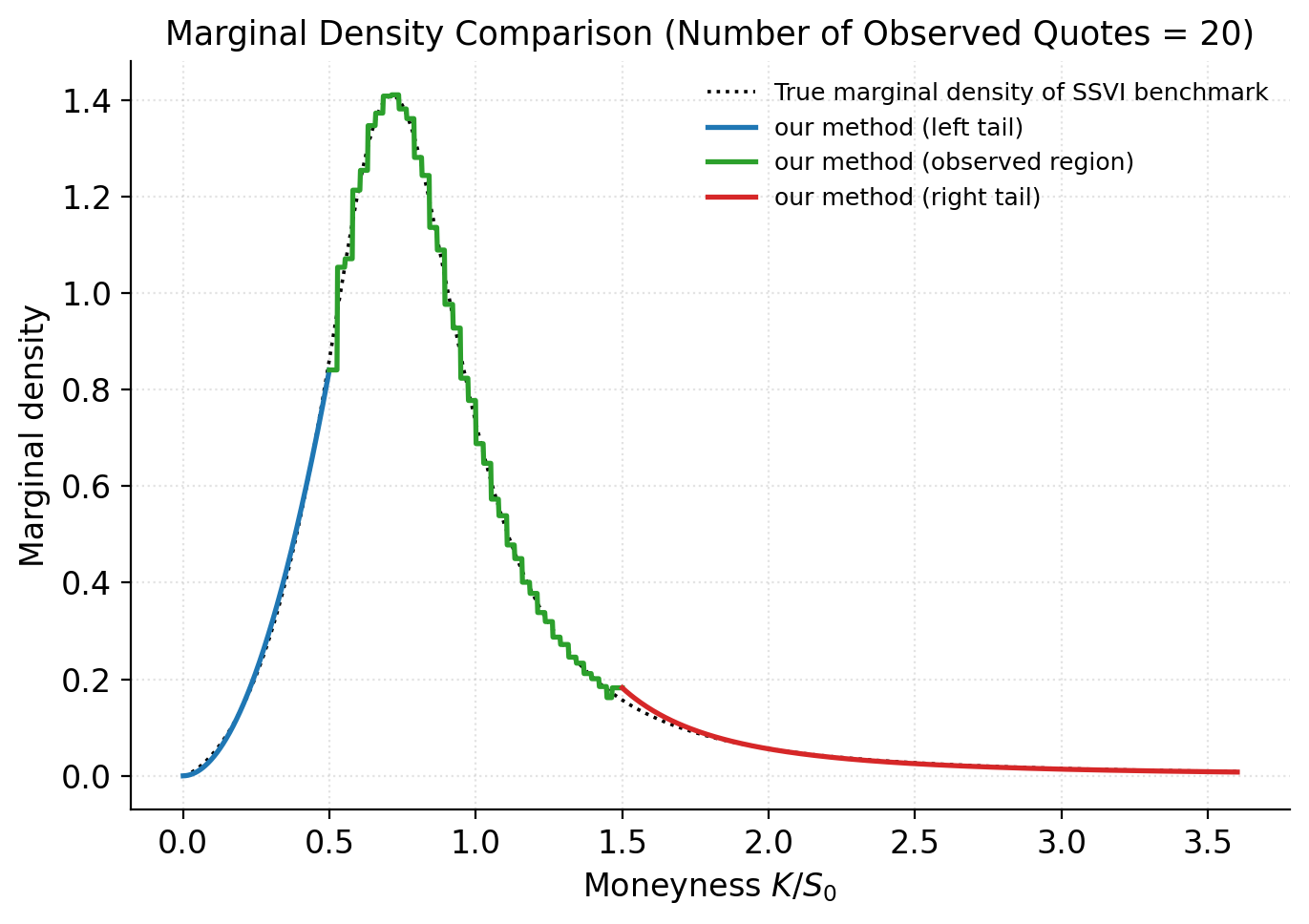}\hfill
    \includegraphics[width=.31\textwidth]{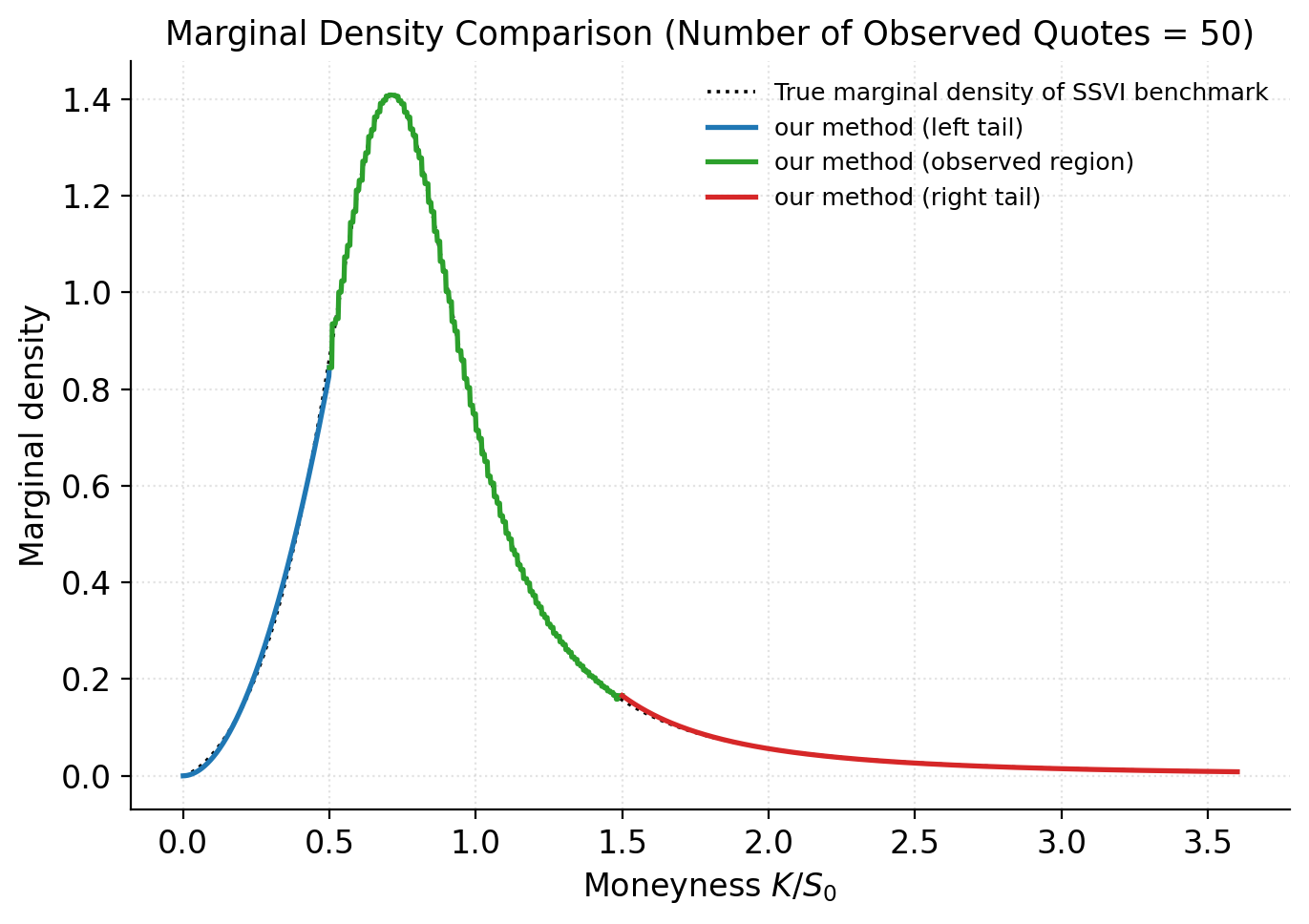}

    \vspace{0.4em}

    \includegraphics[width=.31\textwidth]{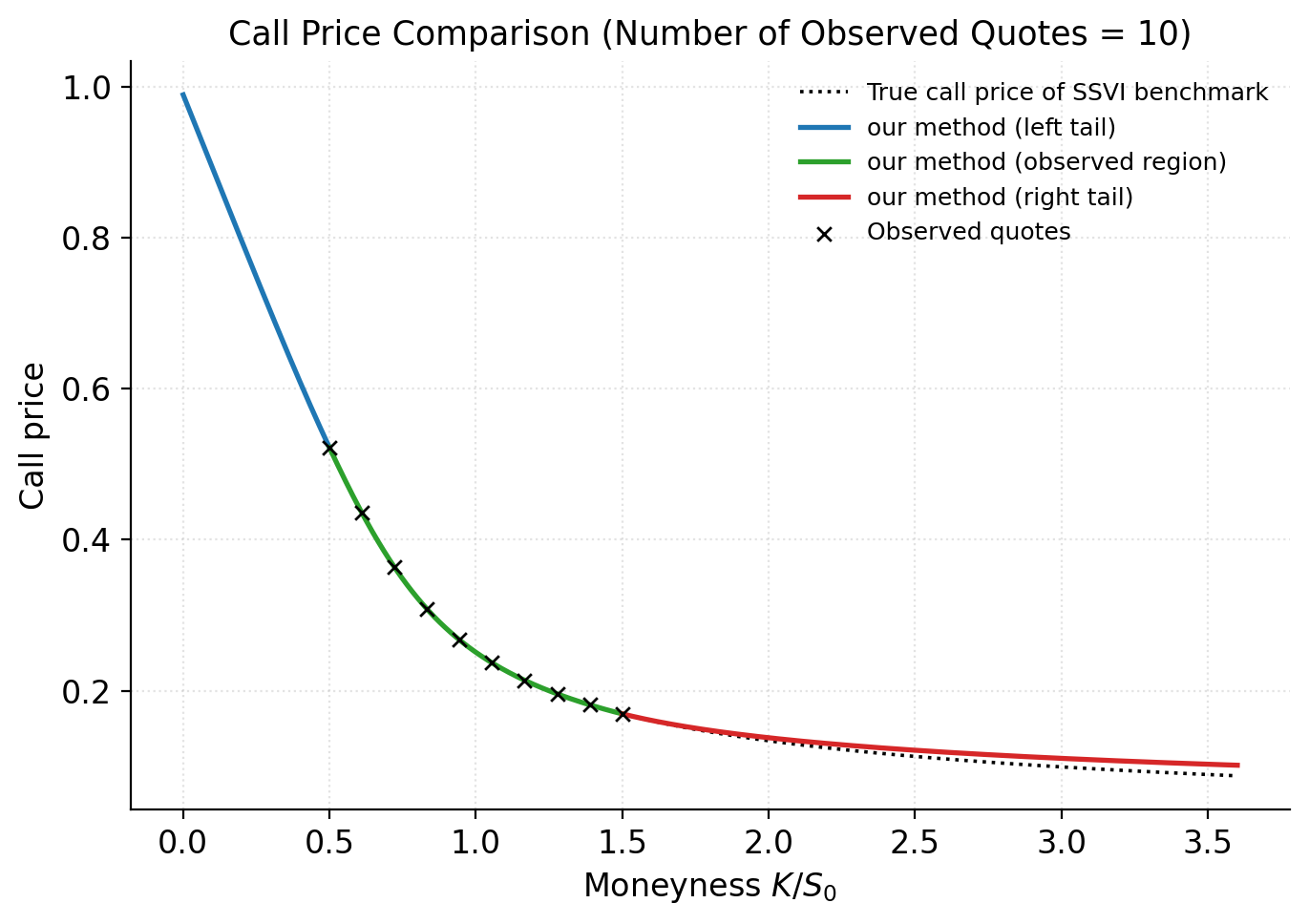}\hfill
    \includegraphics[width=.31\textwidth]{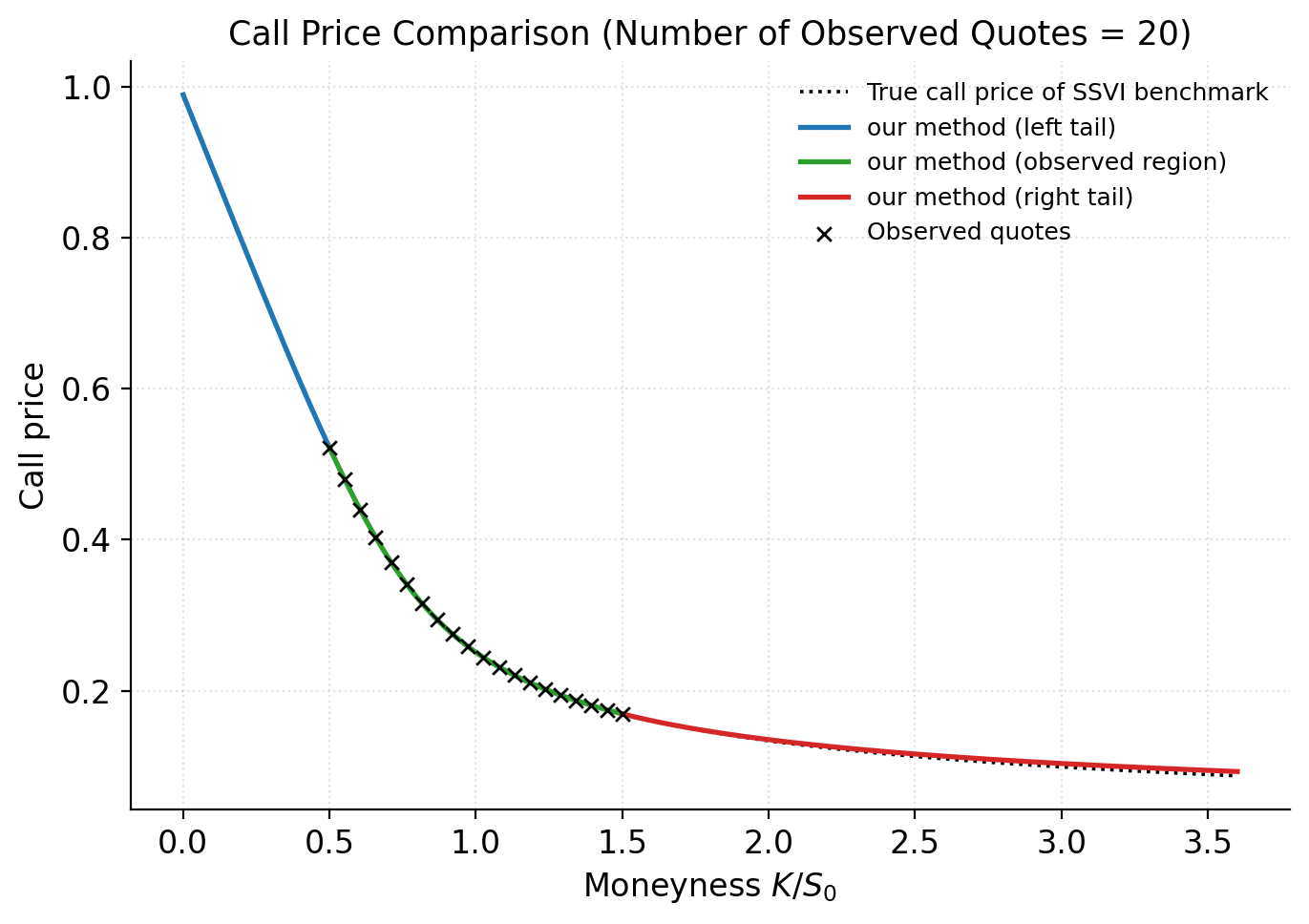}\hfill
    \includegraphics[width=.31\textwidth]{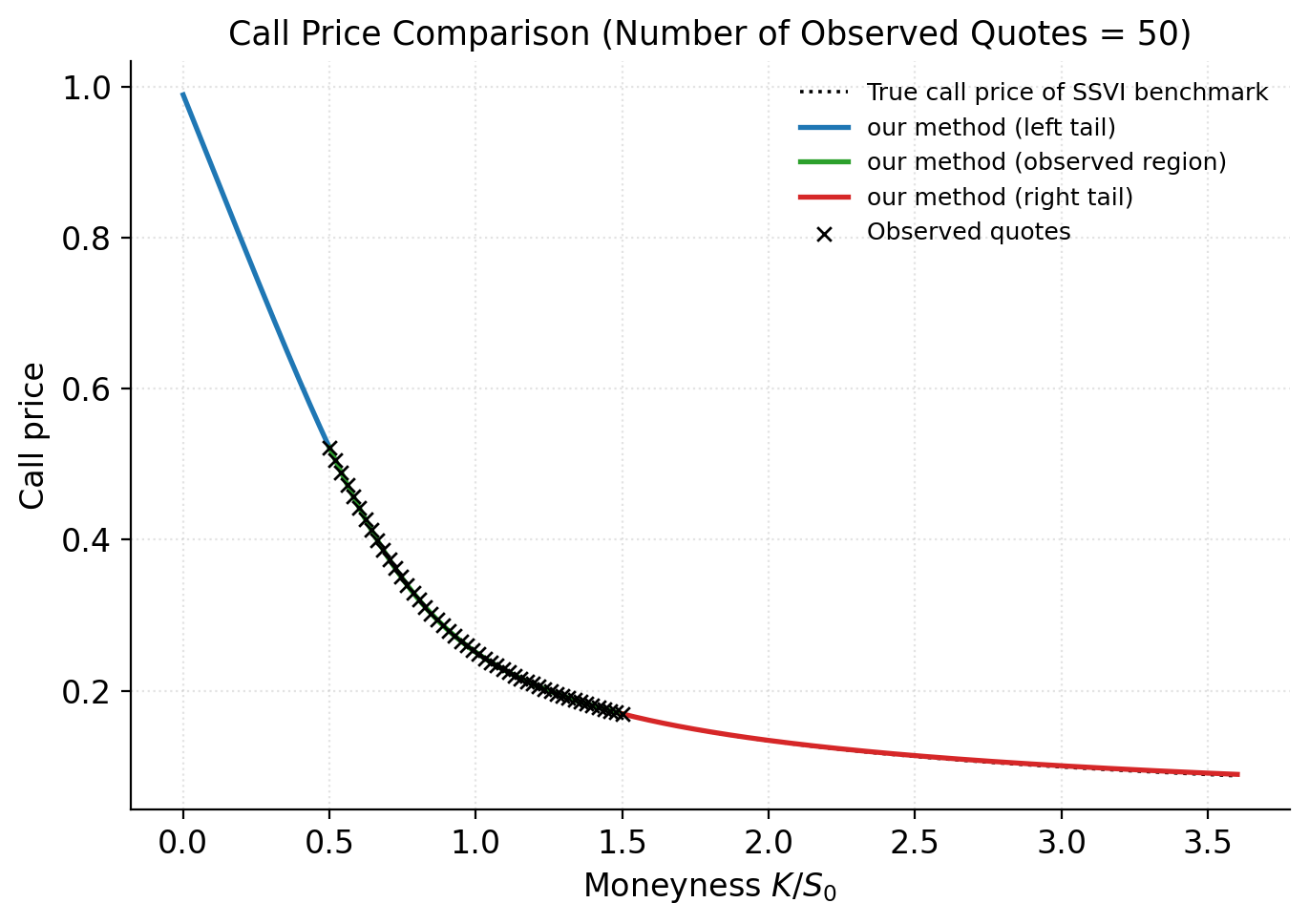}

    \caption{Single-maturity SSVI benchmark comparison. The top row shows the reconstructed marginal density, and the bottom row shows the reconstructed call price. Columns correspond to $n=10$, $20$, and $50$ observed strikes. Solid curves are the reconstruction, dotted curves are the SSVI benchmark, and crosses in the call-price panels mark the observed quotes matched exactly by construction.}
    \label{fig:ssvi_density_price_trio}
\end{figure}

\subsubsection{Performance of the Density-Anchored Right Tail}\label{sec:ssvi_tail_compare}

We next assess the performance of the density-anchored right-tail extension of Section~\ref{sec:density_anchored_tail} on the SSVI Power-Law benchmark with $T=1$, using $n=100$ observed strikes uniformly spaced in moneyness on $[0.5,1.5]$. PL$_2$ and PL$_3$ are built from the same quotes, interior construction, and boundary data at the largest observed strike $k_n$. The baseline PL$_2$ tail matches the boundary price and slope, while PL$_3$ additionally imposes $C''(k_n^+)=C''(k_n^-)$ when the scalar condition in~\eqref{eq:pl3_admissibility} holds. Therefore, the two reconstructions coincide on the observed strike range, and differ only in the right-tail extension.

For the SSVI input in this experiment, the PL$_3$ condition is satisfied. Relative to PL$_2$, the density-anchored tail improves the curvature-level fit at the junction: PL$_2$ leaves a roughly $22\%$ downward jump between $f(k_n^-)$ and $f(k_n^+)$, whereas PL$_3$ enforces continuity of the marginal density at $k_n$. Figure~\ref{fig:ssvi_tail_compare} shows this comparison in the full marginal-density plot and in a zoom around the junction strike.

This comparison illustrates the role of the PL$_3$ upgrade when the stronger local condition is satisfied. The condition is natural in non-extreme tail settings, where the last observed strike is not too deep in the right tail and the junction density is not too small relative to the residual tail mass. In such cases, PL$_3$ is preferred: it preserves the same boundary price, slope, residual mass, and no-arbitrage properties as PL$_2$, while giving a better curvature-level fit at the right-tail junction.

\begin{figure}[ht]
    \centering
    \includegraphics[width=0.48\linewidth]{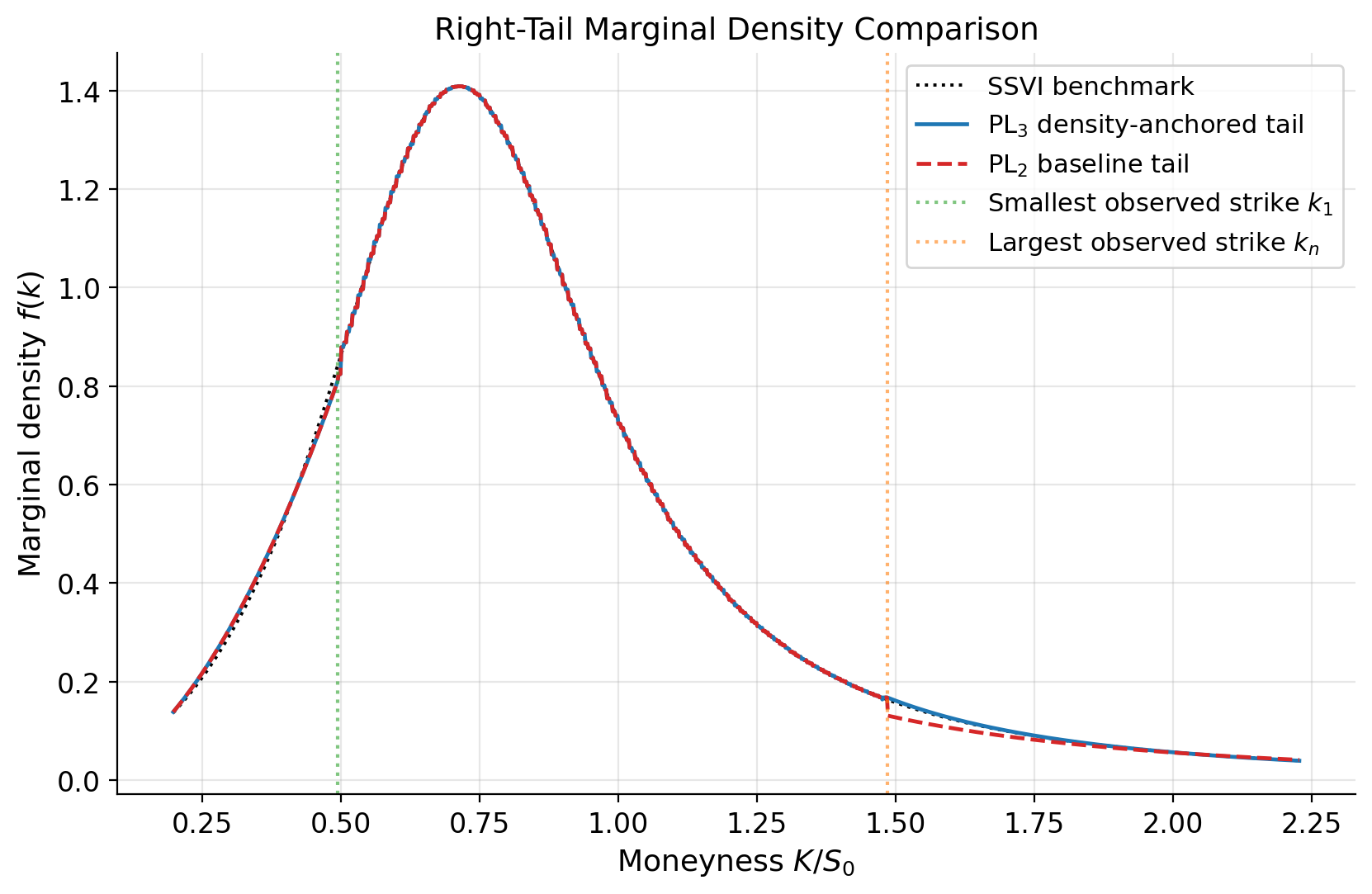}\hfill
    \includegraphics[width=0.48\linewidth]{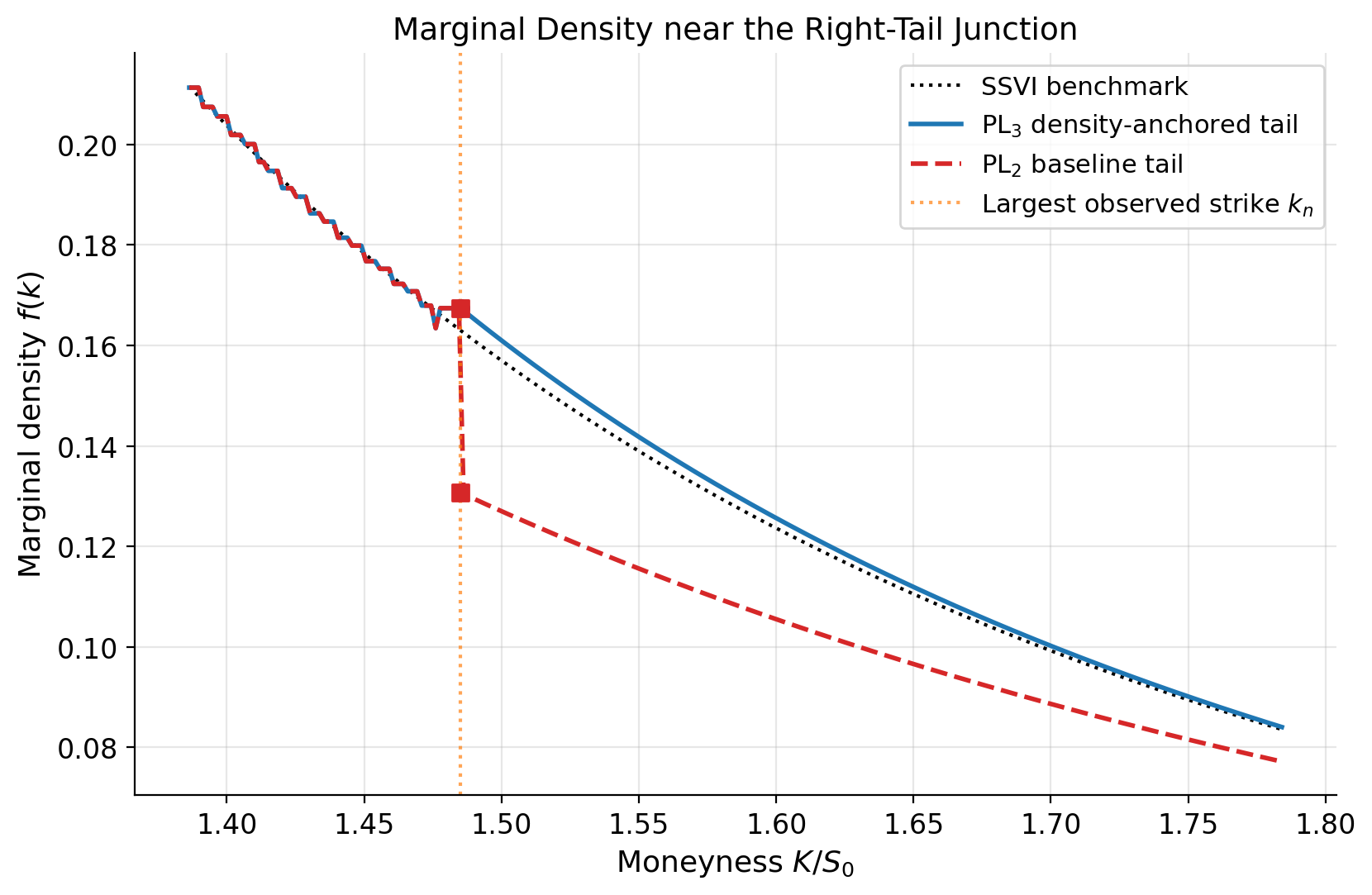}
    \caption{Performance of the density-anchored right tail in the single-maturity SSVI experiment. The left panel shows the full marginal-density comparison, and the right panel zooms in around the largest observed strike $k_n$. PL$_2$ and PL$_3$ use the same observed quotes, the same interior construction, and the same boundary data at $k_n$; they differ only in the right-tail extension. PL$_3$ removes the density jump left by PL$_2$ and gives a continuous curvature-level continuation at the junction.}
    \label{fig:ssvi_tail_compare}
\end{figure}

\subsection{Multiple Maturities---SSVI Power-Law Surface}\label{sec:ssvi_multi}

We next apply the construction to the multi-maturity SSVI Power-Law benchmark. The maturities are $T\in\{0.5,0.6,\ldots,1.5\}$, giving eleven slices with spacing $\Delta T=0.1$. Each maturity has $n=100$ observed strikes, and the strike grids are forward-aligned as in Section~\ref{sec:global_slope}. The construction matches every input quote exactly. The butterfly check passes on all eleven slices, and the observed-strike calendar condition passes on all $1000$ cross-maturity checks. The end-to-end runtime is $16.282$ seconds, corresponding to about $14.802$ ms per input strike.

Figure~\ref{fig:ssvi_multi_surface} shows the reconstructed implied-volatility surface and the reconstructed marginal-density surface across the eleven maturities. Both surfaces agree closely with the closed-form SSVI benchmark at the resolution of the plots.

We examine the calendar slack $C(k,T_{j+1})-L(k)$, where $L(k)=e^{-q\Delta T_j}C(e^{-(r-q)\Delta T_j}k,T_j)$ is the calendar lower bound inherited from $T_j$. Since the assembled call curves are closed-form on each piece, the slack can be evaluated directly across moneyness. Figure~\ref{fig:ssvi_calendar_slack} plots this slack for two representative adjacent maturity pairs, $(T_j,T_{j+1})=(0.9,1.0)$ and $(1.3,1.4)$, as examples of this guaranteed non-negative function.

\begin{figure}[ht]
    \centering
    \includegraphics[width=0.48\linewidth]{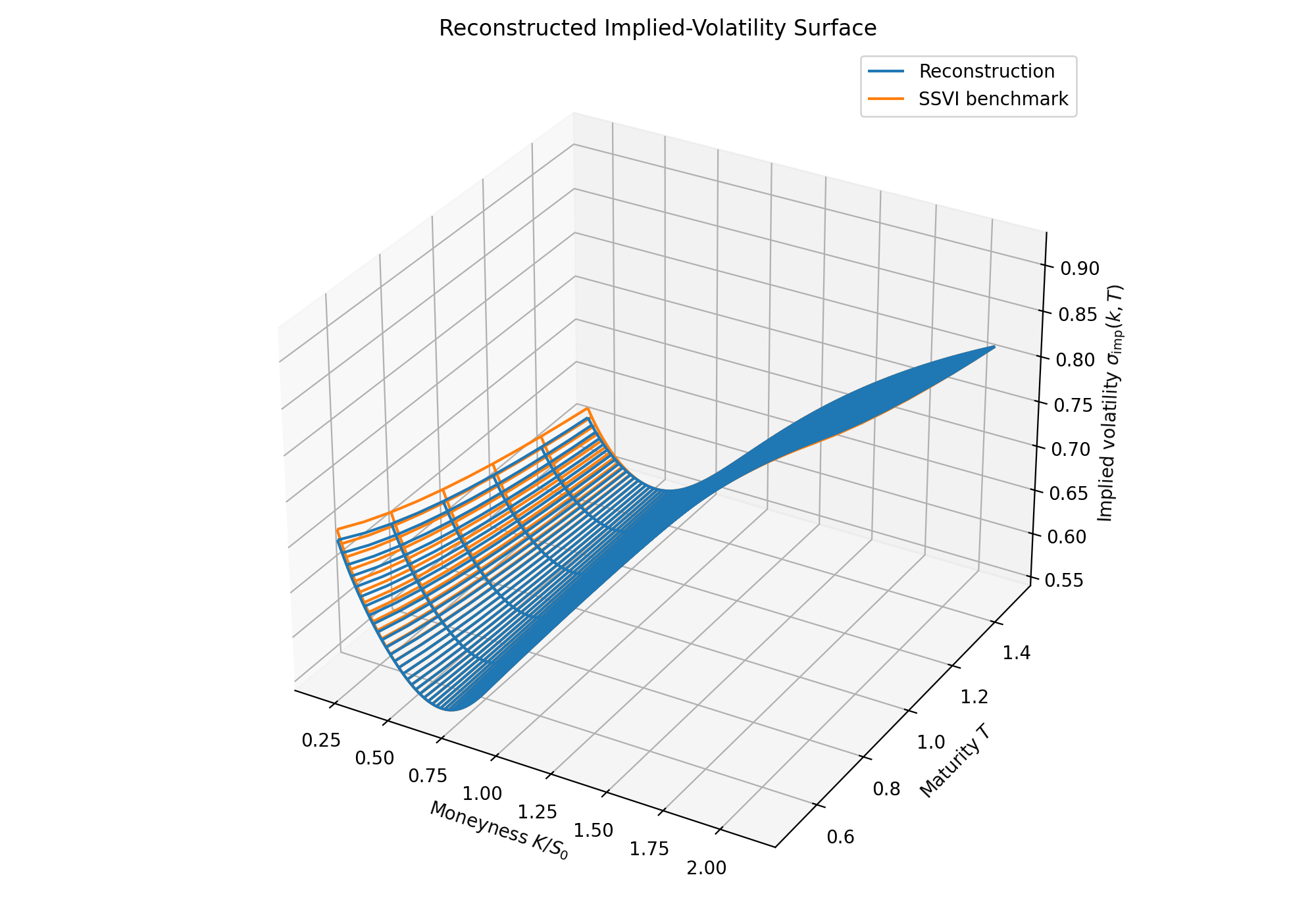}\hfill
    \includegraphics[width=0.48\linewidth]{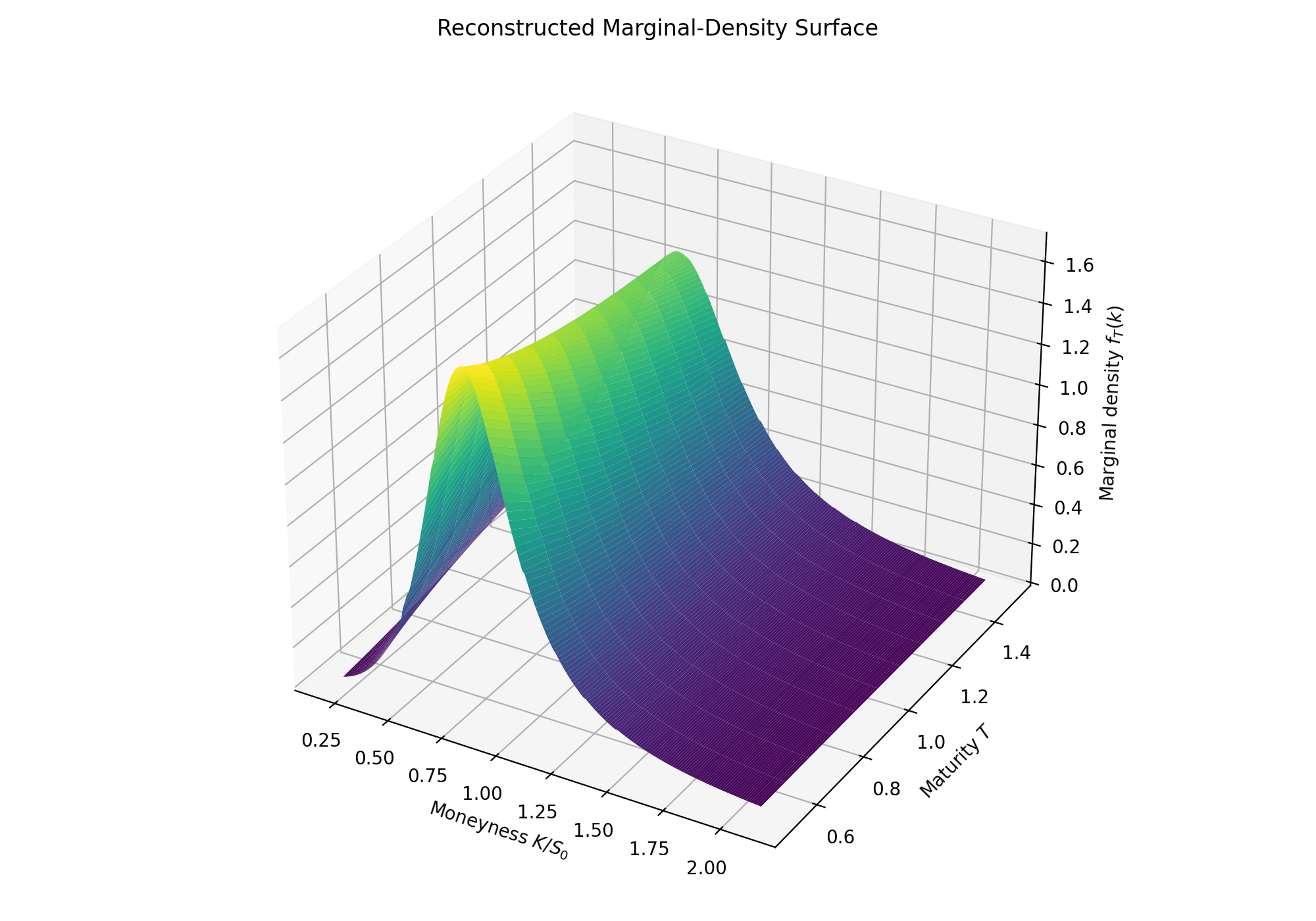}
    \caption{Multi-maturity SSVI benchmark comparison. The left panel shows the reconstructed implied-volatility surface, and the right panel shows the reconstructed marginal-density surface across the eleven maturities.}
    \label{fig:ssvi_multi_surface}
\end{figure}

\begin{figure}[ht]
    \centering
    \includegraphics[width=0.48\linewidth]{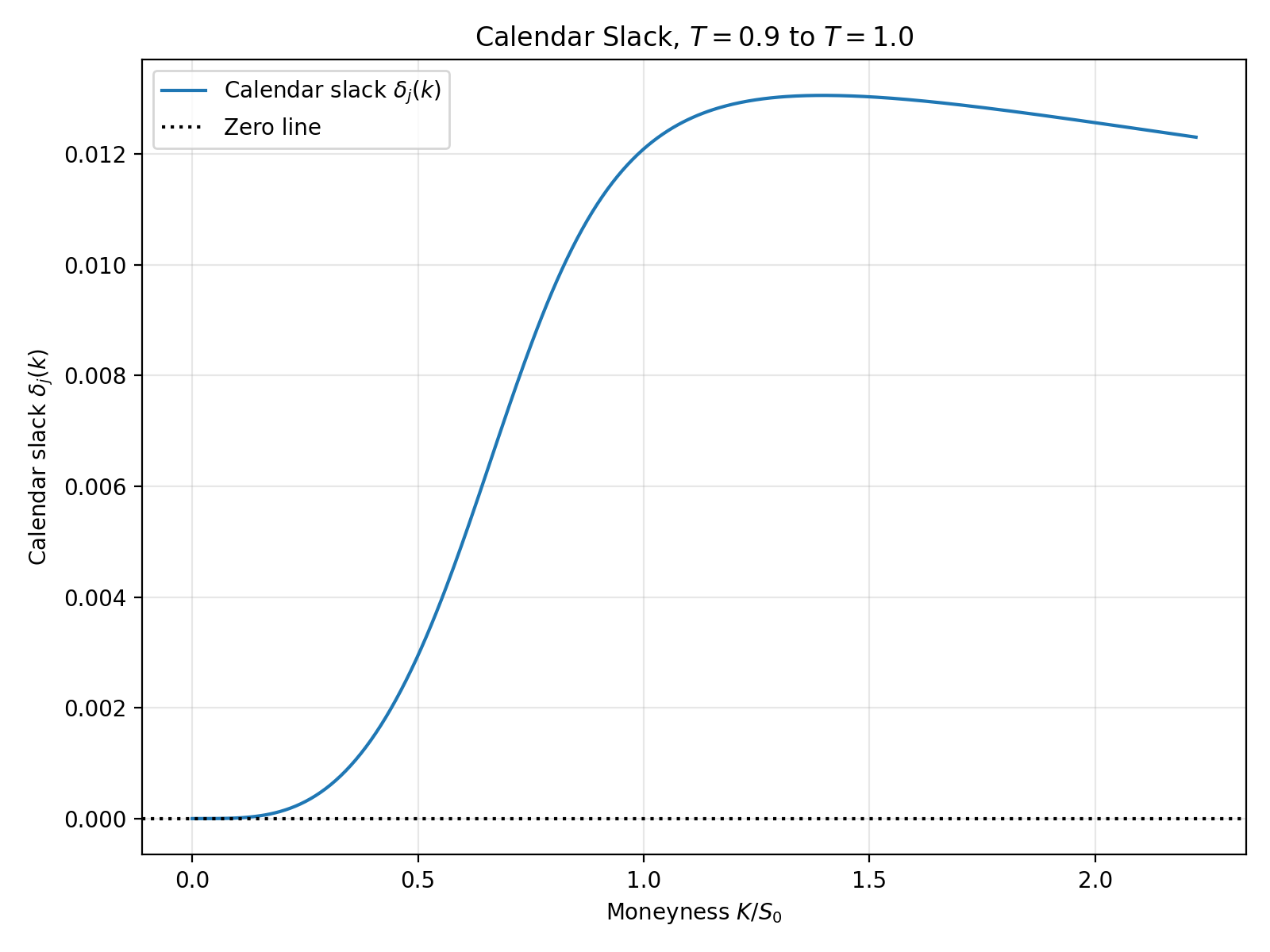}\hfill
    \includegraphics[width=0.48\linewidth]{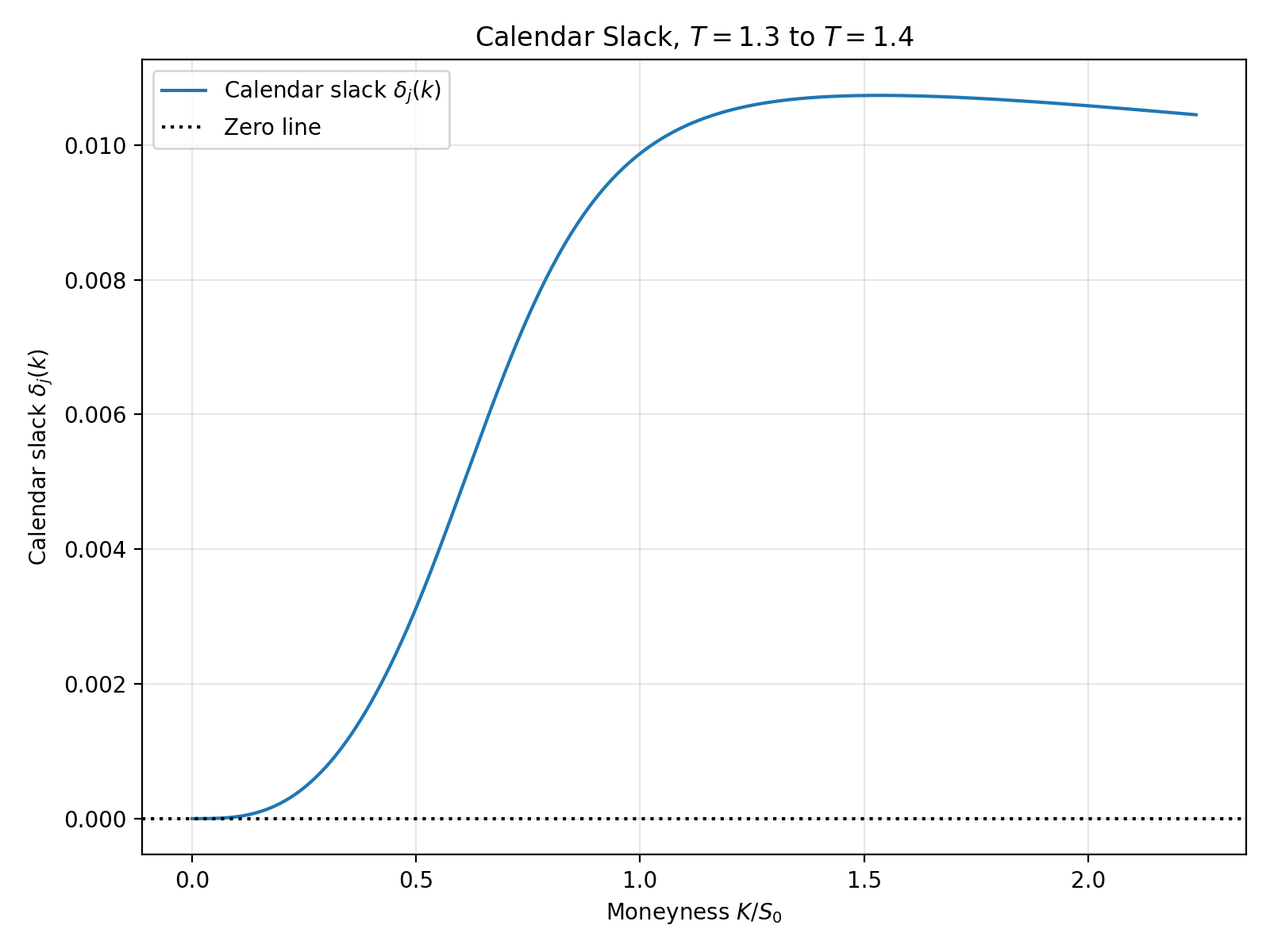}
\caption{Calendar-slack diagnostics for the multi-maturity SSVI experiment. The left panel shows the calendar slack for $(T_j,T_{j+1})=(0.9,1.0)$, and the right panel shows it for $(T_j,T_{j+1})=(1.3,1.4)$. In both panels, the plotted slack remains non-negative, showing pointwise calendar consistency.}
    \label{fig:ssvi_calendar_slack}
\end{figure}

\subsection{S\&P~500 Market Case via Upstream Smoothing}\label{sec:sp500_numerical}

We test the construction on end-of-day European call quotes on the S\&P~500 index on 2022-02-22, with spot $S_0=4304.76$ and quotes from OptionMetrics. The raw market quotes contain local butterfly violations caused by bid--ask noise and a long right tail of sub-tick option values. We therefore use SANOS \citep{sanos2025} as an upstream de-arbitrage step. SANOS supplies a forward-aligned static-arbitrage-free discrete call-price grid. Only this discrete grid is passed to our construction; the SANOS density curves are used below only as an external shape reference.

This setup tests the role of our method as a marginal-construction layer after an upstream smoother. Starting from the de-arbitraged discrete grid, our construction produces full marginal laws on $[0,\infty)$, with explicit tails, unit-mass control, exact matching of the input prices, and pointwise no-arbitrage checks across maturities. We use six maturities, $\{87,143,206,297,360,388\}$ days. At each maturity, we discard strikes whose discounted time value is below half a tick, since these far-tail quotes carry little information about the marginal shape and are highly sensitive to quote discreteness. Table~\ref{tab:spx_summary} reports the resulting input size at each maturity.

\begin{table}[ht]
    \centering
    \begin{tabular}{lcccccc}
        \toprule
        Maturity (days) & $87$ & $143$ & $206$ & $297$ & $360$ & $388$ \\
        \midrule
        Input strikes & $133$ & $138$ & $142$ & $148$ & $150$ & $150$ \\
        \bottomrule
    \end{tabular}
\caption{S\&P~500 2022-02-22 market experiment. The table reports the number of input strikes passed to our construction after the half-tick time-value filter.}
    \label{tab:spx_summary}
\end{table}

The reconstructed marginals satisfy the unit-mass condition up to machine precision, as expected from the boundary slope conditions and the closed-form tail completion. A per-segment calendar-slack diagnostic detects no violation across the five adjacent maturity pairs. The end-to-end runtime is $12.741$ seconds. With $861$ input strikes in total, this corresponds to $14.798$ ms per input strike, essentially the same per-strike cost as in the multi-maturity SSVI experiment.

Figure~\ref{fig:spx_density_comparison} gives two views of the market reconstruction. The left panel compares our reconstructed marginal densities with the SANOS reference densities at three representative maturities chosen to display different density shapes. The close agreement shows that, starting only from the de-arbitraged discrete price grid, our construction recovers marginal-density shapes that are highly consistent with the upstream reference.

The right panel shows our reconstructed marginal densities across all six maturities. It displays the term structure of the market-implied marginal densities produced by our construction. In particular, the reconstruction preserves the bimodal market-implied shape and tracks its evolution across maturities.

\begin{figure}[ht]
    \centering
    \includegraphics[width=0.46\linewidth]{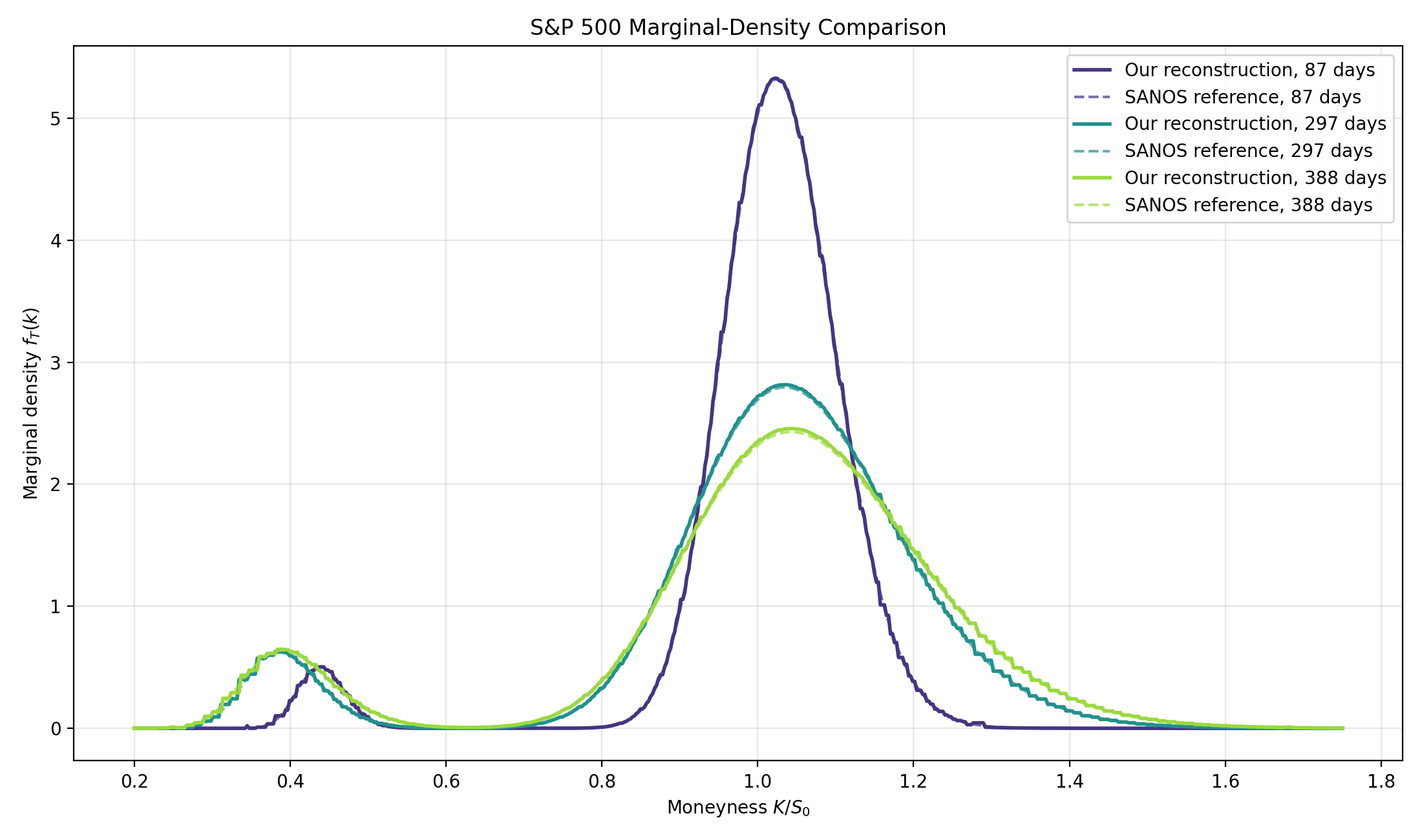}\hfill
    \includegraphics[width=0.54\linewidth]{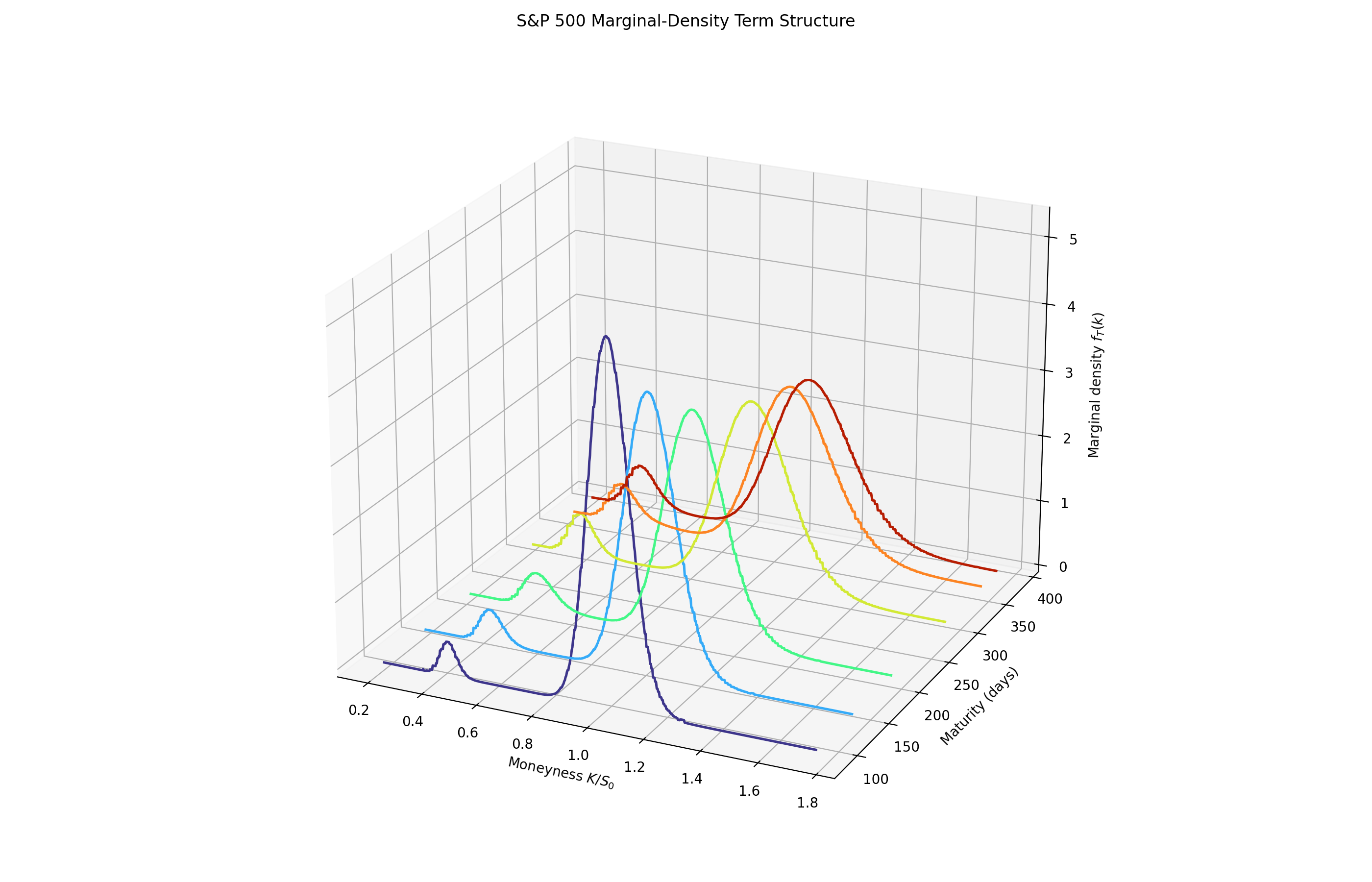}
\caption{S\&P~500 market experiment. The left panel compares our reconstructed marginal densities with the SANOS reference densities at three representative maturities. The right panel shows our reconstructed marginal densities across all six maturities in $(K/S_0,T)$ space, illustrating the term structure of the market-implied marginal densities.}
    \label{fig:spx_density_comparison}
\end{figure}

\section{Conclusion}\label{sec:conclusion}

This paper shows that arbitrage-free call-price grids provide a direct channel to option-implied risk-neutral marginal laws. The construction proposed here uses this channel directly. Starting from discrete arbitrage-free call prices, it constructs full marginal laws on $[0,\infty)$ without requiring an intermediate step like implied-volatility parametrization. The resulting marginals match the input quotes exactly, have unit mass, and satisfy butterfly and calendar no-arbitrage at the distribution level.

The construction uses the slope and curvature structure of call prices directly. At each maturity, selected slopes determine the interval masses, and non-negative piecewise-constant curvature gives closed-form interval curves that match the input prices. Power-law tails complete the unquoted regions while matching the boundary data and residual mass. Across maturities, the construction proceeds sequentially: the earlier curve defines the calendar lower bound, the three-piece form is used only when the two-piece construction is not enough, and tail consistency is enforced through low-dimensional slope selection.

The numerical experiments confirm the intended role of the method as a price-to-marginal layer. The SSVI tests verify exact repricing, off-grid accuracy, marginal-density reconstruction, calendar consistency, and fast runtime. The S\&P~500 market experiment shows that, once an upstream smoother supplies a discrete arbitrage-free input grid, the same construction produces full marginal laws with unit mass up to machine precision, no detected calendar violation across the selected maturities, and marginal-density shapes consistent with the upstream reference.

\newpage
\section{Appendix}

We collect all proofs here, organized in the order the corresponding results appear in the main text.

\subsection{Proof of Theorem~\ref{thm:two-piece}: Two-Piece Interval Construction}\label{sec:app_lemma3}

Fix an interval $[k_i, k_{i+1}]$ with butterfly-free data $(C_{k_i}, C'_{k_i}, C_{k_{i+1}}, C'_{k_{i+1}})$, and write $A_i = C'_{k_{i+1}} - C'_{k_i} \geq 0$, $\Delta k_i = k_{i+1} - k_i$, and $\Delta C = C_{k_{i+1}} - C_{k_i} - C'_{k_i}\Delta k_i$.  By convexity $\Delta C \geq 0$ (the call lies above its tangent), and by the chord bound $\Delta C \leq A_i\Delta k_i$.

If $A_i = 0$ the data are affine on $[k_i, k_{i+1}]$, and the single linear segment $(\beta_i, h_i) = (1, 0)$ matches them; both parts hold trivially.  Assume $A_i > 0$ for the remainder, so that $\Delta C > 0$ and $r_i := \Delta C/(A_i\Delta k_i) \in (0, 1]$.

With the second-piece curvature $c_i = \frac{A_i - h_i\beta_i\Delta k_i}{(1-\beta_i)\Delta k_i}$, the endpoint price condition $C(k_{i+1}) = C_{k_{i+1}}$ is linear in $h_i$ and gives
$$
h_i(\beta_i) = \frac{2\Delta C - A_i(1-\beta_i)\Delta k_i}{\beta_i(\Delta k_i)^2}.
$$

\paragraph{Part (a): the feasible set is non-empty.}
The sign of $h_i(\beta_i)$ is set by its numerator, so $h_i \geq 0 \iff \beta_i \geq 1 - 2r_i$.  Substituting $h_i(\beta_i)$ into $c_i$ gives $A_i - h_i\beta_i\Delta k_i = A_i(2 - 2r_i - \beta_i)$, hence $c_i \geq 0 \iff \beta_i \leq 2 - 2r_i$.  The two constraints confine $\beta_i$ to $[\beta_{\mathrm{LB}}, \beta_{\mathrm{UB}}]$ with
$$
\beta_{\mathrm{LB}} = \max(0,\, 1 - 2r_i), \qquad \beta_{\mathrm{UB}} = \min(1,\, 2 - 2r_i).
$$
For $r_i \in (0, 1)$ this range is non-empty, since $1 - 2r_i < 2 - 2r_i$ and $\beta_{\mathrm{LB}} < 1$; for $r_i = 1$ the one-piece point $\beta_i = 1$, where the second piece has zero width and $h_i = 2\Delta C/(\Delta k_i)^2 \geq 0$, is feasible.  This proves part~(a).

\paragraph{Part (b): explicit closed-form solution.}
For $\beta_i \in (0,1)$ a direct computation gives $c_i - h_i = \frac{A_i(1 - 2r_i)}{\beta_i(1-\beta_i)\Delta k_i}$, so the curvature-jump objective satisfies $F(\beta_i) = (c_i - h_i)^2 \propto [\beta_i(1-\beta_i)]^{-2}$ and is minimized at $\beta_i = \tfrac{1}{2}$.  The optimal feasible choice is therefore
\begin{enumerate}[label=(\roman*)]
    \item $r_i \in [\tfrac{1}{4}, \tfrac{3}{4}]$, so that $\tfrac{1}{2} \in [\beta_{\mathrm{LB}}, \beta_{\mathrm{UB}}]$: set $\beta_i = \tfrac{1}{2}$; when $r_i = \tfrac{1}{2}$ this gives $h_i = c_i = A_i/\Delta k_i$ (uniform curvature, $F = 0$).
    \item $r_i \in (0, \tfrac{1}{4})$, so that $\beta_{\mathrm{LB}} > \tfrac{1}{2}$: set $\beta_i = \beta_{\mathrm{LB}} = 1 - 2r_i$, giving $h_i = 0$ and $c_i = A_i/(2r_i\Delta k_i) > 0$.
    \item $r_i \in (\tfrac{3}{4}, 1)$, so that $\beta_{\mathrm{UB}} < \tfrac{1}{2}$: set $\beta_i = \beta_{\mathrm{UB}} = 2 - 2r_i$, giving $c_i = 0$ and $h_i > 0$.
    \item $r_i = 1$: set $\beta_i = 1$, the one-piece solution $h_i = 2A_i/\Delta k_i > 0$, with $F = 0$ by convention.
\end{enumerate}
In each case $h_i$ follows from the price-matching formula above, and both curvatures are non-negative by the bracket of part~(a), so $C''(k) \geq 0$ on $[k_i, k_{i+1}]$.  Every quantity involves only elementary arithmetic, so the per-interval solution is obtained in $O(1)$ operations.  \qed

\subsection{Proof of Proposition~\ref{prop:powerlaw}: Power-Law Tail Properties}\label{sec:app_powerlaw}

Throughout this proof, fix $k_n > 0$ and let
\[
f(x) = \lambda\,\frac{\alpha-1}{k_n}\left(\frac{k_n}{x}\right)^{\!\alpha},
\qquad x \geq k_n,
\]
with $\lambda := -C'(k_n)\,e^{rT}\in (0,1]$ defined from the observed slope and
$\alpha>2$ to be determined.  All integrals below converge precisely because
$\alpha>2$, and the following elementary identity is used repeatedly: for
any $p>1$,
\begin{equation}\label{eq:tail-integral}
    \int_{k_n}^{\infty} x^{-p}\,dx
    = \left[-\frac{x^{1-p}}{p-1}\right]_{k_n}^{\infty}
    = \frac{k_n^{\,1-p}}{p-1},
\end{equation}
where the upper-limit term vanishes since $1-p<0$.

\begin{proof}[Proof of Proposition~\ref{prop:powerlaw}]
    We verify Properties 1--6 in an order that exposes their logical dependence.
    
    \smallskip
    \paragraph{Property 4 (Non-negativity and strict monotonicity).}
    Writing $f(x) = \lambda(\alpha-1)\,k_n^{\alpha-1}\,x^{-\alpha}$, all factors
    ($\lambda$, $\alpha-1$, $k_n^{\alpha-1}$, $x^{-\alpha}$) are strictly
    positive on $x>k_n>0$, so $f(x)>0$.  Differentiating in $x$,
    \[
    f'(x) = -\alpha\,\lambda(\alpha-1)\,k_n^{\alpha-1}\,x^{-\alpha-1} < 0,
    \]
    hence $f$ is strictly decreasing on $(k_n,\infty)$.
    
    \smallskip
    \paragraph{Property 3 (Tail mass).}
    Pulling the constants out of the integral and applying~\eqref{eq:tail-integral}
    with $p = \alpha > 1$,
    \begin{equation}\label{eq:tail-mass-derivation}
        \int_{k_n}^{\infty} f(x)\,dx
        = \lambda\,\frac{\alpha-1}{k_n}\,k_n^{\alpha}
        \int_{k_n}^{\infty} x^{-\alpha}\,dx
        = \lambda\,\frac{\alpha-1}{k_n}\,k_n^{\alpha}\cdot
        \frac{k_n^{\,1-\alpha}}{\alpha-1}
        = \lambda.
    \end{equation}
    By the no-arbitrage identity $-C'(k_n) = e^{-rT}\,\mathbb{Q}(S_T>k_n)$, the
    tail mass is $\lambda = -C'(k_n)\,e^{rT} = \mathbb{Q}(S_T>k_n)$, strictly
    below one since $k_n > 0$.
    
    \smallskip
    \paragraph{Property 2 (Derivative matching).}
    Multiplying~\eqref{eq:tail-mass-derivation} by $-e^{-rT}$ and using the
    defining identity $\lambda = -C'(k_n)\,e^{rT}$,
    \[
    -e^{-rT}\int_{k_n}^{\infty} f(x)\,dx
    = -e^{-rT}\lambda = C'(k_n).
    \]
    This is the special case at $K=k_n$ of the general Breeden--Litzenberger
    identity $C'(K) = -e^{-rT}\!\int_K^{\infty} f(x)\,dx$, obtained by
    differentiating $C(K) = e^{-rT}\!\int_K^{\infty}(x-K)\,f(x)\,dx$ under the
    integral sign.
    
    \smallskip
    \paragraph{Property 1 (Price matching).}
    Pulling the constants out as before and applying~\eqref{eq:tail-integral} with
    $p = \alpha-1 > 1$ (which requires $\alpha>2$),
    \begin{equation}\label{eq:first-moment-tail}
        \int_{k_n}^{\infty} x\,f(x)\,dx
        = \lambda\,\frac{\alpha-1}{k_n}\,k_n^{\alpha}\,
        \frac{k_n^{\,2-\alpha}}{\alpha-2}
        = \frac{\lambda(\alpha-1)\,k_n}{\alpha-2}.
    \end{equation}
    Subtracting $k_n$ times the tail-mass identity,
    \[
    \int_{k_n}^{\infty}(x-k_n)\,f(x)\,dx
    = \frac{\lambda(\alpha-1)\,k_n}{\alpha-2} - \lambda\,k_n
    = \frac{\lambda\,k_n}{\alpha-2}.
    \]
    Discounting and equating to $C_{k_n}$ yields
    \[
    e^{-rT}\,\frac{\lambda\,k_n}{\alpha-2} = C_{k_n}
    \quad\Longleftrightarrow\quad
    \alpha = 2 + \frac{\lambda\,k_n\,e^{-rT}}{C_{k_n}},
    \]
    which is precisely formula~\eqref{eq:alpha_match}.  Since
    $\lambda > 0$ and $k_n > 0$, the matched value satisfies $\alpha > 2$
    automatically; no additional constraint on the input data is required.
    
    \smallskip
    \paragraph{Property 5 (Finite first moment).}
    The tail contribution to the forward price is precisely the integral
    in~\eqref{eq:first-moment-tail}, which is finite by virtue of
    $\alpha-2 > 0$.
    
    \smallskip
    \paragraph{Property 6 (Asymptotic boundary slope).}
    Repeating the Property~3 derivation, with the lower limit $k_n$ replaced
    by an arbitrary $K \geq k_n$, gives
    \begin{equation}\label{eq:tail-survival}
        \int_K^{\infty} f(x)\,dx
        = \lambda\,\left(\frac{k_n}{K}\right)^{\!\alpha-1},
    \end{equation}
    hence by the Breeden--Litzenberger identity
    \[
    C'(K) = -e^{-rT}\,\lambda\,\left(\frac{k_n}{K}\right)^{\!\alpha-1}.
    \]
    Since $\alpha-1>1$, $(k_n/K)^{\alpha-1}\to 0$ as $K\to\infty$, so
    $\lim_{K\to\infty}C'(K)=0$.  This is the right-side boundary condition
    invoked in Section~\ref{sec:prelim} to certify that the constructed
    density integrates to one.
    
    \smallskip
    \paragraph{Closed-form call price beyond $k_n$.}
    For $K \geq k_n$, the same argument applied to $\int_K^{\infty}(x-K)f(x)\,dx$
    gives
    \[
    C(K) = e^{-rT}\,\frac{\lambda\,K}{\alpha-2}\left(\frac{k_n}{K}\right)^{\!\alpha-1}
    = e^{-rT}\,\frac{\lambda\,k_n}{\alpha-2}\left(\frac{k_n}{K}\right)^{\!\alpha-2}
    = C_{k_n}\left(\frac{k_n}{K}\right)^{\!\alpha-2},
    \]
    so the call price decays as a power of the moneyness ratio with the same
    exponent $\alpha-2$ that drives the Lee wing slope.
\end{proof}


\subsection{Proof of Lemma~\ref{lemma:powerlaw_calendar}: Power-Law Right-Tail Calendar Condition}\label{sec:app_powerlaw_calendar}

Under the power-law tail, the closed-form call price (see Appendix~\ref{sec:app_powerlaw}) is $C(k, T_j) = C(k_n^{T_j}, T_j)(k_n^{T_j}/k)^{\alpha_j - 2}$ for $k \geq k_n^{T_j}$.

With forward moneyness alignment $k_n^{T_j} = e^{-(r-q)\Delta T_j}k_n^{T_{j+1}}$, let $u := k/k_n^{T_{j+1}} \geq 1$. The forward-aligned strike at $T_j$ satisfies $k^{T_j}/k_n^{T_j} = k/k_n^{T_{j+1}} = u$ (since $k^{T_j} = e^{-(r-q)\Delta T_j}k$ and $k_n^{T_j} = e^{-(r-q)\Delta T_j}k_n^{T_{j+1}}$). The calendar condition becomes:
$$
C(k_n^{T_{j+1}}, T_{j+1})\,u^{-(\alpha_{j+1}-2)} \geq e^{-q\Delta T_j}\,C(k_n^{T_j}, T_j)\,u^{-(\alpha_j - 2)}.
$$
Rearranging:
$$
\frac{C(k_n^{T_{j+1}}, T_{j+1})}{e^{-q\Delta T_j}C(k_n^{T_j}, T_j)} \geq u^{(\alpha_{j+1}-2) - (\alpha_j-2)} = u^{\alpha_{j+1} - \alpha_j}.
$$
The left side is the anchor ratio, which is $\geq 1$ by the anchor-point calendar condition (ensured by the interior construction).  For this to hold for all $u \geq 1$:
\begin{itemize}
    \item If $\alpha_{j+1} > \alpha_j$, the right side $u^{\alpha_{j+1}-\alpha_j} \to \infty$ as $u \to \infty$, so the condition fails.
    \item If $\alpha_{j+1} \leq \alpha_j$, the right side $u^{\alpha_{j+1}-\alpha_j} \leq 1$ for all $u \geq 1$, so the condition holds whenever the anchor ratio $\geq 1$.
\end{itemize}
Therefore, if $\alpha_j$ is non-increasing in $j$ and the anchor-point calendar condition holds, the calendar condition is satisfied for all $u \geq 1$. \qed

\subsection{Proof of Proposition~\ref{prop:powerlaw_density}: Density-Anchored Power-Law Tail Properties}\label{sec:app_powerlaw_density}

Throughout this proof fix $T = T_j$, $k_n = k_n^{T_j}$, and drop the maturity index where unambiguous.  Set $\lambda = -C'(k_n)\,e^{rT}$ and $\Pi = e^{rT}\,C(k_n, T)$.  The family is
\[
f(x) = a\,(x - k_n + \mu)^{-\alpha}, \qquad x \geq k_n,
\]
with admissibility $D_0\,\Pi > \lambda^2$.  We use the elementary substitution $u := x - k_n + \mu$, so as $x$ ranges over $[k_n, \infty)$, $u$ ranges over $[\mu, \infty)$, and $du = dx$.  For any $p > 1$,
\begin{equation}\label{eq:shifted-tail-integral}
    \int_\mu^{\infty} u^{-p}\,du = \frac{\mu^{\,1-p}}{p - 1}.
\end{equation}

\begin{proof}[Proof of Proposition~\ref{prop:powerlaw_density}]
    We verify the six properties in dependency order and then establish the necessity of~\eqref{eq:pl3_admissibility} for solvability of~\eqref{eq:pl3_mass}--\eqref{eq:pl3_density}.
    
    \smallskip
    \paragraph{$\alpha > 2$, $\mu > 0$, $a > 0$ (existence under admissibility).}
    From the $\alpha$ formula in~\eqref{eq:pl3_closed_form},
    \[
    \alpha - 2 = \frac{2 D_0\,\Pi - \lambda^2 - 2(D_0\,\Pi - \lambda^2)}{D_0\,\Pi - \lambda^2} = \frac{\lambda^2}{D_0\,\Pi - \lambda^2}.
    \]
    The numerator $\lambda^2 > 0$, and the denominator $D_0\,\Pi - \lambda^2 > 0$ by admissibility~\eqref{eq:pl3_admissibility}; hence $\alpha - 2 > 0$.  Then $\mu = (\alpha - 1)\,\lambda/D_0 > 0$ and $a = D_0\,\mu^{\alpha} > 0$.
    
    \smallskip
    \paragraph{Property 5 (Non-negativity and monotonicity).}
    With $a > 0$, $\mu > 0$, $\alpha > 2$: $x - k_n + \mu \geq \mu > 0$, so $f(x) > 0$.  Differentiating, $f'(x) = -a\,\alpha\,(x - k_n + \mu)^{-\alpha-1} < 0$.
    
    \smallskip
    \paragraph{Property 4 (Tail mass).}
    By~\eqref{eq:shifted-tail-integral} with $p = \alpha > 1$,
    \[
    \int_{k_n}^{\infty} f(x)\,dx = a\,\int_\mu^{\infty} u^{-\alpha}\,du = \frac{a\,\mu^{1-\alpha}}{\alpha - 1}.
    \]
    Substituting $a = D_0\,\mu^\alpha$ and $\mu = (\alpha - 1)\,\lambda/D_0$:
    \[
    \frac{a\,\mu^{1-\alpha}}{\alpha - 1} = \frac{D_0\,\mu^\alpha\,\mu^{1-\alpha}}{\alpha - 1} = \frac{D_0\,\mu}{\alpha - 1} = \frac{D_0\,(\alpha - 1)\,\lambda/D_0}{\alpha - 1} = \lambda.
    \]
    The bound $\lambda \in (0, 1)$ follows from the hypothesis $C'(k_n) \in (-e^{-rT}, 0)$.
    
    \smallskip
    \paragraph{Property 2 (Derivative matching).}
    By Property 4 and the definition of $\lambda$:
    \[
    -e^{-rT}\!\int_{k_n}^{\infty} f(x)\,dx = -e^{-rT}\,\lambda = -e^{-rT}\,(-C'(k_n)\,e^{rT}) = C'(k_n).
    \]
    
    \smallskip
    \paragraph{Property 3 (Density matching).}
    $f(k_n^+) = a\,\mu^{-\alpha} = D_0\,\mu^\alpha\,\mu^{-\alpha} = D_0$.
    
    \smallskip
    \paragraph{Property 1 (Price matching).}
    Writing $x - k_n = u - \mu$ in the price integral and using~\eqref{eq:shifted-tail-integral} for $p = \alpha - 1 > 1$ and $p = \alpha > 1$:
    \[
    \int_{k_n}^{\infty}(x - k_n)\,f(x)\,dx = a\!\int_\mu^{\infty}(u - \mu)\,u^{-\alpha}\,du = a\,\left[\frac{\mu^{2-\alpha}}{\alpha - 2} - \mu\,\frac{\mu^{1-\alpha}}{\alpha - 1}\right] = \frac{a\,\mu^{2-\alpha}}{(\alpha-1)(\alpha-2)}.
    \]
    Substituting $a\,\mu^{2-\alpha} = D_0\,\mu^2 = D_0\,(\alpha - 1)^2\,\lambda^2/D_0^2 = (\alpha - 1)^2\,\lambda^2/D_0$ gives
    \[
    \int_{k_n}^{\infty}(x - k_n)\,f(x)\,dx = \frac{(\alpha - 1)\,\lambda^2}{D_0\,(\alpha - 2)}.
    \]
    From the $\alpha$ formula, $\alpha - 1 = D_0\,\Pi/(D_0\,\Pi - \lambda^2)$ and $\alpha - 2 = \lambda^2/(D_0\,\Pi - \lambda^2)$, so $(\alpha - 1)/(\alpha - 2) = D_0\,\Pi/\lambda^2$.  Multiplying by $e^{-rT}$ on both sides yields
    \[
    e^{-rT}\!\int_{k_n}^{\infty}(x - k_n)\,f(x)\,dx = e^{-rT}\,\frac{\lambda^2}{D_0} \cdot \frac{D_0\,\Pi}{\lambda^2} = e^{-rT}\,\Pi = C(k_n, T).
    \]
    
    \smallskip
    \paragraph{Property 6 (Finite first moment).}
    $\int x\,f\,dx = \int(x - k_n)\,f\,dx + k_n\,\lambda < \infty$ since both terms are finite by the preceding derivations.
    
    \smallskip
    \paragraph{Converse (necessity of admissibility).}
    Suppose $D_0\,\Pi \leq \lambda^2$ and that $(a, \alpha, \mu)$ with $\alpha > 2$, $\mu > 0$, $a > 0$ satisfy~\eqref{eq:pl3_mass}--\eqref{eq:pl3_density}.  The same algebra as in Property 1 gives $\Pi(\alpha - 2) = (\alpha - 1)\,\lambda^2/D_0$, equivalent to $\alpha\,(D_0\,\Pi - \lambda^2) = 2\,D_0\,\Pi - \lambda^2$.
    
    \smallskip\noindent\textbf{Case A:} $D_0\,\Pi = \lambda^2$.  The right-hand side $2 D_0\Pi - \lambda^2 = \lambda^2 \neq 0$ contradicts the left-hand side being zero, so no solution.
    
    \smallskip\noindent\textbf{Case B:} $D_0\,\Pi < \lambda^2$.  Then $D_0\,\Pi - \lambda^2 < 0$.  Solving for $\alpha$ gives $\alpha = (2 D_0\,\Pi - \lambda^2)/(D_0\,\Pi - \lambda^2)$, hence
    \[
    \alpha - 2 = \frac{\lambda^2}{D_0\,\Pi - \lambda^2} < 0,
    \]
    so $\alpha < 2$, contradicting $\alpha > 2$.
    
    In both cases the system has no solution with $\alpha > 2$, $\mu > 0$, $a > 0$, establishing the necessity.
\end{proof}

\subsection{Proof of Lemma~\ref{lemma:mixed_tail_calendar}: Mixed-Tail Serial Calendar Condition}\label{sec:app_mixed_tail_calendar}

Throughout this proof we use the closed-form call price~\eqref{eq:pl3_call_price} for both maturities (with PL$_2$ corresponding to the $\mu = k_n$ specialization).  Let $\gamma_j := e^{(r-q)\Delta T_j}$ and set $s := K - k_n^{T_{j+1}} \geq 0$.  By forward moneyness alignment $k_n^{T_j} = \gamma_j^{-1}\,k_n^{T_{j+1}}$:
\[
K^{T_j} - k_n^{T_j} = \gamma_j^{-1}\,(K - k_n^{T_{j+1}}) = \gamma_j^{-1}\,s.
\]
Hence the arguments of the closed-form call prices simplify to $s + \mu_{j+1}$ on the left and $\gamma_j^{-1}\,s + \mu_j$ on the right.

\begin{proof}[Proof of Lemma~\ref{lemma:mixed_tail_calendar}]
    Substituting~\eqref{eq:pl3_call_price} into the calendar inequality $C(K, T_{j+1}) \geq e^{-q\Delta T_j}\,C(K^{T_j}, T_j)$:
    \[
    e^{-rT_{j+1}}\,\frac{a_{j+1}\,(s + \mu_{j+1})^{2-\alpha_{j+1}}}{(\alpha_{j+1} - 1)(\alpha_{j+1} - 2)} \geq e^{-q\Delta T_j}\,e^{-rT_j}\,\frac{a_j\,(\gamma_j^{-1}\,s + \mu_j)^{2-\alpha_j}}{(\alpha_j - 1)(\alpha_j - 2)}.
    \]
    Using $\gamma_j^{-1}\,s + \mu_j = \gamma_j^{-1}\,(s + \gamma_j\,\mu_j)$, we have $(\gamma_j^{-1}\,s + \mu_j)^{2-\alpha_j} = \gamma_j^{\alpha_j - 2}\,(s + \gamma_j\,\mu_j)^{2-\alpha_j}$.  Also $e^{-q\Delta T_j}\,e^{-rT_j} = e^{-rT_{j+1}}\,\gamma_j$ since $T_{j+1} = T_j + \Delta T_j$.  After dividing both sides by $e^{-rT_{j+1}}/[(\alpha_{j+1} - 1)(\alpha_{j+1} - 2)]$, the inequality becomes
    \begin{equation}\label{eq:pf_mixed_tail_key}
        a_{j+1}\,(s + \mu_{j+1})^{2-\alpha_{j+1}} \geq \gamma_j^{\alpha_j - 1}\,a_j\,\frac{(\alpha_{j+1} - 1)(\alpha_{j+1} - 2)}{(\alpha_j - 1)(\alpha_j - 2)}\,(s + \gamma_j\,\mu_j)^{2-\alpha_j}.
    \end{equation}
    
    Take logarithms and set $h(s) := \log\bigl(\text{LHS of \eqref{eq:pf_mixed_tail_key}}\bigr) - \log\bigl(\text{RHS of \eqref{eq:pf_mixed_tail_key}}\bigr)$.  Using $(2 - \alpha) = -(\alpha - 2)$, the $s$-dependent part of $h(s)$ is
    \[
    g(s) := (\alpha_j - 2)\,\log(s + \gamma_j\,\mu_j) - (\alpha_{j+1} - 2)\,\log(s + \mu_{j+1}),
    \]
    i.e.\ $h(s) = g(s) + (\text{constant independent of } s)$.  Differentiating,
    \[
    g'(s) = \frac{\alpha_j - 2}{s + \gamma_j\,\mu_j} - \frac{\alpha_{j+1} - 2}{s + \mu_{j+1}}.
    \]
    The sign condition $g'(s) \geq 0$ for all $s \geq 0$ is equivalent to
    \[
    (\alpha_j - 2)(s + \mu_{j+1}) \geq (\alpha_{j+1} - 2)(s + \gamma_j\,\mu_j) \qquad \forall s \geq 0,
    \]
    which rearranges to the linear-in-$s$ statement
    \begin{equation}\label{eq:pf_mixed_tail_linear}
        \underbrace{[\alpha_j - \alpha_{j+1}]}_{\geq\, 0 \text{ by (ii)}}\,s + \underbrace{[(\alpha_j - 2)\,\mu_{j+1} - (\alpha_{j+1} - 2)\,\gamma_j\,\mu_j]}_{\geq\, 0 \text{ by (iii)}} \geq 0.
    \end{equation}
    Since both coefficients are non-negative under (ii) and (iii), the inequality holds for every $s \geq 0$; hence $g$, and therefore $h$, is non-decreasing on $[0, \infty)$.
    
    Next we verify that $h(0) \geq 0$ is equivalent to (iv).  At $s = 0$,
    \[
    \text{LHS of \eqref{eq:pf_mixed_tail_key}}\bigr|_{s=0} = a_{j+1}\,\mu_{j+1}^{2-\alpha_{j+1}},
    \]
    \[
    \text{RHS of \eqref{eq:pf_mixed_tail_key}}\bigr|_{s=0} = \gamma_j^{\alpha_j - 1}\,a_j\,\frac{(\alpha_{j+1} - 1)(\alpha_{j+1} - 2)}{(\alpha_j - 1)(\alpha_j - 2)}\,(\gamma_j\,\mu_j)^{2-\alpha_j}.
    \]
    The factor $\gamma_j^{\alpha_j - 1}\,(\gamma_j\,\mu_j)^{2-\alpha_j} = \gamma_j^{\alpha_j - 1}\,\gamma_j^{2-\alpha_j}\,\mu_j^{2-\alpha_j} = \gamma_j\,\mu_j^{2-\alpha_j}$.  Restoring the previously factored normalization $e^{-rT_{j+1}}/[(\alpha_{j+1} - 1)(\alpha_{j+1} - 2)]$ on both sides, the $s = 0$ inequality reads
    \[
    e^{-rT_{j+1}}\,\frac{a_{j+1}\,\mu_{j+1}^{2-\alpha_{j+1}}}{(\alpha_{j+1} - 1)(\alpha_{j+1} - 2)} \geq e^{-rT_{j+1}}\,\gamma_j\,\frac{a_j\,\mu_j^{2-\alpha_j}}{(\alpha_j - 1)(\alpha_j - 2)}.
    \]
    The left-hand side equals $C(k_n^{T_{j+1}}, T_{j+1})$ by~\eqref{eq:pl3_call_price}.  Using $e^{-rT_{j+1}}\,\gamma_j = e^{-rT_j - q\Delta T_j} = e^{-q\Delta T_j}\,e^{-rT_j}$ on the right-hand side yields $e^{-q\Delta T_j}\,C(k_n^{T_j}, T_j)$.  Thus $h(0) \geq 0$ is exactly (iv).
    
    Combining the monotonicity $g(s) \geq g(0)$ (equivalently $h(s) \geq h(0)$) with the anchor $h(0) \geq 0$, we conclude $h(s) \geq 0$ for all $s \geq 0$, which is precisely the pointwise calendar inequality.
\end{proof}

\begin{remark}[Recovery of Lemma~\ref{lemma:powerlaw_calendar}]\label{rmk:mixed_tail_recovers}
    In the (PL$_2$, PL$_2$) sub-case, $\mu_j = k_n^{T_j}$ and $\mu_{j+1} = \gamma_j\,k_n^{T_j} = \gamma_j\,\mu_j$.  Then (iii) reduces to $(\alpha_j - 2)\,\gamma_j\,k_n^{T_j} \geq (\alpha_{j+1} - 2)\,\gamma_j\,k_n^{T_j}$, i.e.\ $\alpha_{j+1} \leq \alpha_j$, which coincides with (ii).  Hence (ii) and (iv) alone are sufficient in the all-two-anchor case, recovering Lemma~\ref{lemma:powerlaw_calendar}.
\end{remark}

\subsection{Proof of Lemma~\ref{lem:left_data_pos} and Proposition~\ref{prop:left_pow}: Power-Law Left Tail Properties}\label{sec:app_left_pow}

Throughout this part, we work on the interval $[0,k_1]$ with $k_1>0$, and the
input data $C(0)=e^{-qT}$, $C'(0)=-e^{-rT}$, $C(k_1)=C_{k_1}$, $C'(k_1)=C'_{k_1}$
are assumed butterfly-free (i.e., consistent with a convex call price curve).
Define
\[
\gamma_L := C_{k_1} - e^{-qT} + k_1\,e^{-rT}, \qquad
\delta_L := C'_{k_1} + e^{-rT},
\]
so that $\gamma_L$ is the excess of $C(k_1)$ over the tangent line
$\ell(k) := e^{-qT} - e^{-rT} k$ at $k=0$, and $\delta_L = C'(k_1) - C'(0)$
is the slope increment over $[0,k_1]$.

\begin{proof}[Proof of Lemma~\ref{lem:left_data_pos}]
    Since $C'$ is non-decreasing (by convexity), $\delta_L = C'(k_1) - C'(0) \geq 0$,
    and $\delta_L = \int_0^{k_1} C''(s)\,ds > 0$ when $C''$ is not identically
    zero on $(0,k_1)$.  By the Taylor remainder identity for a $C^2$ convex function,
    \begin{equation}\label{eq:left_taylor}
        \gamma_L
        = C(k_1) - e^{-qT} + k_1 e^{-rT}
        = C(k_1) - C(0) - C'(0)\,k_1
        = \int_0^{k_1}(k_1 - t)\,C''(t)\,dt,
    \end{equation}
    which is strictly positive under the same hypothesis.  Likewise,
    $\delta_L k_1 = k_1\!\int_0^{k_1} C''(t)\,dt$, so
    \[
    \delta_L k_1 - \gamma_L
    = \int_0^{k_1} t\,C''(t)\,dt > 0,
    \]
    and dividing by $\gamma_L$ gives $k_1\delta_L/\gamma_L > 1$.
\end{proof}

\begin{proof}[Proof of Proposition~\ref{prop:left_pow}]
    Posit $C''(k) = a_L\,k^{\beta_L}$ on $[0,k_1]$ with $\beta_L>-1$ and $a_L>0$;
    the constants will be uniquely determined below.  Integrating once and twice
    under the boundary conditions $C'(0)=-e^{-rT}$ and $C(0)=e^{-qT}$ yields
    \begin{align}
        C'(k) &= -e^{-rT} + \frac{a_L}{\beta_L+1}\,k^{\beta_L+1},
        \label{eq:left_C_prime}\\
        C(k)  &= e^{-qT} - e^{-rT}k + \frac{a_L}{(\beta_L+1)(\beta_L+2)}\,k^{\beta_L+2}.
        \label{eq:left_C}
    \end{align}
    Since $\beta_L+1>0$ and $\beta_L+2>1$, the higher-order terms vanish at
    $k=0$, so the left-boundary identities $C'(0)=-e^{-rT}$ and $C(0)=e^{-qT}$
    are satisfied automatically; this is the first half of Property~1 and
    also Property~6.
    
    \smallskip
    \paragraph{Property 1 (Matching at $k_1$ and uniqueness).}
    Imposing $C'(k_1)=C'_{k_1}$ in~\eqref{eq:left_C_prime} gives
    \begin{equation}\label{eq:left_match_slope}
        \frac{a_L}{\beta_L+1}\,k_1^{\beta_L+1} = \delta_L,
    \end{equation}
    and imposing $C(k_1)=C_{k_1}$ in~\eqref{eq:left_C} gives
    \begin{equation}\label{eq:left_match_price}
        \frac{a_L}{(\beta_L+1)(\beta_L+2)}\,k_1^{\beta_L+2} = \gamma_L.
    \end{equation}
    Dividing~\eqref{eq:left_match_price} by~\eqref{eq:left_match_slope} eliminates
    $a_L$,
    \[
    \frac{k_1}{\beta_L+2} = \frac{\gamma_L}{\delta_L}
    \quad\Longleftrightarrow\quad
    \beta_L = \frac{k_1\,\delta_L}{\gamma_L} - 2,
    \]
    which is uniquely determined by $(\gamma_L,\delta_L,k_1)$ and satisfies
    $\beta_L>-1$ by Lemma~\ref{lem:left_data_pos}.  Substituting back
    into~\eqref{eq:left_match_slope} gives
    \[
    a_L = \frac{(\beta_L+1)\,\delta_L}{k_1^{\beta_L+1}},
    \]
    also uniquely determined, which establishes Property~4.
    
    \smallskip
    \paragraph{Property 2 (Non-negativity of the density).}
    By Lemma~\ref{lem:left_data_pos}, $\delta_L > 0$ and $\beta_L+1>0$, hence
    $a_L > 0$.  Consequently
    \[
    f(k) = e^{rT}\,C''(k) = e^{rT}\,a_L\,k^{\beta_L} \geq 0
    \]
    on $[0,k_1]$.  When $\beta_L \geq 0$, $f$ extends continuously to $k=0$
    with $f(0)=0$ if $\beta_L>0$ and $f(0)=e^{rT}a_L$ if $\beta_L=0$.  When
    $\beta_L\in(-1,0)$, $f(k)\to\infty$ as $k\to 0^+$, but the singularity is
    integrable since $\beta_L+1>0$.
    
    \smallskip
    \paragraph{Property 3 (Left tail mass).}
    Using~\eqref{eq:left_match_slope},
    \[
    \int_0^{k_1}\! f(k)\,dk
    = e^{rT}\int_0^{k_1}\! a_L\,k^{\beta_L}\,dk
    = e^{rT}\,\frac{a_L}{\beta_L+1}\,k_1^{\beta_L+1}
    = e^{rT}\,\delta_L.
    \]
    By the Breeden--Litzenberger identity $\mathbb{Q}(S_T\leq k) = 1 + e^{rT}\,C'(k)$,
    this equals $\mathbb{Q}(S_T\leq k_1) = 1 + e^{rT}C'(k_1) = e^{rT}(C'(k_1)+e^{-rT})
    = e^{rT}\delta_L$, so the left-tail mass equals the cumulative
    distribution at $k_1$.
    
    \smallskip
    \paragraph{Property 5 (Finite first moment).}
    Since $\beta_L+2 > 1$,
    \[
    \int_0^{k_1}\! k\,f(k)\,dk
    = e^{rT}\,a_L\!\int_0^{k_1}\! k^{\beta_L+1}\,dk
    = e^{rT}\,a_L\,\frac{k_1^{\beta_L+2}}{\beta_L+2} < \infty.
    \]
    
    \smallskip
    \paragraph{Property 6 (Boundary slope coherence).}
    The matched values $C(0)=e^{-qT}$ and $C'(0)=-e^{-rT}$ established above
    constitute the left-side boundary conditions used in Section~\ref{sec:prelim}
    to certify unit total mass; combined with $C'(\infty)=0$ from
    Proposition~\ref{prop:powerlaw}, Property~6, the global density
    $\mathbf{1}_{[0,k_1]}f + \mathbf{1}_{(k_1,k_n)}f_{\mathrm{interior}} + \mathbf{1}_{[k_n,\infty)}f_{\mathrm{tail}}$
    integrates to one.
    
    \smallskip
    \paragraph{Compact form of $C$.}
    Substituting $a_L = (\beta_L+1)\delta_L/k_1^{\beta_L+1}$
    into~\eqref{eq:left_C}, the higher-order term simplifies to
    \[
    \frac{a_L\,k^{\beta_L+2}}{(\beta_L+1)(\beta_L+2)}
    = \frac{\delta_L\,k^{\beta_L+2}}{(\beta_L+2)\,k_1^{\beta_L+1}}
    = \frac{\delta_L\,k_1}{\beta_L+2}\left(\frac{k}{k_1}\right)^{\!\beta_L+2}
    = \gamma_L\!\left(\frac{k}{k_1}\right)^{\!\beta_L+2},
    \]
    where the last equality uses $\gamma_L = \delta_L k_1/(\beta_L+2)$ from the
    proof of Property~1.  Hence
    \[
    C(k) = e^{-qT} - e^{-rT}\,k + \gamma_L\!\left(\frac{k}{k_1}\right)^{\!\beta_L+2},
    \qquad k\in[0,k_1],
    \]
    which is the closed form~\eqref{eq:left_pow_compact}.
\end{proof}

\subsection{Proof of Lemma~\ref{lemma:lhs_pow_calendar}: Power-Law Left-Tail Calendar Condition}\label{sec:app_lhs_pow_calendar}

Under the power-law left tail at maturity~$T_j$, $C(k, T_j) = e^{-qT_j} - e^{-rT_j}k + \gamma_{L,j}(k/k_1^{T_j})^{\beta_j+2}$ on $[0, k_1^{T_j}]$.  With the forward factor $\phi_{j+1} = e^{-(r-q)(T_{j+1}-T_j)}$:
\begin{align*}
    H(k) &= \gamma_{L,j+1}\!\left(\frac{k}{k_1^{T_{j+1}}}\right)^{\!\beta_{j+1}+2} - e^{-q\Delta T_j}\,\gamma_{L,j}\!\left(\frac{\phi_{j+1} k}{k_1^{T_j}}\right)^{\!\beta_j+2}.
\end{align*}

Using $k_1^{T_j} = \phi_{j+1}\,k_1^{T_{j+1}}$, the second term simplifies to $e^{-q\Delta T_j}\gamma_{L,j}\phi_{j+1}^{\beta_j+2}(k/k_1^{T_{j+1}})^{\beta_j+2}/\phi_{j+1}^{\beta_j+2} = e^{-q\Delta T_j}\gamma_{L,j}(k/k_1^{T_{j+1}})^{\beta_j+2}$.  Setting $v = k/k_1^{T_{j+1}} \in [0,1]$:
$$
H = \gamma_{L,j+1}\,v^{\beta_{j+1}+2} - e^{-q\Delta T_j}\,\gamma_{L,j}\,v^{\beta_j+2}.
$$

\textbf{Case $\beta_{j+1} > \beta_j$.}  For small $v > 0$, $v^{\beta_{j+1}+2} \ll v^{\beta_j+2}$ (since $\beta_{j+1}+2 > \beta_j+2$), so $H(k) < 0$: calendar arbitrage exists.

\textbf{Case $\beta_{j+1} \leq \beta_j$.}  Write $H = v^{\beta_{j+1}+2}[\gamma_{L,j+1} - e^{-q\Delta T_j}\gamma_{L,j}\,v^{\beta_j - \beta_{j+1}}]$.  For $v \in (0,1]$ and $\beta_j - \beta_{j+1} \geq 0$, the term $v^{\beta_j-\beta_{j+1}} \leq 1$, so $H \geq v^{\beta_{j+1}+2}[\gamma_{L,j+1} - e^{-q\Delta T_j}\gamma_{L,j}]$.  Hence $H \geq 0$ provided $\gamma_{L,j+1} \geq e^{-q\Delta T_j}\gamma_{L,j}$, which is precisely the calendar condition at the anchor point~$k_1^{T_{j+1}}$ (ensured by the interior construction).

Therefore calendar freedom requires $\beta_j$ non-increasing, together with the anchor-point calendar condition. \qed


\subsection{Proof of Lemma~\ref{lemma:cross-maturity}: Cross-Maturity Calendar Slack}\label{sec:app_cross_maturity}

Fix $j \ge 2$ and an observed strike $k$ of $T_j$ satisfying the
forward-alignment hypothesis of~\eqref{eq:S-input}, so that the
forward-aligned strike $\phi_j k$ lies in the interior body of
$T_{j-1}$.  Let $[k_a, k_b]$ be the observed-strike interval of
$T_{j-1}$ containing $\phi_j k$, and write the convex-combination
weight
\[
t := \frac{\phi_j k - k_a}{k_b - k_a} \in [0, 1],
\]
so that $\phi_j k = (1 - t)\,k_a + t\,k_b$.  Convexity of
$C(\cdot, T_{j-1})$ on $[k_a, k_b]$ together with the price-matching
values $C(k_a, T_{j-1}) = C^{T_{j-1}}_{k_a}$ and
$C(k_b, T_{j-1}) = C^{T_{j-1}}_{k_b}$ gives, by Jensen's inequality,
\[
C\!\bigl(\phi_j k,\, T_{j-1}\bigr)
\le (1-t)\,C^{T_{j-1}}_{k_a} + t\,C^{T_{j-1}}_{k_b}
= I_C^{T_{j-1}}\!\bigl(\phi_j k\bigr),
\]
where $I_C^{T_{j-1}}$ denotes the chord interpolation of the input
prices of $T_{j-1}$ on $[k_a, k_b]$.  Multiplying by the positive
constant $e^{-q\,\Delta T_j}$ and applying
condition~\eqref{eq:S-input} at the observed strike $k$ of $T_j$,
\[
L(k) = e^{-q\,\Delta T_j}\,C\!\bigl(\phi_j k, T_{j-1}\bigr)
\le e^{-q\,\Delta T_j}\,I_C^{T_{j-1}}\!\bigl(\phi_j k\bigr)
< C^{T_j}_{k}.
\]
Hence $C^{T_j}_k - L(k) > 0$ at every observed strike of $T_j$
satisfying the forward-alignment hypothesis.  Applying this to the
two endpoints of an interior interval $[k_l, k_r]$ of $T_j$ whose
endpoints both satisfy the forward-alignment hypothesis yields
$\delta_l > 0$ and $\delta_r > 0$, so
$\delta_{\min} = \min(\delta_l, \delta_r) > 0$.
\qed


\subsection{Proof of Lemma~\ref{lemma:three-piece-feasibility}: An Explicit Feasibility Construction}\label{sec:app_three_piece_feasibility}

We establish feasibility by explicitly constructing a feasible solution of the per-interval optimization~\eqref{eq:three-piece-opt}.

In this lemma we show that, for every valid per-interval input, one can always find a three-piece solution that simultaneously matches the endpoint prices and slopes and is both butterfly- and calendar-arbitrage-free.  The construction below provides one such solution and thereby demonstrates nonemptiness of the feasible set.

The aim is to show that the admissible set is non-empty, that is, that the constraints can be met in every case.  The solution returned in practice need not take the closed form constructed here: a strictly convex smoothness objective such as $\mathcal{F}(h_1, h_{\rm mid}, h_3, \beta_1, \beta_2) = (h_1 - h_{\rm mid})^2 + (h_{\rm mid} - h_3)^2$, which penalizes the middle-piece curvature approaching zero, yields a better-behaved minimizer that typically lies in the interior of the feasible set with all curvatures positive.

Define the parameters of the explicit feasible solution:
\begin{equation}\label{eq:corner-point}
    \begin{aligned}
        h_1^\star &:= h_1^{\rm cal}, \\
        \beta_1^\star &:= \frac{\rho\,A_l}{h_1^{\rm cal}\,\Delta k_l}, \\
        \beta_2^\star &:= \frac{\rho^{\,2}\,A_l}{(1 - \rho)\,h_1^{\rm cal}\,\Delta k_l}, \\
        h_{\rm mid}^\star &:= 0, \\
        h_3^\star &:= \frac{(1 - \rho)^2}{\rho^2}\,h_1^{\rm cal}.
    \end{aligned}
\end{equation}
We verify that these satisfy every constraint of~\eqref{eq:three-piece-opt}.

\textbf{Slope equality~\eqref{eq:three-slope}.}  The middle piece contributes zero.  The first-piece term equals
\[
h_1^\star\,\beta_1^\star\,\Delta k_l
= h_1^{\rm cal} \cdot \frac{\rho\,A_l}{h_1^{\rm cal}\,\Delta k_l} \cdot \Delta k_l
= \rho\,A_l.
\]
The third-piece term equals
\[
h_3^\star\,\beta_2^\star\,\Delta k_l
= \frac{(1-\rho)^2}{\rho^2}\,h_1^{\rm cal}
\cdot \frac{\rho^{\,2}\,A_l}{(1-\rho)\,h_1^{\rm cal}\,\Delta k_l}
\cdot \Delta k_l
= (1-\rho)\,A_l.
\]
The total is $\rho A_l + (1-\rho) A_l = A_l$, matching the right-hand side of~\eqref{eq:three-slope}.

\textbf{Price equality~\eqref{eq:three-h}.}  The middle piece contributes zero.  The first-piece term evaluates to
\[
h_1^\star\,\beta_1^\star\,\Bigl(1 - \tfrac{\beta_1^\star}{2}\Bigr)\,(\Delta k_l)^2
= \rho\,A_l\,\Delta k_l - \frac{\rho^{\,2}\,A_l^{\,2}}{2\,h_1^{\rm cal}}
\]
after substituting $\beta_1^\star = \rho A_l/(h_1^{\rm cal}\,\Delta k_l)$ and simplifying.  The third-piece term evaluates to
\[
h_3^\star\,\frac{(\beta_2^\star)^{\,2}}{2}\,(\Delta k_l)^2
= \frac{\rho^{\,2}\,A_l^{\,2}}{2\,h_1^{\rm cal}}
\]
by direct algebra.  Their sum is $\rho A_l \Delta k_l$.  Recalling $\rho := P_l/(A_l \Delta k_l)$, this equals $P_l$, matching the right-hand side of~\eqref{eq:three-h}.

\textbf{Butterfly freedom and calendar floor.}  Under the standing hypotheses ($A_l > 0$, $\rho \in (0,1)$, $\delta_{\min} > 0$), both arguments of the max in~\eqref{eq:h1-cal} are strictly positive, so $h_1^\star = h_1^{\rm cal} > 0$ and $h_3^\star = \tfrac{(1-\rho)^2}{\rho^2}\,h_1^{\rm cal} > 0$ since $\rho \in (0,1)$; the middle-piece curvature equals zero and trivially satisfies butterfly non-negativity.  The calendar floor $h_1^\star \ge h_1^{\rm cal}$ holds with equality.

\textbf{Width bounds.}  Both $\beta_1^\star, \beta_2^\star > 0$ by construction.  Their sum is
\[
\beta_1^\star + \beta_2^\star
= \frac{\rho\,A_l}{h_1^{\rm cal}\,\Delta k_l}
+ \frac{\rho^{\,2}\,A_l}{(1-\rho)\,h_1^{\rm cal}\,\Delta k_l}
= \frac{\rho\,A_l}{(1-\rho)\,h_1^{\rm cal}\,\Delta k_l}.
\]
Since $h_1^{\rm cal} \ge \rho A_l/((1-\rho)\,\Delta k_l)$ (the first argument of the max in~\eqref{eq:h1-cal}), this sum satisfies $\beta_1^\star + \beta_2^\star \le 1$, with strict inequality whenever the second argument $\rho^2 A_l^2/(2\,\delta_{\min})$ strictly dominates and equality (degenerate middle piece, $\beta_1^\star + \beta_2^\star = 1$) when both arguments of the max are equal.  In particular, $\beta_1^\star, \beta_2^\star \in (0, 1)$.

This verifies that~\eqref{eq:corner-point} satisfies all constraints of the per-interval optimization~\eqref{eq:three-piece-opt}, so the feasible set is non-empty.

The interval curve~\eqref{eq:corner-point} is exactly the one constructed in the proof of Proposition~\ref{prop:calendar-reduction}\textup{(b)}, evaluated at $h_1 = h_1^{\rm cal}$ (compare~\eqref{eq:corner-point} with~\eqref{eq:dip-params}), so by Proposition~\ref{prop:calendar-reduction}\textup{(c)} it is pointwise calendar-arbitrage-free, $C^\star(k) \geq L(k)$ on $[k_l, k_r]$.  Hence, for every per-interval input $(C_{k_l}, C_{k_r}, C'_{k_l}, C'_{k_r})$ producing $\rho \in (0,1)$, $A_l > 0$, $\delta_{\min} > 0$, and $\Delta k_l > 0$, the interval curve $C^\star$ is a feasible solution of~\eqref{eq:three-piece-opt} that is butterfly-free, matches the endpoint slopes and prices, respects the calendar floor and the width bounds, and is pointwise calendar-arbitrage-free.
\qed

\subsection{Proof of Proposition~\ref{prop:calendar-reduction}: Calendar Reduction}\label{sec:app_calendar_reduction}

\textbf{(a) Chord domination.}
Set $t := (k - k_l)/\Delta k_l \in [0,1]$. By convexity of $L$,
\[ L(k) = L\bigl((1-t)\,k_l + t\,k_r\bigr) \leq (1-t)\,L(k_l) + t\,L(k_r). \]
The chord interpolates linearly: $I_C(k) = (1-t)\,C_{k_l} + t\,C_{k_r}$. Subtracting the two,
\[ I_C(k) - L(k) \geq (1-t)\,\delta_l + t\,\delta_r \geq \delta_{\min}. \]
The argument uses only the convexity of $L$ and is independent of the construction class on $[k_l, k_r]$.

\textbf{(b) Gap bound under the geometric lower bound.}
Fix any first-piece curvature $h_1$ satisfying the width-feasibility bound~\eqref{eq:h1-geom}, and construct one explicit three-piece interval curve $C^\star$ as follows.  Write $u = k - k_l$ and $D(u) := I_C(k_l + u) - C^\star(k_l + u)$ for the gap below the chord.  Since $I_C$ has constant slope $(C_{k_r} - C_{k_l})/\Delta k_l = C'_{k_l} + \rho A_l$ (using $C_{k_r} - C_{k_l} = C'_{k_l}\,\Delta k_l + P_l$ and $P_l = \rho A_l\,\Delta k_l$),
\[
D'(u) = \rho A_l - \int_0^u (C^\star)''(k_l + s)\,ds,
\qquad D''(u) = -(C^\star)''(k_l + u) \le 0,
\]
so $D$ is concave with $D(0) = 0$ and $D'(0) = \rho A_l > 0$, and its maximum is attained where $\int_0^u (C^\star)'' = \rho A_l$.

Set $h_{\mathrm{mid}} = 0$, the smallest curvature allowed by butterfly freedom, so that $D' = 0$ on the middle piece and $D$ is constant there.  With $h_1$ given and $h_{\mathrm{mid}} = 0$ fixed, slope and price matching are two equations leaving one degree of freedom, which we use to place the right endpoint of piece~1, $a_1 = \beta_{l_1}\,\Delta k_l$, at $u^\star = \rho A_l/h_1$, so that $\beta_{l_1} = \rho A_l/(h_1\,\Delta k_l)$.  Slope matching then forces $h_3\,\beta_{l_2}\,\Delta k_l = (1-\rho)\,A_l$, and price matching, through the identity $P_l = \int_0^{\Delta k_l} (\Delta k_l - u)\,(C^\star)''(k_l + u)\,du$, forces $h_3\,\beta_{l_2}^2\,(\Delta k_l)^2 = \rho^2 A_l^2/h_1$.  Solving,
\begin{equation}\label{eq:dip-params}
    \beta_{l_1} = \frac{\rho\,A_l}{h_1\,\Delta k_l}, \qquad
    h_3 = \frac{(1-\rho)^2}{\rho^2}\,h_1, \qquad
    \beta_{l_2} = \frac{\rho^2\,A_l}{(1-\rho)\,h_1\,\Delta k_l},
\end{equation}
all positive, so butterfly freedom holds.  The outer widths satisfy $\beta_{l_1} + \beta_{l_2} = \rho A_l/((1-\rho)\,h_1\,\Delta k_l) \le 1$ exactly when $h_1 \ge \rho A_l/((1-\rho)\,\Delta k_l)$, the width-feasibility bound~\eqref{eq:h1-geom}; the middle width $1 - \beta_{l_1} - \beta_{l_2}$ is then nonnegative, vanishing only at equality, where the interval curve reduces to two pieces.

It remains to evaluate $D$, with $a_1 = \beta_{l_1}\,\Delta k_l$ and $a_2 = (1 - \beta_{l_2})\,\Delta k_l$.  On $[0, a_1]$, integrating $(C^\star)'' = h_1$ gives $D(u) = \rho A_l\,u - h_1 u^2/2$, with maximum $\rho^2 A_l^2/(2 h_1)$ at $u^\star = \rho A_l/h_1 = a_1$.  On $[a_1, a_2]$, $(C^\star)'' = 0$ gives $D' = 0$, so $D = \rho^2 A_l^2/(2 h_1)$.  On $[a_2, \Delta k_l]$, $(C^\star)'' = h_3$ gives $D(\Delta k_l) = 0$ by the price relation.  Hence
\[
\sup_{k \in [k_l,k_r]} \bigl[I_C(k) - C^\star(k)\bigr] = \frac{\rho^2\,A_l^2}{2\,h_1},
\qquad \rho := \frac{P_l}{A_l\,\Delta k_l},
\]
which proves part~\textup{(b)}.

\textbf{(c) Combination.}
For the interval curve $C^\star$ of part~(b), part~(a) gives
\[
C^\star(k) - L(k) = \bigl[I_C(k) - L(k)\bigr] - \bigl[I_C(k) - C^\star(k)\bigr]
\geq \delta_{\min} - \frac{\rho^2 A_l^2}{2\,h_1},
\]
which is non-negative iff $h_1 \geq \rho^2 A_l^2/(2\,\delta_{\min})$, the slack bound~\eqref{eq:h1-slack}.  Together with the width-feasibility bound~\eqref{eq:h1-geom} under which (b) was derived, this gives the calendar floor $h_1 \geq h_1^{\rm cal}$ of~\eqref{eq:h1-cal}.  In the three-piece construction of Section~\ref{sec:three-piece}, $h_1$ is the free local first-piece curvature, enforced as a single-interval bound constraint in the per-interval optimization~\eqref{eq:three-piece-opt}; under the floor $h_1 \geq h_1^{\rm cal}$, the interval curve $C^\star$ is calendar-arbitrage-free on $[k_l, k_r]$.  \qed

\subsection{Proof of Theorem~\ref{thm:unified}: Unified Interior Feasibility}\label{sec:app_unified}

    We proceed by single induction on the maturity index $j = 1, \ldots, m$.  Throughout, write $\Delta T_j := T_j - T_{j-1}$ and define the calendar lower bound $L(k) := e^{-q\Delta T_j}\,C(e^{-(r-q)\Delta T_j}k,\,T_{j-1})$, which is convex.
    
    \medskip
    \noindent\textbf{Outer base case ($j = 1$).}
    No prior maturity exists, so calendar freedom is vacuous.  Butterfly freedom on each interior interval $[k_i^{T_1}, k_{i+1}^{T_1}]$ is delivered by the two-piece body construction of Theorem~\ref{thm:two-piece}, and on the right tail $(k_{n_1}^{T_1}, \infty)$ by the power-law construction of Proposition~\ref{prop:powerlaw}.  The left region $[0, k_1^{T_1}]$ is handled either by the two-piece body or by the power-law left tail of Proposition~\ref{prop:left_pow}.  All pieces match prices and slopes exactly at observed strikes by construction.

    \medskip
    \noindent\textbf{Outer inductive step ($j \ge 2$).}
    Assume the construction at $T_1, \ldots, T_{j-1}$ is complete, butterfly-free, price-matching at every observed strike, and calendar-free.  We construct the surface at $T_j$ on three regions: the left region $[0, k_1^{T_j}]$, the interior intervals $[k_l^{T_j}, k_{l+1}^{T_j}]$ for $l = 1, \ldots, n_j - 1$, and the right tail $(k_{n_j}^{T_j}, \infty)$.

    Left region $[0, k_1^{T_j}]$.  Mirroring the right-tail treatment below, this region is constructed as a power-law left tail (Proposition~\ref{prop:left_pow}), with calendar freedom from Lemma~\ref{lemma:lhs_pow_calendar} and the boundary slope fixed by the left-tail program of Section~\ref{sec:global_slope}; it requires only that the left segment be non-degenerate, the hypothesis of Proposition~\ref{prop:left_pow} that $C''$ does not vanish identically on $(0, k_1)$.  When that hypothesis fails, that is, when the left segment is exactly linear, the left region is that linear segment itself, which Lemma~\ref{lemma:base-case} (Appendix~\ref{app:degenerate_left_tail}) shows to be butterfly-free and calendar-free, with $C(\cdot,T_j) = L$ on $[0,k_1]$.

    Interior intervals $[k_l^{T_j}, k_{l+1}^{T_j}]$ for $l = 1, \ldots, n_j - 1$.  By the input calendar condition~\eqref{eq:S-input} and Lemma~\ref{lemma:cross-maturity}, the cross-maturity slack
    $\delta_{\min, l} := \min(C_{k_l} - L(k_l),\, C_{k_{l+1}} - L(k_{l+1}))$
    is strictly positive at every observed strike of $T_j$ whose forward-aligned strike falls inside the interior body of $T_{j-1}$.  Combined with $\Delta k_l > 0$, $A_l > 0$, and $\rho_l := P_l/(A_l\,\Delta k_l) \in (0, 1)$ (the standing geometric hypotheses on the input data), Lemma~\ref{lemma:three-piece-feasibility} guarantees that the per-interval optimization~\eqref{eq:three-piece-opt} has a non-empty feasible set; the appendix proof~(Section~\ref{sec:app_three_piece_feasibility}) constructs an explicit feasible solution~\eqref{eq:corner-point} satisfying butterfly non-negativity, slope and price matching, the calendar floor $h_1^\star \ge h_1^{\rm cal}$, and the width bounds $\beta_1^\star + \beta_2^\star \le 1$.  This feasible solution is exactly the interval curve $C^\star$ at $h_1 = h_1^{\rm cal}$, so by Proposition~\ref{prop:calendar-reduction}(c) it is pointwise calendar-arbitrage-free, $C(k, T_j) \ge L(k)$ on $[k_l^{T_j}, k_{l+1}^{T_j}]$.  This argument applies independently to each interior interval $l$: the per-interval constructions decouple completely (no inheritance, no forward-looking constraint), so the inner loop over $l$ can run in parallel.

    Right tail $(k_{n_j}^{T_j}, \infty)$.  The power-law construction of Proposition~\ref{prop:powerlaw}, with the tail index $\alpha_j$ selected by the global slope-coherence quadratic program of Section~\ref{sec:global_slope} subject to $\alpha_j \le \alpha_{j-1}$, yields butterfly freedom and pointwise calendar freedom on $(k_{n_j}^{T_j}, \infty)$ by Lemma~\ref{lemma:powerlaw_calendar}.
    
    \medskip
    By induction on $j$, the construction extends to all maturities $T_1, \ldots, T_m$, and the resulting call surface is butterfly-free and calendar-free everywhere.

\subsection{Proof of Lemma~\ref{lemma:base-case}: Degenerate left segment under calendar alignment}\label{app:degenerate_left_tail}

The left region $[0,k_1]$ is built as the power-law left tail of Proposition~\ref{prop:left_pow} whenever the left segment is non-degenerate.  The only remaining case is the degenerate one, where the left segment has no curvature to fit and the power-law tail is replaced by a linear segment; we show this segment is calendar-arbitrage-free with no further construction.

\begin{proof}
Since $C'' \equiv 0$, the curve $C(\cdot,T_j)$ on $[0,k_1]$ is the straight line determined by the boundary data $C(0,T_j) = e^{-qT_j}$ and $C'(0,T_j) = -e^{-rT_j}$, namely $C(k,T_j) = e^{-qT_j} - e^{-rT_j}k$.  Butterfly freedom is immediate, since $C'' = 0 \ge 0$, and the endpoint price and slope are matched by construction.

It remains to check the calendar floor $L$.  Degeneracy gives $\gamma_{L,j} = 0$, so the anchor-point calendar condition $\gamma_{L,j} \ge e^{-q\Delta T_j}\gamma_{L,j-1}$ (the condition of Lemma~\ref{lemma:lhs_pow_calendar} for the pair $(T_j, T_{j-1})$, ensured at the observed strike $k_1$ by the interior construction) reads $0 \ge e^{-q\Delta T_j}\gamma_{L,j-1}$.  Since $\gamma_{L,j-1} \ge 0$ (a butterfly-free input is convex, so its call price lies above the intrinsic-value line at the anchor strike), this forces $\gamma_{L,j-1} = 0$.  A zero excess means the convex price also meets that line at $k_1$; lying above its $k = 0$ tangent and below its chord on $[0,k_1]$, both of which here coincide with the line, it equals the line throughout, so the previous maturity's left segment is linear and $L$ is itself linear on $[0,k_1]$.  Both $C(\cdot,T_j)$ and $L$ are straight lines, and they agree in value and slope at $k = 0$, where the boundary identities give $g(0) = g'(0) = 0$ for $g := C(\cdot,T_j) - L$.  A linear function with zero value and slope at a point vanishes identically, so $g \equiv 0$, that is $C(k,T_j) = L(k)$ on $[0,k_1]$, which satisfies the calendar constraint $C(\cdot,T_j) \ge L$.
\end{proof}

\bibliographystyle{unsrtnat}
\bibliography{references}

\section*{Disclaimer}
This paper was prepared for informational purposes with contributions from the Quantitative Trading \& Research team of JPMorgan Chase \& Co. This paper is not a product of the Research Department of JPMorgan Chase \& Co. or its affiliates. Neither JPMorgan Chase \& Co. nor any of its affiliates makes any explicit or implied representation or warranty and none of them accept any liability in connection with this paper, including, without limitation, with respect to the completeness, accuracy, or reliability of the information contained herein and the potential legal, compliance, tax, or accounting effects thereof. This document is not intended as investment research or investment advice, or as a recommendation, offer, or solicitation for the purchase or sale of any security, financial instrument, financial product or service, or to be used in any way for evaluating the merits of participating in any transaction.

\end{document}